\newif\ifextended\extendedtrue
\newif\ifcomments\commentsfalse
\newif\ifaec\aecfalse

\RequirePackage[usenames,dvipsnames]{xcolor}

\documentclass[acmsmall]{acmart}
\makeatletter
\@ifundefined{lhs2tex.lhs2tex.sty.read}%
  {\@namedef{lhs2tex.lhs2tex.sty.read}{}%
   \newcommand\SkipToFmtEnd{}%
   \newcommand\EndFmtInput{}%
   \long\def\SkipToFmtEnd#1\EndFmtInput{}%
  }\SkipToFmtEnd

\newcommand\ReadOnlyOnce[1]{\@ifundefined{#1}{\@namedef{#1}{}}\SkipToFmtEnd}
\usepackage{amstext}
\usepackage{amssymb}
\usepackage{stmaryrd}
\DeclareFontFamily{OT1}{cmtex}{}
\DeclareFontShape{OT1}{cmtex}{m}{n}
  {<5><6><7><8>cmtex8
   <9>cmtex9
   <10><10.95><12><14.4><17.28><20.74><24.88>cmtex10}{}
\DeclareFontShape{OT1}{cmtex}{m}{it}
  {<-> ssub * cmtt/m/it}{}

\DeclareFontShape{OT1}{cmtt}{bx}{n}
  {<5><6><7><8>cmtt8
   <9>cmbtt9
   <10><10.95><12><14.4><17.28><20.74><24.88>cmbtt10}{}
\DeclareFontShape{OT1}{cmtex}{bx}{n}
  {<-> ssub * cmtt/bx/n}{}

\newcommand{\anonymous}{\kern0.06em \vbox{\hrule\@width.5em}}

\usepackage{polytable}

\@ifundefined{mathindent}%
  {\newdimen\mathindent\mathindent\leftmargini}%
  {}%

\def\resethooks{%
  \global\let\SaveRestoreHook\empty
  \global\let\ColumnHook\empty}
\newcommand*{\savecolumns}[1][default]%
  {\g@addto@macro\SaveRestoreHook{\savecolumns[#1]}}
\newcommand*{\restorecolumns}[1][default]%
  {\g@addto@macro\SaveRestoreHook{\restorecolumns[#1]}}
\newcommand*{\aligncolumn}[2]%
  {\g@addto@macro\ColumnHook{\column{#1}{#2}}}

\resethooks

\newcommand{\onelinecommentchars}{\quad-{}- }
\newcommand{\commentbeginchars}{\enskip\{-}
\newcommand{\commentendchars}{-\}\enskip}

\newcommand{\visiblecomments}{%
  \let\onelinecomment=\onelinecommentchars
  \let\commentbegin=\commentbeginchars
  \let\commentend=\commentendchars}

\newcommand{\invisiblecomments}{%
  \let\onelinecomment=\empty
  \let\commentbegin=\empty
  \let\commentend=\empty}

\visiblecomments

\newlength{\blanklineskip}
\newcommand{\hsindent}[1]{\quad}
\let\hspre\empty
\let\hspost\empty

\EndFmtInput
\makeatother
\ReadOnlyOnce{polycode.fmt}%
\makeatletter

\newcommand{\hsnewpar}[1]%
  {{\parskip=0pt\parindent=0pt\par\vskip #1\noindent}}

\newcommand{\hscodestyle}{}

\newcommand{\sethscode}[1]%
  {\expandafter\let\expandafter\hscode\csname #1\endcsname
   \expandafter\let\expandafter\endhscode\csname end#1\endcsname}

  {\par\noindent
   \advance\leftskip\mathindent
   \hscodestyle
   \let\\=\@normalcr
   \let\hspre\(\let\hspost\)%
   \pboxed}%
  {\endpboxed\)%
   \par\noindent
   \ignorespacesafterend}

  {\hsnewpar\abovedisplayskip
   \advance\leftskip\mathindent
   \hscodestyle
   \let\hspre\(\let\hspost\)%
   \pboxed}%
  {\endpboxed%
   \hsnewpar\belowdisplayskip
   \ignorespacesafterend}

  {\hsnewpar\abovedisplayskip
   \advance\leftskip\mathindent
   \hscodestyle
   \let\\=\@normalcr
   \(\pboxed}%
  {\endpboxed\)%
   \hsnewpar\belowdisplayskip
   \ignorespacesafterend}

\newcommand{\plainhs}{\sethscode{plainhscode}}

\plainhs

  {\hsnewpar\abovedisplayskip
   \advance\leftskip\mathindent
   \hscodestyle
   \let\\=\@normalcr
   \(\parray}%
  {\endparray\)%
   \hsnewpar\belowdisplayskip
   \ignorespacesafterend}

  {\parray}{\endparray}

  {\(\parray}{\endparray\)}

\def\codeframewidth{\arrayrulewidth}
\RequirePackage{calc}

  {\parskip=\abovedisplayskip\par\noindent
   \hscodestyle
   \arrayrulewidth=\codeframewidth
   \tabular{@{}|p{\linewidth-2\arraycolsep-2\arrayrulewidth-2pt}|@{}}%
   \hline\framedhslinecorrect\\{-1.5ex}%
   \let\endoflinesave=\\
   \let\\=\@normalcr
   \(\pboxed}%
  {\endpboxed\)%
   \framedhslinecorrect\endoflinesave{.5ex}\hline
   \endtabular
   \parskip=\belowdisplayskip\par\noindent
   \ignorespacesafterend}

\newcommand{\framedhslinecorrect}[2]%
  {#1[#2]}

  {\(\def\column##1##2{}%
   \let\>\undefined\let\<\undefined\let\\\undefined
   \newcommand\>[1][]{}\newcommand\<[1][]{}\newcommand\\[1][]{}%
   \def\fromto##1##2##3{##3}%
   }{\) }%

  {\let\orighscode=\hscode
   \let\origendhscode=\endhscode
   \def\endhscode{\def\hscode{\endgroup\def\@currenvir{hscode}\\}\begingroup}
   \orighscode\def\hscode{\endgroup\def\@currenvir{hscode}}}%
  {\origendhscode
   \global\let\hscode=\orighscode
   \global\let\endhscode=\origendhscode}%

\makeatother
\EndFmtInput
\ReadOnlyOnce{forall.fmt}%
\makeatletter

\let\HaskellResetHook\empty
\newcommand*{\AtHaskellReset}[1]{%
  \g@addto@macro\HaskellResetHook{#1}}
\newcommand*{\HaskellReset}{\HaskellResetHook}

\newcommand\hsforall{\global\let\hsdot=\hsperiodonce}
\newcommand*\hsperiodonce[2]{#2\global\let\hsdot=\hscompose}
\newcommand*\hscompose[2]{#1}

\AtHaskellReset{\global\let\hsdot=\hscompose}

\HaskellReset

\makeatother
\EndFmtInput

\ifextended
\newcommand{\auxiliarymatter}{the appendix}
\else
\newcommand{\auxiliarymatter}{our evaluated proof artifact}
\fi

\ifcomments
\newcommand{\JGM}[1]{\commentbox{OliveGreen}{JGM}{#1}}
\newcommand{\RAE}[1]{\commentbox{RawSienna}{RAE}{#1}}
\else
\newcommand{\JGM}[1]{}
\newcommand{\RAE}[1]{}
\fi

\newcommand{\ottdrule}[4][]{{\displaystyle\frac{\begin{array}{l}#2\end{array}}{#3}\quad\ottdrulename{#4}}}
\newcommand{\ottusedrule}[1]{\[#1\]}
\newcommand{\ottpremise}[1]{ #1 \\}
\newenvironment{ottdefnblock}[3][]{ \framebox{\mbox{#2}} \quad #3 \\[0pt]}{}

\newcommand{\ottnt}[1]{\mathit{#1}}
\newcommand{\ottmv}[1]{\mathit{#1}}
\newcommand{\ottkw}[1]{\mathbf{#1}}
\newcommand{\ottsym}[1]{#1}
\newcommand{\ottcom}[1]{\text{#1}}
\newcommand{\ottdrulename}[1]{\textsc{#1}}
\newcommand{\ottcomplu}[5]{\overline{#1}^{\,#2\in #3 #4 #5}}
\newcommand{\ottcompu}[3]{\overline{#1}^{\,#2<#3}}
\newcommand{\ottcomp}[2]{\overline{#1}^{\,#2}}

\newcommand{\ottdruleSTXXTyCon}[1]{\ottdrule[#1]{%
\ottpremise{\ottcompu{ P  \mid  \Gamma   \vdash   \tau_{\ottmv{i}}  \ \mathsf{type} }{\ottmv{i}}{\ottmv{n}}}%
\ottpremise{ (  \ottnt{H}  :  \ottmv{n}  ) \in  \Sigma  \quad \hspace{-.6em} \quad  \vdash   P  \mid  \Gamma  \ \mathsf{ctx} }%
}{
 P  \mid  \Gamma   \vdash   \ottnt{H} \, \overline{\tau}  \ \mathsf{type} }{%
{\ottdrulename{ST\_TyCon}}{}%
}}

\newcommand{\ottdruleSTXXVar}[1]{\ottdrule[#1]{%
\ottpremise{\alpha  \in  \Gamma \quad \hspace{-.6em} \quad  \vdash   P  \mid  \Gamma  \ \mathsf{ctx} }%
}{
 P  \mid  \Gamma   \vdash   \alpha  \ \mathsf{type} }{%
{\ottdrulename{ST\_Var}}{}%
}}

\newcommand{\ottdruleSTXXForall}[1]{\ottdrule[#1]{%
\ottpremise{ P  \mid  \Gamma  \ottsym{,}  \alpha   \vdash   \tau  \ \mathsf{type} }%
}{
 P  \mid  \Gamma   \vdash   \forall \, \alpha  \ottsym{.}  \tau  \ \mathsf{type} }{%
{\ottdrulename{ST\_Forall}}{}%
}}

\newcommand{\ottdruleSTXXQual}[1]{\ottdrule[#1]{%
\ottpremise{ P  \ottsym{,}  \pi  \mid  \Gamma   \vdash   \tau  \ \mathsf{type} }%
}{
 P  \mid  \Gamma   \vdash   \pi  \Rightarrow  \tau  \ \mathsf{type} }{%
{\ottdrulename{ST\_Qual}}{}%
}}

\newcommand{\ottdruleSTXXFamily}[1]{\ottdrule[#1]{%
\ottpremise{ ( C   \Rightarrow   \ottmv{F}  :  \ottmv{n} ) \in  \Sigma  \quad \hspace{-.6em} \quad  \vdash   P  \mid  \Gamma  \ \mathsf{ctx} }%
\ottpremise{\ottcompu{ P  \mid  \Gamma   \vdash   \tau_{\ottmv{i}}  \ \mathsf{type} }{\ottmv{i}}{\ottmv{n}} \quad \hspace{-.6em} \quad  P  \Vdash  C \, \overline{\tau} }%
}{
 P  \mid  \Gamma   \vdash   \ottmv{F} \, \overline{\tau}  \ \mathsf{type} }{%
{\ottdrulename{ST\_Family}}{}%
}}

\newcommand{\ottdruleSTXXTFamily}[1]{\ottdrule[#1]{%
\ottpremise{ (  \top    \Rightarrow   \ottmv{F}  :  \ottmv{n} ) \in  \Sigma  \quad \hspace{-.6em} \quad  \vdash   P  \mid  \Gamma  \ \mathsf{ctx}  \quad \hspace{-.6em} \quad \ottcompu{ P  \mid  \Gamma   \vdash   \tau_{\ottmv{i}}  \ \mathsf{type} }{\ottmv{i}}{\ottmv{n}}}%
}{
 P  \mid  \Gamma   \vdash   \ottmv{F} \, \overline{\tau}  \ \mathsf{type} }{%
{\ottdrulename{ST\_TFamily}}{}%
}}

\newcommand{\ottdruleTXXTyCon}[1]{\ottdrule[#1]{%
\ottpremise{\ottnt{H}  \ottsym{:}  \ottmv{n} \quad \hspace{-.6em} \quad  \vdash   \Gamma  \ \mathsf{ctx} }%
\ottpremise{\ottcompu{ \Gamma   \vdash   \tau_{\ottmv{i}}  \ \mathsf{type} }{\ottmv{i}}{\ottmv{n}}}%
}{
 \Gamma   \vdash   \ottnt{H} \, \overline{\tau}  \ \mathsf{type} }{%
{\ottdrulename{T\_TyCon}}{}%
}}

\newcommand{\ottdruleTXXArrow}[1]{\ottdrule[#1]{%
\ottpremise{ \Gamma   \vdash   \tau_{{\mathrm{1}}}  \ \mathsf{type} }%
\ottpremise{ \Gamma   \vdash   \tau_{{\mathrm{2}}}  \ \mathsf{type} }%
}{
 \Gamma   \vdash   \tau_{{\mathrm{1}}}  \to  \tau_{{\mathrm{2}}}  \ \mathsf{type} }{%
{\ottdrulename{T\_Arrow}}{}%
}}

\newcommand{\ottdruleTXXVar}[1]{\ottdrule[#1]{%
\ottpremise{\alpha  \in  \Gamma}%
\ottpremise{ \vdash   \Gamma  \ \mathsf{ctx} }%
}{
 \Gamma   \vdash   \alpha  \ \mathsf{type} }{%
{\ottdrulename{T\_Var}}{}%
}}

\newcommand{\ottdruleTXXForall}[1]{\ottdrule[#1]{%
\ottpremise{ \Gamma  \ottsym{,}  \alpha   \vdash   \tau  \ \mathsf{type} }%
}{
 \Gamma   \vdash   \forall \, \alpha  \ottsym{.}  \tau  \ \mathsf{type} }{%
{\ottdrulename{T\_Forall}}{}%
}}

\newcommand{\ottdruleTXXQual}[1]{\ottdrule[#1]{%
\ottpremise{ \Gamma   \vdash   \phi  \ \mathsf{prop}  \quad \hspace{-.6em} \quad  \Gamma   \vdash   \tau  \ \mathsf{type} }%
}{
 \Gamma   \vdash   \phi  \Rightarrow  \tau  \ \mathsf{type} }{%
{\ottdrulename{T\_Qual}}{}%
}}

\newcommand{\ottdrulePXXTypes}[1]{\ottdrule[#1]{%
\ottpremise{ \Gamma   \vdash   \tau_{{\mathrm{1}}}  \ \mathsf{type}  \quad \hspace{-.6em} \quad  \Gamma   \vdash   \tau_{{\mathrm{2}}}  \ \mathsf{type} }%
}{
 \Gamma   \vdash   \tau_{{\mathrm{1}}}  \sim  \tau_{{\mathrm{2}}}  \ \mathsf{prop} }{%
{\ottdrulename{P\_Types}}{}%
}}

\newcommand{\ottdrulePXXFamily}[1]{\ottdrule[#1]{%
\ottpremise{\ottmv{F} \, : \, \ottmv{n}  \in  \Sigma \quad \hspace{-.6em} \quad \ottcompu{ \Gamma   \vdash   \tau_{\ottmv{i}}  \ \mathsf{type} }{\ottmv{i}}{\ottmv{n}} \quad \hspace{-.6em} \quad  \Gamma   \vdash   \sigma  \ \mathsf{type} }%
}{
 \Gamma   \vdash   \ottmv{F} \, \overline{\tau}  \sim  \sigma  \ \mathsf{prop} }{%
{\ottdrulename{P\_Family}}{}%
}}

\newcommand{\ottdruleGXXNil}[1]{\ottdrule[#1]{%
}{
 \vdash   \emptyset  \ \mathsf{ctx} }{%
{\ottdrulename{G\_Nil}}{}%
}}

\newcommand{\ottdruleGXXTyVar}[1]{\ottdrule[#1]{%
\ottpremise{ \vdash   \Gamma  \ \mathsf{ctx} }%
\ottpremise{\alpha  \mathop{\#}  \Gamma}%
}{
 \vdash   \Gamma  \ottsym{,}  \alpha  \ \mathsf{ctx} }{%
{\ottdrulename{G\_TyVar}}{}%
}}

\newcommand{\ottdruleGXXCoVar}[1]{\ottdrule[#1]{%
\ottpremise{ \Gamma   \vdash   \phi  \ \mathsf{prop} }%
\ottpremise{\ottmv{c}  \mathop{\#}  \Gamma}%
}{
 \vdash   \Gamma  \ottsym{,}   \ottmv{c} {:} \phi   \ \mathsf{ctx} }{%
{\ottdrulename{G\_CoVar}}{}%
}}

\newcommand{\ottdruleGXXVar}[1]{\ottdrule[#1]{%
\ottpremise{ \Gamma   \vdash   \tau  \ \mathsf{type} }%
\ottpremise{\ottmv{x}  \mathop{\#}  \Gamma}%
}{
 \vdash   \Gamma  \ottsym{,}   \ottmv{x} {:} \tau   \ \mathsf{ctx} }{%
{\ottdrulename{G\_Var}}{}%
}}

\newcommand{\ottdruleDXXNil}[1]{\ottdrule[#1]{%
}{
 \vdash   \emptyset  \ \mathsf{ok} }{%
{\ottdrulename{D\_Nil}}{}%
}}

\newcommand{\ottdruleDXXPartial}[1]{\ottdrule[#1]{%
\ottpremise{ \vdash   \Sigma  \ \mathsf{ok} }%
}{
 \vdash   \Sigma  \ottsym{,}  \ottmv{F} \,  \mathop{:_{\not\top} }  \, \ottmv{n}  \ \mathsf{ok} }{%
{\ottdrulename{D\_Partial}}{}%
}}

\newcommand{\ottdruleDXXTotal}[1]{\ottdrule[#1]{%
\ottpremise{ \vdash   \Sigma  \ \mathsf{ok} }%
}{
 \vdash   \Sigma  \ottsym{,}  \ottmv{F} \,  \mathop{ :_{\top} }  \, \ottmv{n}  \ \mathsf{ok} }{%
{\ottdrulename{D\_Total}}{}%
}}

\newcommand{\ottdruleDXXAxiom}[1]{\ottdrule[#1]{%
\ottpremise{\ottmv{F} \, : \, \ottmv{n}  \in  \Sigma \quad \hspace{-.6em} \quad  \vdash   \Sigma  \ \mathsf{ok} }%
\ottpremise{\forall \, \ottmv{i}  \ottsym{:}}%
\ottpremise{\qquad  \ottnt{E_{\ottmv{i}}} \, \ottsym{=} \, \forall \, \overline{\alpha} \, \overline{\chi}  \ottsym{.}  \ottmv{F} \, \overline{\tau}  \sim  \tau_{{\mathrm{0}}}}%
\ottpremise{\qquad  \ottcomplu{ \overline{\alpha}   \vdash   \tau_{\ottmv{j}}  \ \mathsf{type} }{\ottmv{j}}{{\mathrm{1}}}{..}{\ottmv{n}}}%
\ottpremise{\qquad   \overline{\alpha}  \ottsym{,}  \mathit{tvs}\! \, \ottsym{(}  \overline{\chi}  \ottsym{)}   \vdash   \tau_{{\mathrm{0}}}  \ \mathsf{type} }%
\ottpremise{\qquad   \overline{\alpha}   \vdash   \overline{\chi} \ \mathsf{assumps} }%
}{
 \vdash   \Sigma  \ottsym{,}  \xi  \ottsym{:}  \overline{\ottnt{E} }  \ \mathsf{ok} }{%
{\ottdrulename{D\_Axiom}}{}%
}}

\newcommand{\ottdruleXXXNil}[1]{\ottdrule[#1]{%
}{
 \Gamma   \vdash   \emptyset \ \mathsf{assumps} }{%
{\ottdrulename{X\_Nil}}{}%
}}

\newcommand{\ottdruleXXXCons}[1]{\ottdrule[#1]{%
\ottpremise{\ottmv{F} \, : \, \ottmv{n}  \in  \Sigma}%
\ottpremise{\ottcomplu{ \Gamma   \vdash   \tau_{\ottmv{i}}  \ \mathsf{type} }{\ottmv{i}}{{\mathrm{1}}}{..}{\ottmv{n}}}%
\ottpremise{ \Gamma  \ottsym{,}  \alpha   \vdash   \overline{\chi} \ \mathsf{assumps} }%
}{
 \Gamma   \vdash    (  \alpha  \pipe  \ottmv{c}  :  \ottmv{F} \, \overline{\tau}   \sim   \alpha  )   \ottsym{,}  \overline{\chi} \ \mathsf{assumps} }{%
{\ottdrulename{X\_Cons}}{}%
}}

\newcommand{\ottdruleCXXRefl}[1]{\ottdrule[#1]{%
\ottpremise{ \Gamma   \vdash   \tau  \ \mathsf{type} }%
}{
\Gamma  \vdash   \langle  \tau  \rangle   \ottsym{:}  \tau  \sim  \tau}{%
{\ottdrulename{C\_Refl}}{}%
}}

\newcommand{\ottdruleCXXSym}[1]{\ottdrule[#1]{%
\ottpremise{\Gamma  \vdash  \gamma  \ottsym{:}  \tau_{{\mathrm{1}}}  \sim  \tau_{{\mathrm{2}}}}%
}{
\Gamma  \vdash  \ottkw{sym} \, \gamma  \ottsym{:}  \tau_{{\mathrm{2}}}  \sim  \tau_{{\mathrm{1}}}}{%
{\ottdrulename{C\_Sym}}{}%
}}

\newcommand{\ottdruleCXXTrans}[1]{\ottdrule[#1]{%
\ottpremise{\Gamma  \vdash  \gamma_{{\mathrm{1}}}  \ottsym{:}  \tau_{{\mathrm{1}}}  \sim  \tau_{{\mathrm{2}}}}%
\ottpremise{\Gamma  \vdash  \gamma_{{\mathrm{2}}}  \ottsym{:}  \tau_{{\mathrm{2}}}  \sim  \tau_{{\mathrm{3}}}}%
}{
\Gamma  \vdash  \gamma_{{\mathrm{1}}}  \fatsemi  \gamma_{{\mathrm{2}}}  \ottsym{:}  \tau_{{\mathrm{1}}}  \sim  \tau_{{\mathrm{3}}}}{%
{\ottdrulename{C\_Trans}}{}%
}}

\newcommand{\ottdruleCXXApp}[1]{\ottdrule[#1]{%
\ottpremise{\ottnt{H}  \ottsym{:}  \ottmv{n} \quad \hspace{-.6em} \quad  \vdash   \Gamma  \ \mathsf{ctx}  \quad \hspace{-.6em} \quad \ottcompu{\Gamma  \vdash  \gamma_{\ottmv{i}}  \ottsym{:}  \tau_{\ottmv{i}}  \sim  \sigma_{\ottmv{i}}}{\ottmv{i}}{\ottmv{n}}}%
}{
\Gamma  \vdash  \ottnt{H} \, \overline{\gamma}  \ottsym{:}  \ottnt{H} \, \overline{\tau}  \sim  \ottnt{H} \, \overline{\sigma}}{%
{\ottdrulename{C\_App}}{}%
}}

\newcommand{\ottdruleCXXFun}[1]{\ottdrule[#1]{%
\ottpremise{\Gamma  \vdash  \gamma_{{\mathrm{1}}}  \ottsym{:}  \tau_{{\mathrm{1}}}  \sim  \sigma_{{\mathrm{1}}} \quad \hspace{-.6em} \quad \Gamma  \vdash  \gamma_{{\mathrm{2}}}  \ottsym{:}  \tau_{{\mathrm{2}}}  \sim  \sigma_{{\mathrm{2}}}}%
}{
\Gamma  \vdash  \gamma_{{\mathrm{1}}}  \to  \gamma_{{\mathrm{2}}}  \ottsym{:}  \ottsym{(}  \tau_{{\mathrm{1}}}  \to  \tau_{{\mathrm{2}}}  \ottsym{)}  \sim  \ottsym{(}  \sigma_{{\mathrm{1}}}  \to  \sigma_{{\mathrm{2}}}  \ottsym{)}}{%
{\ottdrulename{C\_Fun}}{}%
}}

\newcommand{\ottdruleCXXForall}[1]{\ottdrule[#1]{%
\ottpremise{\Gamma  \ottsym{,}  \alpha  \vdash  \gamma  \ottsym{:}  \tau_{{\mathrm{1}}}  \sim  \tau_{{\mathrm{2}}}}%
}{
\Gamma  \vdash  \forall \, \alpha  \ottsym{.}  \gamma  \ottsym{:}  \ottsym{(}  \forall \, \alpha  \ottsym{.}  \tau_{{\mathrm{1}}}  \ottsym{)}  \sim  \ottsym{(}  \forall \, \alpha  \ottsym{.}  \tau_{{\mathrm{2}}}  \ottsym{)}}{%
{\ottdrulename{C\_Forall}}{}%
}}

\newcommand{\ottdruleCXXQual}[1]{\ottdrule[#1]{%
\ottpremise{\Gamma  \vdash  \gamma_{{\mathrm{1}}}  \ottsym{:}  \tau_{{\mathrm{1}}}  \sim  \sigma_{{\mathrm{1}}} \quad \hspace{-.6em} \quad \Gamma  \vdash  \gamma_{{\mathrm{2}}}  \ottsym{:}  \tau_{{\mathrm{2}}}  \sim  \sigma_{{\mathrm{2}}} \quad \hspace{-.6em} \quad \Gamma  \vdash  \gamma_{{\mathrm{3}}}  \ottsym{:}  \tau_{{\mathrm{3}}}  \sim  \sigma_{{\mathrm{3}}}}%
}{
\Gamma  \vdash  \gamma_{{\mathrm{1}}}  \sim  \gamma_{{\mathrm{2}}}  \Rightarrow  \gamma_{{\mathrm{3}}}  \ottsym{:}  \ottsym{(}  \tau_{{\mathrm{1}}}  \sim  \tau_{{\mathrm{2}}}  \Rightarrow  \tau_{{\mathrm{3}}}  \ottsym{)}  \sim  \ottsym{(}  \sigma_{{\mathrm{1}}}  \sim  \sigma_{{\mathrm{2}}}  \Rightarrow  \sigma_{{\mathrm{3}}}  \ottsym{)}}{%
{\ottdrulename{C\_Qual}}{}%
}}

\newcommand{\ottdruleCXXFam}[1]{\ottdrule[#1]{%
\ottpremise{\ottmv{F} \, : \, \ottmv{n}  \in  \Sigma \quad \hspace{-.6em} \quad  \vdash   \Gamma  \ \mathsf{ctx}  \quad \hspace{-.6em} \quad \ottcompu{\Gamma  \vdash  \gamma_{\ottmv{i}}  \ottsym{:}  \tau_{\ottmv{i}}  \sim  \sigma_{\ottmv{i}}}{\ottmv{i}}{\ottmv{n}}}%
}{
\Gamma  \vdash  \ottmv{F} \, \overline{\gamma}  \ottsym{:}  \ottmv{F} \, \overline{\tau}  \sim  \ottmv{F} \, \overline{\sigma}}{%
{\ottdrulename{C\_Fam}}{}%
}}

\newcommand{\ottdruleCXXNth}[1]{\ottdrule[#1]{%
\ottpremise{\Gamma  \vdash  \gamma  \ottsym{:}  \ottnt{H} \, \overline{\tau}  \sim  \ottnt{H} \, \overline{\sigma}}%
}{
\Gamma  \vdash   \ottkw{nth} _{ \ottmv{i} }\  \gamma   \ottsym{:}  \tau_{\ottmv{i}}  \sim  \sigma_{\ottmv{i}}}{%
{\ottdrulename{C\_Nth}}{}%
}}

\newcommand{\ottdruleCXXNthArrow}[1]{\ottdrule[#1]{%
\ottpremise{\Gamma  \vdash  \gamma  \ottsym{:}  \ottsym{(}  \tau_{{\mathrm{0}}}  \to  \tau_{{\mathrm{1}}}  \ottsym{)}  \sim  \ottsym{(}  \sigma_{{\mathrm{0}}}  \to  \sigma_{{\mathrm{1}}}  \ottsym{)}}%
}{
\Gamma  \vdash   \ottkw{nth} _{ \ottmv{i} }\  \gamma   \ottsym{:}  \tau_{\ottmv{i}}  \sim  \sigma_{\ottmv{i}}}{%
{\ottdrulename{C\_NthArrow}}{}%
}}

\newcommand{\ottdruleCXXNthQual}[1]{\ottdrule[#1]{%
\ottpremise{\Gamma  \vdash  \gamma  \ottsym{:}  \ottsym{(}  \tau_{{\mathrm{0}}}  \sim  \tau_{{\mathrm{1}}}  \Rightarrow  \tau_{{\mathrm{2}}}  \ottsym{)}  \sim  \ottsym{(}  \sigma_{{\mathrm{0}}}  \sim  \sigma_{{\mathrm{1}}}  \Rightarrow  \sigma_{{\mathrm{2}}}  \ottsym{)}}%
}{
\Gamma  \vdash   \ottkw{nth} _{ \ottmv{i} }\  \gamma   \ottsym{:}  \tau_{\ottmv{i}}  \sim  \sigma_{\ottmv{i}}}{%
{\ottdrulename{C\_NthQual}}{}%
}}

\newcommand{\ottdruleCXXInst}[1]{\ottdrule[#1]{%
\ottpremise{\Gamma  \vdash  \gamma  \ottsym{:}  \ottsym{(}  \forall \, \alpha  \ottsym{.}  \sigma_{{\mathrm{1}}}  \ottsym{)}  \sim  \ottsym{(}  \forall \, \alpha  \ottsym{.}  \sigma_{{\mathrm{2}}}  \ottsym{)}}%
\ottpremise{ \Gamma   \vdash   \tau  \ \mathsf{type} }%
}{
\Gamma  \vdash  \gamma  \at  \tau  \ottsym{:}  \sigma_{{\mathrm{1}}}  \ottsym{[}  \tau  \ottsym{/}  \alpha  \ottsym{]}  \sim  \sigma_{{\mathrm{2}}}  \ottsym{[}  \tau  \ottsym{/}  \alpha  \ottsym{]}}{%
{\ottdrulename{C\_Inst}}{}%
}}

\newcommand{\ottdruleCXXAxiom}[1]{\ottdrule[#1]{%
\ottpremise{\xi  \ottsym{:}  \overline{\ottnt{E} }  \in  \Sigma \quad \hspace{-.6em} \quad \ottnt{E_{\ottmv{i}}} \, \ottsym{=} \, \forall \, \overline{\alpha} \, \overline{\chi}  \ottsym{.}  \ottmv{F} \, \overline{\tau}  \sim  \tau_{{\mathrm{0}}} \quad \hspace{-.6em} \quad  \vdash   \Gamma  \ \mathsf{ctx} }%
\ottpremise{\ottcomp{ \Gamma   \vdash   \sigma_{\ottmv{j}}  \ \mathsf{type} }{\ottmv{j}} \quad \hspace{-.6em} \quad \Gamma  \vdash  \overline{\ottnt{q} }  \ottsym{:}  \overline{\chi}  \ottsym{[}  \overline{\sigma}  \ottsym{/}  \overline{\alpha}  \ottsym{]} \quad \hspace{-.6em} \quad \forall \, \ottmv{n}  \ottsym{<}  \ottmv{i}  \ottsym{,}   \mathsf{no\_conflict} ( \overline{\ottnt{E} } ,  \ottmv{i} ,  \overline{\sigma} ,  \ottmv{n} ) }%
}{
\Gamma  \vdash   \xi _{ \ottnt{i} }\  \overline{\sigma} \  \overline{\ottnt{q} }   \ottsym{:}  \ottmv{F} \, \overline{\tau}  \ottsym{[}  \overline{\sigma}  \ottsym{/}  \overline{\alpha}  \ottsym{]}  \sim  \tau_{{\mathrm{0}}}  \ottsym{[}  \overline{\sigma}  \ottsym{/}  \overline{\alpha}  \ottsym{,}  \overline{\ottnt{q} }  \ottsym{/}  \overline{\chi}  \ottsym{]}}{%
{\ottdrulename{C\_Axiom}}{}%
}}

\newcommand{\ottdruleCXXVar}[1]{\ottdrule[#1]{%
\ottpremise{ \ottmv{c} {:} \phi   \in  \Gamma}%
\ottpremise{ \vdash   \Gamma  \ \mathsf{ctx} }%
}{
\Gamma  \vdash  \ottmv{c}  \ottsym{:}  \phi}{%
{\ottdrulename{C\_Var}}{}%
}}

\newcommand{\ottdruleAXXNil}[1]{\ottdrule[#1]{%
\ottpremise{ \vdash   \Gamma  \ \mathsf{ctx} }%
}{
\Gamma  \vdash  \emptyset  \ottsym{:}  \emptyset}{%
{\ottdrulename{A\_Nil}}{}%
}}

\newcommand{\ottdruleAXXCons}[1]{\ottdrule[#1]{%
\ottpremise{ \Gamma   \vdash   \sigma  \ \mathsf{type}  \quad \hspace{-.6em} \quad \Gamma  \vdash  \gamma  \ottsym{:}  \ottmv{F} \, \overline{\tau}  \sim  \sigma \quad \hspace{-.6em} \quad \Gamma  \vdash  \overline{\ottnt{q} }  \ottsym{:}  \overline{\chi}  \ottsym{[}  \sigma  \ottsym{/}  \alpha  \ottsym{]}}%
}{
\Gamma  \vdash   ( \sigma  \pipe  \gamma )   \ottsym{,}  \overline{\ottnt{q} }  \ottsym{:}   (  \alpha  \pipe  \ottmv{c}  :  \ottmv{F} \, \overline{\tau}   \sim   \alpha  )   \ottsym{,}  \overline{\chi}}{%
{\ottdrulename{A\_Cons}}{}%
}}

\newcommand{\ottdruleNCXXApart}[1]{\ottdrule[#1]{%
\ottpremise{\ottnt{E_{\ottmv{i}}} \, \ottsym{=} \, \forall \, \overline{\alpha}_{{\mathrm{1}}} \, \overline{\chi}_{{\mathrm{1}}}  \ottsym{.}  \ottmv{F} \, \overline{\tau}_{{\mathrm{1}}}  \sim  \tau_{{\mathrm{01}}}}%
\ottpremise{\ottnt{E_{\ottmv{j}}} \, \ottsym{=} \, \forall \, \overline{\alpha}_{{\mathrm{2}}} \, \overline{\chi}_{{\mathrm{2}}}  \ottsym{.}  \ottmv{F} \, \overline{\tau}_{{\mathrm{2}}}  \sim  \tau_{{\mathrm{02}}}}%
\ottpremise{ \mathsf{unify} ( \overline{\tau}_{{\mathrm{2}}} ;\, \overline{\tau}_{{\mathrm{1}}}  \ottsym{[}  \overline{\sigma}  \ottsym{/}  \overline{\alpha}_{{\mathrm{1}}}  \ottsym{]} )  \, \ottsym{=} \, \mathsf{Nothing}}%
}{
 \mathsf{no\_conflict} ( \overline{\ottnt{E} } ,  \ottmv{i} ,  \overline{\sigma} ,  \ottmv{j} ) }{%
{\ottdrulename{NC\_Apart}}{}%
}}

\newcommand{\ottdruleNCXXCompatible}[1]{\ottdrule[#1]{%
\ottpremise{ \mathsf{compat} ( \ottnt{E_{\ottmv{i}}} ,  \ottnt{E_{\ottmv{j}}} ) }%
}{
 \mathsf{no\_conflict} ( \overline{\ottnt{E} } ,  \ottmv{i} ,  \overline{\sigma} ,  \ottmv{j} ) }{%
{\ottdrulename{NC\_Compatible}}{}%
}}

\newcommand{\ottdruleCoXXCoinc}[1]{\ottdrule[#1]{%
\ottpremise{\ottnt{E_{{\mathrm{1}}}} \, \ottsym{=} \, \forall \, \overline{\alpha}_{{\mathrm{1}}} \, \overline{\chi}_{{\mathrm{1}}}  \ottsym{.}  \ottmv{F} \, \overline{\tau}_{{\mathrm{1}}}  \sim  \tau_{{\mathrm{01}}}}%
\ottpremise{\ottnt{E_{{\mathrm{2}}}} \, \ottsym{=} \, \forall \, \overline{\alpha}_{{\mathrm{2}}} \, \overline{\chi}_{{\mathrm{2}}}  \ottsym{.}  \ottmv{F} \, \overline{\tau}_{{\mathrm{2}}}  \sim  \tau_{{\mathrm{02}}}}%
\ottpremise{ \mathsf{unify} ( \overline{\tau}_{{\mathrm{1}}} ;\, \overline{\tau}_{{\mathrm{2}}} )  \, \ottsym{=} \, \mathsf{Just} \, \theta}%
\ottpremise{\tau_{{\mathrm{01}}}  \ottsym{[}  \theta  \circ  \mathit{subst}\! \, \ottsym{(}  \overline{\chi}_{{\mathrm{1}}}  \ottsym{)}  \ottsym{]} \, \ottsym{=} \, \tau_{{\mathrm{02}}}  \ottsym{[}  \theta  \circ  \mathit{subst}\! \, \ottsym{(}  \overline{\chi}_{{\mathrm{2}}}  \ottsym{)}  \ottsym{]}}%
}{
 \mathsf{compat} ( \ottnt{E_{{\mathrm{1}}}} ,  \ottnt{E_{{\mathrm{2}}}} ) }{%
{\ottdrulename{Co\_Coinc}}{}%
}}

\newcommand{\ottdruleCoXXDistinct}[1]{\ottdrule[#1]{%
\ottpremise{\ottnt{E_{{\mathrm{1}}}} \, \ottsym{=} \, \forall \, \overline{\alpha}_{{\mathrm{1}}} \, \overline{\chi}_{{\mathrm{1}}}  \ottsym{.}  \ottmv{F} \, \overline{\tau}_{{\mathrm{1}}}  \sim  \tau_{{\mathrm{01}}}}%
\ottpremise{\ottnt{E_{{\mathrm{2}}}} \, \ottsym{=} \, \forall \, \overline{\alpha}_{{\mathrm{2}}} \, \overline{\chi}_{{\mathrm{2}}}  \ottsym{.}  \ottmv{F} \, \overline{\tau}_{{\mathrm{2}}}  \sim  \tau_{{\mathrm{02}}}}%
\ottpremise{ \mathsf{unify} ( \overline{\tau}_{{\mathrm{1}}} ;\, \overline{\tau}_{{\mathrm{2}}} )  \, \ottsym{=} \, \mathsf{Nothing}}%
}{
 \mathsf{compat} ( \ottnt{E_{{\mathrm{1}}}} ,  \ottnt{E_{{\mathrm{2}}}} ) }{%
{\ottdrulename{Co\_Distinct}}{}%
}}

\newcommand{\ottdruleEXXVar}[1]{\ottdrule[#1]{%
\ottpremise{ \ottmv{x} {:} \tau   \in  \Gamma \quad \hspace{-.6em} \quad  \vdash   \Gamma  \ \mathsf{ctx} }%
}{
\Gamma  \vdash  \ottmv{x}  \ottsym{:}  \tau}{%
{\ottdrulename{E\_Var}}{}%
}}

\newcommand{\ottdruleEXXConst}[1]{\ottdrule[#1]{%
\ottpremise{\ottmv{K}  \ottsym{:}  \ottnt{H} \, \overline{\tau} \quad \hspace{-.6em} \quad  \vdash   \Gamma  \ \mathsf{ctx} }%
}{
\Gamma  \vdash  \ottmv{K}  \ottsym{:}  \ottnt{H}}{%
{\ottdrulename{E\_Const}}{}%
}}

\newcommand{\ottdruleEXXLam}[1]{\ottdrule[#1]{%
\ottpremise{\Gamma  \ottsym{,}   \ottmv{x} {:} \tau_{{\mathrm{1}}}   \vdash  \ottnt{e}  \ottsym{:}  \tau_{{\mathrm{2}}}}%
}{
\Gamma  \vdash  \lambda  \ottmv{x}  \ottsym{:}  \tau_{{\mathrm{1}}}  \ottsym{.}  \ottnt{e}  \ottsym{:}  \tau_{{\mathrm{1}}}  \to  \tau_{{\mathrm{2}}}}{%
{\ottdrulename{E\_Lam}}{}%
}}

\newcommand{\ottdruleEXXApp}[1]{\ottdrule[#1]{%
\ottpremise{\Gamma  \vdash  \ottnt{e_{{\mathrm{1}}}}  \ottsym{:}  \tau_{{\mathrm{1}}}  \to  \tau_{{\mathrm{2}}}}%
\ottpremise{\Gamma  \vdash  \ottnt{e_{{\mathrm{2}}}}  \ottsym{:}  \tau_{{\mathrm{1}}}}%
}{
\Gamma  \vdash  \ottnt{e_{{\mathrm{1}}}} \, \ottnt{e_{{\mathrm{2}}}}  \ottsym{:}  \tau_{{\mathrm{2}}}}{%
{\ottdrulename{E\_App}}{}%
}}

\newcommand{\ottdruleEXXTLam}[1]{\ottdrule[#1]{%
\ottpremise{\Gamma  \ottsym{,}  \alpha  \vdash  \ottnt{e}  \ottsym{:}  \tau}%
}{
\Gamma  \vdash  \Lambda  \alpha  \ottsym{.}  \ottnt{e}  \ottsym{:}  \forall \, \alpha  \ottsym{.}  \tau}{%
{\ottdrulename{E\_TLam}}{}%
}}

\newcommand{\ottdruleEXXTApp}[1]{\ottdrule[#1]{%
\ottpremise{\Gamma  \vdash  \ottnt{e}  \ottsym{:}  \forall \, \alpha  \ottsym{.}  \tau}%
\ottpremise{ \Gamma   \vdash   \sigma  \ \mathsf{type} }%
}{
\Gamma  \vdash  \ottnt{e} \, \sigma  \ottsym{:}  \tau  \ottsym{[}  \sigma  \ottsym{/}  \alpha  \ottsym{]}}{%
{\ottdrulename{E\_TApp}}{}%
}}

\newcommand{\ottdruleEXXCLam}[1]{\ottdrule[#1]{%
\ottpremise{\Gamma  \ottsym{,}   \ottmv{c} {:} \phi   \vdash  \ottnt{e}  \ottsym{:}  \tau}%
}{
\Gamma  \vdash  \lambda  \ottmv{c}  \ottsym{:}  \phi  \ottsym{.}  \ottnt{e}  \ottsym{:}  \phi  \Rightarrow  \tau}{%
{\ottdrulename{E\_CLam}}{}%
}}

\newcommand{\ottdruleEXXCApp}[1]{\ottdrule[#1]{%
\ottpremise{\Gamma  \vdash  \ottnt{e}  \ottsym{:}  \phi  \Rightarrow  \tau}%
\ottpremise{\Gamma  \vdash  \gamma  \ottsym{:}  \phi}%
}{
\Gamma  \vdash  \ottnt{e} \, \gamma  \ottsym{:}  \tau}{%
{\ottdrulename{E\_CApp}}{}%
}}

\newcommand{\ottdruleEXXCast}[1]{\ottdrule[#1]{%
\ottpremise{\Gamma  \vdash  \ottnt{e}  \ottsym{:}  \tau_{{\mathrm{1}}} \quad \hspace{-.6em} \quad \Gamma  \vdash  \gamma  \ottsym{:}  \tau_{{\mathrm{1}}}  \sim  \tau_{{\mathrm{2}}}}%
\ottpremise{ \Gamma   \vdash   \tau_{{\mathrm{2}}}  \ \mathsf{type} }%
}{
\Gamma  \vdash  \ottnt{e}  \triangleright  \gamma  \ottsym{:}  \tau_{{\mathrm{2}}}}{%
{\ottdrulename{E\_Cast}}{}%
}}

\newcommand{\ottdruleEXXAssume}[1]{\ottdrule[#1]{%
\ottpremise{\ottmv{F} \,  \mathop{ :_{\top} }  \, \ottmv{n}  \in  \Sigma \quad \hspace{-.6em} \quad \ottcompu{ \Gamma   \vdash   \tau_{\ottmv{i}}  \ \mathsf{type} }{\ottmv{i}}{\ottmv{n}}}%
\ottpremise{\Gamma  \ottsym{,}  \alpha  \ottsym{,}   \ottmv{c} {:} \ottmv{F} \, \overline{\tau}  \sim  \alpha   \vdash  \ottnt{e}  \ottsym{:}  \sigma \quad \hspace{-.6em} \quad \alpha  \not\in  \mathit{fv}\! \, \ottsym{(}  \sigma  \ottsym{)}}%
}{
\Gamma  \vdash  \ottkw{assume} \,  (  \alpha  \pipe  \ottmv{c}  :  \ottmv{F} \, \overline{\tau}   \sim   \alpha  )  \, \ottkw{in} \, \ottnt{e}  \ottsym{:}  \sigma}{%
{\ottdrulename{E\_Assume}}{}%
}}

\newcommand{\ottdruleSXXApp}[1]{\ottdrule[#1]{%
\ottpremise{\ottnt{e_{{\mathrm{1}}}}  \longrightarrow  \ottnt{e'_{{\mathrm{1}}}}}%
}{
\ottnt{e_{{\mathrm{1}}}} \, \ottnt{e_{{\mathrm{2}}}}  \longrightarrow  \ottnt{e'_{{\mathrm{1}}}} \, \ottnt{e_{{\mathrm{2}}}}}{%
{\ottdrulename{S\_App}}{}%
}}

\newcommand{\ottdruleSXXBeta}[1]{\ottdrule[#1]{%
}{
\ottsym{(}  \lambda  \ottmv{x}  \ottsym{:}  \tau  \ottsym{.}  \ottnt{e_{{\mathrm{1}}}}  \ottsym{)} \, \ottnt{e_{{\mathrm{2}}}}  \longrightarrow  \ottnt{e_{{\mathrm{1}}}}  \ottsym{[}  \ottnt{e_{{\mathrm{2}}}}  \ottsym{/}  \ottmv{x}  \ottsym{]}}{%
{\ottdrulename{S\_Beta}}{}%
}}

\newcommand{\ottdruleSXXPush}[1]{\ottdrule[#1]{%
\ottpremise{\ottnt{v} \, \ottsym{=} \, \lambda  \ottmv{x}  \ottsym{:}  \tau  \ottsym{.}  \ottnt{e_{{\mathrm{0}}}}}%
\ottpremise{\gamma_{{\mathrm{1}}} \, \ottsym{=} \, \ottkw{sym} \, \ottsym{(}   \ottkw{nth} _{ \ottsym{0} }\  \gamma   \ottsym{)} \quad \hspace{-.6em} \quad \gamma_{{\mathrm{2}}} \, \ottsym{=} \,  \ottkw{nth} _{ \ottsym{1} }\  \gamma }%
}{
\ottsym{(}  \ottnt{v}  \triangleright  \gamma  \ottsym{)} \, \ottnt{e}  \longrightarrow  \ottnt{v} \, \ottsym{(}  \ottnt{e}  \triangleright  \gamma_{{\mathrm{1}}}  \ottsym{)}  \triangleright  \gamma_{{\mathrm{2}}}}{%
{\ottdrulename{S\_Push}}{}%
}}

\newcommand{\ottdruleSXXTApp}[1]{\ottdrule[#1]{%
\ottpremise{\ottnt{e}  \longrightarrow  \ottnt{e'}}%
}{
\ottnt{e} \, \tau  \longrightarrow  \ottnt{e'} \, \tau}{%
{\ottdrulename{S\_TApp}}{}%
}}

\newcommand{\ottdruleSXXTBeta}[1]{\ottdrule[#1]{%
}{
\ottsym{(}  \Lambda  \alpha  \ottsym{.}  \ottnt{e}  \ottsym{)} \, \tau  \longrightarrow  \ottnt{e}  \ottsym{[}  \tau  \ottsym{/}  \alpha  \ottsym{]}}{%
{\ottdrulename{S\_TBeta}}{}%
}}

\newcommand{\ottdruleSXXTPush}[1]{\ottdrule[#1]{%
\ottpremise{\ottnt{v} \, \ottsym{=} \, \Lambda  \alpha  \ottsym{.}  \ottnt{e}}%
\ottpremise{\gamma' \, \ottsym{=} \, \gamma  \at  \tau}%
}{
\ottsym{(}  \ottnt{v}  \triangleright  \gamma  \ottsym{)} \, \tau  \longrightarrow  \ottnt{v} \, \tau  \triangleright  \gamma'}{%
{\ottdrulename{S\_TPush}}{}%
}}

\newcommand{\ottdruleSXXCApp}[1]{\ottdrule[#1]{%
\ottpremise{\ottnt{e}  \longrightarrow  \ottnt{e'}}%
}{
\ottnt{e} \, \gamma  \longrightarrow  \ottnt{e'} \, \gamma}{%
{\ottdrulename{S\_CApp}}{}%
}}

\newcommand{\ottdruleSXXCBeta}[1]{\ottdrule[#1]{%
}{
\ottsym{(}  \lambda  \ottmv{c}  \ottsym{:}  \phi  \ottsym{.}  \ottnt{e}  \ottsym{)} \, \gamma  \longrightarrow  \ottnt{e}  \ottsym{[}  \gamma  \ottsym{/}  \ottmv{c}  \ottsym{]}}{%
{\ottdrulename{S\_CBeta}}{}%
}}

\newcommand{\ottdruleSXXCPush}[1]{\ottdrule[#1]{%
\ottpremise{\ottnt{v} \, \ottsym{=} \, \lambda  \ottmv{c}  \ottsym{:}  \phi  \ottsym{.}  \ottnt{e_{{\mathrm{0}}}} \quad \hspace{-.6em} \quad \eta_{{\mathrm{0}}} \, \ottsym{=} \,  \ottkw{nth} _{ \ottsym{0} }\  \eta }%
\ottpremise{\eta_{{\mathrm{1}}} \, \ottsym{=} \, \ottkw{sym} \, \ottsym{(}   \ottkw{nth} _{ \ottsym{1} }\  \eta   \ottsym{)} \quad \hspace{-.6em} \quad \eta_{{\mathrm{2}}} \, \ottsym{=} \,  \ottkw{nth} _{ \ottsym{2} }\  \eta }%
}{
\ottsym{(}  \ottnt{v}  \triangleright  \eta  \ottsym{)} \, \gamma  \longrightarrow  \ottnt{v} \, \ottsym{(}  \eta_{{\mathrm{0}}}  \fatsemi  \gamma  \fatsemi  \eta_{{\mathrm{1}}}  \ottsym{)}  \triangleright  \eta_{{\mathrm{2}}}}{%
{\ottdrulename{S\_CPush}}{}%
}}

\newcommand{\ottdruleSXXCast}[1]{\ottdrule[#1]{%
\ottpremise{\ottnt{e}  \longrightarrow  \ottnt{e'}}%
}{
\ottnt{e}  \triangleright  \gamma  \longrightarrow  \ottnt{e'}  \triangleright  \gamma}{%
{\ottdrulename{S\_Cast}}{}%
}}

\newcommand{\ottdruleSXXTrans}[1]{\ottdrule[#1]{%
}{
\ottsym{(}  \ottnt{v}  \triangleright  \gamma_{{\mathrm{1}}}  \ottsym{)}  \triangleright  \gamma_{{\mathrm{2}}}  \longrightarrow  \ottnt{v}  \triangleright  \ottsym{(}  \gamma_{{\mathrm{1}}}  \fatsemi  \gamma_{{\mathrm{2}}}  \ottsym{)}}{%
{\ottdrulename{S\_Trans}}{}%
}}

\newcommand{\ottdruleSXXResolve}[1]{\ottdrule[#1]{%
\ottpremise{\chi \, \ottsym{=} \,  (  \alpha  \pipe  \ottmv{c}  :  \ottmv{F} \, \overline{\tau}   \sim   \alpha  )  \quad \hspace{-.6em} \quad  \ottmv{F} \  \overline{\tau}  \Downarrow  \ottnt{q} }%
}{
\ottkw{assume} \, \chi \, \ottkw{in} \, \ottnt{e}  \longrightarrow  \ottnt{e}  \ottsym{[}  \ottnt{q}  \ottsym{/}  \chi  \ottsym{]}}{%
{\ottdrulename{S\_Resolve}}{}%
}}

\newcommand{\ottdruleRTop}[1]{\ottdrule[#1]{%
\ottpremise{\xi  \ottsym{:}  \overline{\ottnt{E} }  \in  \Sigma \quad \hspace{-.6em} \quad \ottnt{E_{\ottmv{i}}} \, \ottsym{=} \, \forall \, \overline{\alpha} \, \overline{\chi}  \ottsym{.}  \ottmv{F} \, \overline{\sigma}  \sim  \sigma_{{\mathrm{0}}}}%
\ottpremise{\overline{\tau} \, \ottsym{=} \, \overline{\sigma}  \ottsym{[}  \overline{\rho}  \ottsym{/}  \overline{\alpha}  \ottsym{]} \quad \hspace{-.6em} \quad \tau' \, \ottsym{=} \, \sigma_{{\mathrm{0}}}  \ottsym{[}  \overline{\rho}  \ottsym{/}  \overline{\alpha}  \ottsym{,}  \overline{\rho}'  \ottsym{/}  \mathit{tvs}\! \, \ottsym{(}  \overline{\chi}  \ottsym{)}  \ottsym{]}}%
\ottpremise{\ottcomp{ \emptyset   \vdash   \rho_{\ottmv{k}}  \ \mathsf{type} }{\ottmv{k}} \quad \hspace{-.6em} \quad \forall \, \ottmv{j}  \ottsym{<}  \ottmv{i}  \ottsym{,}   \mathsf{no\_conflict} ( \overline{\ottnt{E} } ,  \ottmv{i} ,  \overline{\rho} ,  \ottmv{j} ) }%
\ottpremise{\forall \, \ottmv{n}  \ottsym{:}}%
\ottpremise{\qquad  \chi_{\ottmv{n}} \, \ottsym{=} \,  (  \alpha'  \pipe  \ottmv{c'}  :  \ottmv{F'} \, \overline{\sigma}'   \sim   \alpha'  ) }%
\ottpremise{\qquad   \emptyset   \vdash   \rho'_{\ottmv{n}}  \ \mathsf{type} }%
\ottpremise{\qquad  \theta_{\ottmv{n}} \, \ottsym{=} \, \overline{\rho}  \ottsym{/}  \overline{\alpha}  \ottsym{,} \, \ottcomplu{\rho'_{\ottmv{m}}  \ottsym{/}  \mathit{tv}\! \, \ottsym{(}  \chi_{\ottmv{m}}  \ottsym{)}}{\ottmv{m}}{{\mathrm{1}}}{..}{{\ottmv{n}-1}}}%
\ottpremise{\qquad  \ottmv{F'} \, \overline{\sigma}'  \ottsym{[}  \theta_{\ottmv{n}}  \ottsym{]}  \rightsquigarrow_{\top}  \rho'_{\ottmv{n}}}%
}{
\ottmv{F} \, \overline{\tau}  \rightsquigarrow_{\top}  \tau'}{%
{\ottdrulename{RTop}}{}%
}}

\newcommand{\ottdefnRedT}[1]{\begin{ottdefnblock}[#1]{$\tau_{{\mathrm{1}}}  \rightsquigarrow_{\top}  \tau_{{\mathrm{2}}}$}{}
\ottusedrule{\ottdruleRTop{}}
\end{ottdefnblock}}

\newcommand{\ottdruleRed}[1]{\ottdrule[#1]{%
\ottpremise{\ottmv{F} \, \overline{\tau}  \rightsquigarrow_{\top}  \tau'}%
}{
\mathcal{C}  \ottsym{[}  \ottmv{F} \, \overline{\tau}  \ottsym{]}  \rightsquigarrow  \mathcal{C}  \ottsym{[}  \tau'  \ottsym{]}}{%
{\ottdrulename{Red}}{}%
}}

\newcommand{\ottdefnRed}[1]{\begin{ottdefnblock}[#1]{$\tau_{{\mathrm{1}}}  \rightsquigarrow  \tau_{{\mathrm{2}}}$}{\ottcom{Type reduction}}
\ottusedrule{\ottdruleRed{}}
\end{ottdefnblock}}

\renewcommand{\ottdrule}[4][]{{\displaystyle\frac{\begin{array}{l}#2\end{array}}{#3}\;\ottdrulename{#4}}}

\newenvironment{judg}[2]{ \framebox{\mbox{#1}} \quad #2 \csname gather*\endcsname}{\csname endgather*\endcsname}

\newcommand{\keyword}[1]{\textsf{\textbf{#1}}}
\newcommand{\id}[1]{\textsf{\textsl{#1}}}

\newcommand{\rul}[1]{\ottdrulename{#1}}

\usepackage{partiality}

\acmJournal{PACMPL}
\acmConference[ICFP'17]{ACM SIGPLAN International Conference on Functional Programming}{September 04--06, 2017}{Oxford, UK}
\acmYear{2017}
\acmISBN{978-x-xxxx-xxxx-x/YY/MM}
\acmDOI{10.1145/nnnnnnn.nnnnnnn}
\startPage{1}

\setcopyright{acmlicensed}
\acmJournal{PACMPL}
\acmYear{2017} \acmVolume{1} \acmNumber{1} \acmArticle{42} \acmMonth{1} \acmDOI{http://dx.doi.org/10.1145/3110286}

\title{Constrained Type Families\ifextended{} (extended version)\fi}

\author{J. Garrett Morris}
\orcid{0000-0002-3992-1080}
\affiliation{
  \institution{The University of Edinburgh}
  \department{Laboratory for Foundations of Computer Science}
  \streetaddress{10 Crichton Street}
  \city{Edinburgh}
  \country{UK}}
\affiliation{
  \institution{The University of Kansas}
  \department{Information and Telecommunication Technology Center}
  \city{Lawrence}
  \state{KS}
  \country{USA}}
\email{garrett@ittc.ku.edu}

\author{Richard A. Eisenberg}
\affiliation{
  \institution{Bryn Mawr College}
  \department{Department of Computer Science}
  \city{Bryn Mawr}
  \state{PA}
  \country{USA}}
\email{rae@cs.brynmawr.edu}

\acmBadgeR[http://icfp17.sigplan.org/track/icfp-2017-Artifacts]{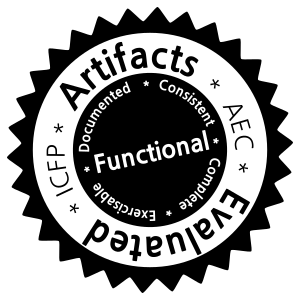}
\acmBadgeL[http://icfp17.sigplan.org/track/icfp-2017-Artifacts]{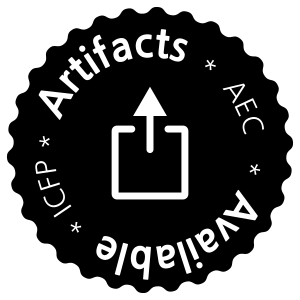}

\begin{document}

\begin{abstract}
  We present an approach to support partiality in type-level computation without compromising
  expressiveness or type safety. Existing frameworks for type-level computation either require
  totality or implicitly assume it. For example, type families in Haskell provide a powerful,
  modular means of defining type-level computation. However, their current design implicitly assumes
  that type families are total, introducing nonsensical types and significantly complicating the
  metatheory of type families and their extensions. We propose an alternative design,
  using qualified types to pair type-level computations with predicates that capture their
  domains. Our approach naturally captures the intuitive partiality of type families, simplifying
  their metatheory. As evidence, we present the first complete proof of consistency for a language
  with closed type families.
\end{abstract}





\maketitle

\section{Introduction}
\label{sec:intro}

\emph{Indexed type families}~\cite{ChakravartyKJ05,SchrijversPJCS08} extend the Haskell type system with modular
type-level computation.  They allow programmers to define and use open mappings from types to types.
These have given rise to further extensions of the language, such as closed type
families~\cite{EisenbergVPJW14} and injective type families~\cite{StolarekJE15}, and they have many
applications, including encoding units of measure in scientific calculation~\cite{MuranushiE14} and
extensible variants~\cite{Bahr14,Morris15}.

Nevertheless, some aspects of type families 
remain counterintuitive. 
Consider a unary type family \ensuremath{\id{F}} with no defining equations.  A type expression such as \ensuremath{\id{F}\;\id{Int}} should be meaningless---quite literally, as there are no equations for \ensuremath{\id{F}} to give it meaning.
Nevertheless, not only is \ensuremath{\id{F}\;\id{Int}} a type, but there are simple programs (such as divergence) that
demonstrate its inhabitation.  This apparent paradox has both practical and theoretical
consequences.
For example, we define a closed type family \ensuremath{\id{Equ}} such that \ensuremath{\id{Equ}\;\tau\;\sigma} should be \ensuremath{\id{True}} iff \ensuremath{\tau}
and \ensuremath{\sigma} are the same type:\footnote{We use here the promoted \ensuremath{\id{Bool}} kind,
  as introduced by \citet{YorgeyWCJVM12}.}
\begin{hscode}\SaveRestoreHook
\column{B}{@{}>{\hspre}l<{\hspost}@{}}%
\column{3}{@{}>{\hspre}l<{\hspost}@{}}%
\column{12}{@{}>{\hspre}c<{\hspost}@{}}%
\column{12E}{@{}l@{}}%
\column{15}{@{}>{\hspre}l<{\hspost}@{}}%
\column{E}{@{}>{\hspre}l<{\hspost}@{}}%
\>[B]{}\keyword{type}\;\keyword{family}\;\id{Equ}\;\id{a}\;\id{b}\mathbin{::}\id{Bool}\;\keyword{where}{}\<[E]%
\\
\>[B]{}\hsindent{3}{}\<[3]%
\>[3]{}\id{Equ}\;\id{a}\;\id{a}{}\<[12]%
\>[12]{}\mathrel{=}{}\<[12E]%
\>[15]{}\id{True}{}\<[E]%
\\
\>[B]{}\hsindent{3}{}\<[3]%
\>[3]{}\id{Equ}\;\id{a}\;\id{b}{}\<[12]%
\>[12]{}\mathrel{=}{}\<[12E]%
\>[15]{}\id{False}{}\<[E]%
\ColumnHook
\end{hscode}\resethooks
The type family application \ensuremath{\id{Equ}\;\id{a}\;[\mskip1.5mu \id{a}\mskip1.5mu]} does not
reduce.  This surprising fact stems from the use of \emph{infinitary}
unification during closed type family reduction~\cite[\S6.2]{EisenbergVPJW14}.
This explanation
raises more questions, however: Haskell does not have infinite types, so why use
 infinitary unification?  Again, type families play a role.  Consider the following:
\begin{hscode}\SaveRestoreHook
\column{B}{@{}>{\hspre}l<{\hspost}@{}}%
\column{E}{@{}>{\hspre}l<{\hspost}@{}}%
\>[B]{}\keyword{type}\;\keyword{family}\;\id{Loop}\mathbin{::}\star{}\<[E]%
\\
\>[B]{}\keyword{type}\;\keyword{instance}\;\id{Loop}\mathrel{=}[\mskip1.5mu \id{Loop}\mskip1.5mu]{}\<[E]%
\ColumnHook
\end{hscode}\resethooks
The type family application \ensuremath{\id{Loop}} will never rewrite to a ground type.  But, \ensuremath{\id{Equ}\;\id{Loop}\;[\mskip1.5mu \id{Loop}\mskip1.5mu]} is equal to \ensuremath{\id{Equ}\;[\mskip1.5mu \id{Loop}\mskip1.5mu]\;[\mskip1.5mu \id{Loop}\mskip1.5mu]}, and thus to \ensuremath{\id{True}}, justifying
not rewriting \ensuremath{\id{Equ}\;\id{a}\;[\mskip1.5mu \id{a}\mskip1.5mu]} to \ensuremath{\id{False}}.  

The complexity in this example arises from
an underlying inconsistency in the notion of type families.  Type families are used
identically to other type constructors; that is, uses of type families come with an unstated assumption
of totality, regardless of the equations 
in the program.
For example, \ensuremath{\id{Loop}} will never reduce to any ground (i.e., type
family-free) Haskell type, but still must be treated as a type for the purposes of reducing \ensuremath{\id{Equ}\;\id{a}\;[\mskip1.5mu \id{a}\mskip1.5mu]}.  In essence, \ensuremath{\id{Loop}} is treated as a total 0-ary function on types, even though its definition
makes it partial.
Similar problems arise in injective type families and in interpreting
definitions using open type families~\secref{partiality}.

We propose a refinement of indexed type families, \emph{constrained type families}, which explicitly
captures partiality in the definition and use of type families.
In our design, partial type families be defined associated with type classes.
  Thus, the domain of a type family is
naturally characterized by its corresponding type class predicate.
We further
insist that uses of type family be qualified by their defining class predicates, guaranteeing that
they be well-defined.  Non-terminating, or otherwise undefinable, type family applications will be
guarded by unsatisfiable class predicates, assuring that they cannot be used to violate type safety.



The introduction of constraints simplifies the metatheory of type families, separating concerns
about partiality from the machinery of rewriting.  We demonstrate this by formally specifying a
calculus with constrained closed type families and showing it is sound, relying on neither infinitary
unification nor an assumption of termination, in contrast to previous work on type
families~\cite{EisenbergVPJW14}.  In terms of our earlier example, this means that we could safely
rewrite \ensuremath{\id{Equ}\;\id{a}\;[\mskip1.5mu \id{a}\mskip1.5mu]} to \ensuremath{\id{False}} without risking type safety.  

In summary, this paper contributes:
\begin{itemize}
\item An analysis of difficulties in the evolution of type families, including the need for
  infinitary unification in the semantics of closed type families and the inexpressiveness of
  injective type families. These warts on the type system belie a hidden assumption of
  totality~\secref{partiality}.
\item The design of constrained type families~\secref{predicates}, which relaxes the assumption of
  totality by using type class predicates to characterize the domains of definition of partial type
  families.
\item The design of closed type classes~\secref{closed}, a simplification of instance
  chains~\cite{MorrisJ10}. Closed type classes enable partial closed type families and increase the
  expressiveness of type classes, subsuming many uses of overlapping
  instances~\cite{PeytonJonesJM97}.
\item A core calculus with constrained type families~\secref{formal}, together with a proof of its
  soundness that requires neither an assumption of termination nor infinitary unification. This is a
  novel result for a calculus supporting type families with non-linear patterns. Even with
  infinitary unification, prior work~\cite{EisenbergVPJW14} was
  unable to fully prove consistency.
\item A design that allows existing Haskell code to remain well typed, so long as it does not
  depend on the behavior of undefined type family applications~\secref{compat}.
\end{itemize}
Although this paper is primarily concerned with Haskell,
our work is applicable to any partial language that
supports type-level computation. 
We hope that our work, among others', will encourage other languages to join in the type-level fun.

\section{Type Families in Haskell}
\label{sec:background}

Associated type synonyms~\cite{ChakravartyKJ05} are a feature of the Haskell type system that allows
the definition and use of extensible maps from types to types.  They address many of the problems
that arise in the use of multi-parameter type classes.  One example is a class of collection types,
\ensuremath{\id{Collects}}.  In defining the \ensuremath{\id{Collects}} methods for a type \ensuremath{\id{c}}, we naturally need access to the
types of its elements.  To capture the types of collection elements, we could define the \ensuremath{\id{Collects}}
class to have an associated type \ensuremath{\id{Elem}}:
\begin{hscode}\SaveRestoreHook
\column{B}{@{}>{\hspre}l<{\hspost}@{}}%
\column{3}{@{}>{\hspre}l<{\hspost}@{}}%
\column{11}{@{}>{\hspre}c<{\hspost}@{}}%
\column{11E}{@{}l@{}}%
\column{15}{@{}>{\hspre}l<{\hspost}@{}}%
\column{E}{@{}>{\hspre}l<{\hspost}@{}}%
\>[B]{}\keyword{class}\;\id{Collects}\;\id{c}\;\keyword{where}{}\<[E]%
\\
\>[B]{}\hsindent{3}{}\<[3]%
\>[3]{}\keyword{type}\;\id{Elem}\;\id{c}\mathbin{::}\star{}\<[E]%
\\
\>[B]{}\hsindent{3}{}\<[3]%
\>[3]{}\id{empty}{}\<[11]%
\>[11]{}\mathbin{::}{}\<[11E]%
\>[15]{}\id{c}{}\<[E]%
\\
\>[B]{}\hsindent{3}{}\<[3]%
\>[3]{}\id{insert}{}\<[11]%
\>[11]{}\mathbin{::}{}\<[11E]%
\>[15]{}\id{Elem}\;\id{c}\to \id{c}\to \id{c}{}\<[E]%
\ColumnHook
\end{hscode}\resethooks
This declares both the \ensuremath{\id{Collects}} class and the type family \ensuremath{\id{Elem}}.  Instances of the \ensuremath{\id{Collects}}
class, must also specify instances of the \ensuremath{\id{Elem}} type family:
\begin{hscode}\SaveRestoreHook
\column{B}{@{}>{\hspre}l<{\hspost}@{}}%
\column{3}{@{}>{\hspre}l<{\hspost}@{}}%
\column{11}{@{}>{\hspre}c<{\hspost}@{}}%
\column{11E}{@{}l@{}}%
\column{14}{@{}>{\hspre}l<{\hspost}@{}}%
\column{18}{@{}>{\hspre}c<{\hspost}@{}}%
\column{18E}{@{}l@{}}%
\column{21}{@{}>{\hspre}l<{\hspost}@{}}%
\column{E}{@{}>{\hspre}l<{\hspost}@{}}%
\>[B]{}\keyword{instance}\;\id{Collects}\;[\mskip1.5mu \id{a}\mskip1.5mu]\;\keyword{where}{}\<[E]%
\\
\>[B]{}\hsindent{3}{}\<[3]%
\>[3]{}\keyword{type}\;\id{Elem}\;[\mskip1.5mu \id{a}\mskip1.5mu]{}\<[18]%
\>[18]{}\mathrel{=}{}\<[18E]%
\>[21]{}\id{a}{}\<[E]%
\\
\>[B]{}\hsindent{3}{}\<[3]%
\>[3]{}\id{empty}{}\<[11]%
\>[11]{}\mathrel{=}{}\<[11E]%
\>[14]{}[\mskip1.5mu \mskip1.5mu]{}\<[E]%
\\
\>[B]{}\hsindent{3}{}\<[3]%
\>[3]{}\id{insert}{}\<[11]%
\>[11]{}\mathrel{=}{}\<[11E]%
\>[14]{}(\mathbin{:}){}\<[E]%
\ColumnHook
\end{hscode}\resethooks

While associated types provide a natural syntactic combination of class and type family definitions,
the class and type family components can actually be specified and formalized entirely
separately~\cite{SchrijversPJCS08}.  Instead of using an associated type synonym, we could
have defined \ensuremath{\id{Elem}} as a distinct top-level entity.
\begin{hscode}\SaveRestoreHook
\column{B}{@{}>{\hspre}l<{\hspost}@{}}%
\column{3}{@{}>{\hspre}l<{\hspost}@{}}%
\column{11}{@{}>{\hspre}c<{\hspost}@{}}%
\column{11E}{@{}l@{}}%
\column{15}{@{}>{\hspre}l<{\hspost}@{}}%
\column{E}{@{}>{\hspre}l<{\hspost}@{}}%
\>[B]{}\keyword{type}\;\keyword{family}\;\id{Elem}\;\id{c}\mathbin{::}\star{}\<[E]%
\\
\>[B]{}\keyword{class}\;\id{Collects}\;\id{c}\;\keyword{where}{}\<[E]%
\\
\>[B]{}\hsindent{3}{}\<[3]%
\>[3]{}\id{empty}{}\<[11]%
\>[11]{}\mathbin{::}{}\<[11E]%
\>[15]{}\id{c}{}\<[E]%
\\
\>[B]{}\hsindent{3}{}\<[3]%
\>[3]{}\id{insert}{}\<[11]%
\>[11]{}\mathbin{::}{}\<[11E]%
\>[15]{}\id{Elem}\;\id{c}\to \id{c}\to \id{c}{}\<[E]%
\ColumnHook
\end{hscode}\resethooks
While there would then be no syntactic requirement to combine instances of the class and type
family, it is easy to see that class instances would be undefinable without corresponding type
family instances, while type family instances would be unusable without corresponding type class
instances.
\begin{hscode}\SaveRestoreHook
\column{B}{@{}>{\hspre}l<{\hspost}@{}}%
\column{3}{@{}>{\hspre}l<{\hspost}@{}}%
\column{11}{@{}>{\hspre}c<{\hspost}@{}}%
\column{11E}{@{}l@{}}%
\column{14}{@{}>{\hspre}l<{\hspost}@{}}%
\column{E}{@{}>{\hspre}l<{\hspost}@{}}%
\>[B]{}\keyword{type}\;\keyword{instance}\;\id{Elem}\;[\mskip1.5mu \id{a}\mskip1.5mu]\mathrel{=}\id{a}{}\<[E]%
\\
\>[B]{}\keyword{instance}\;\id{Collects}\;[\mskip1.5mu \id{a}\mskip1.5mu]\;\keyword{where}{}\<[E]%
\\
\>[B]{}\hsindent{3}{}\<[3]%
\>[3]{}\id{empty}{}\<[11]%
\>[11]{}\mathrel{=}{}\<[11E]%
\>[14]{}[\mskip1.5mu \mskip1.5mu]{}\<[E]%
\\
\>[B]{}\hsindent{3}{}\<[3]%
\>[3]{}\id{insert}{}\<[11]%
\>[11]{}\mathrel{=}{}\<[11E]%
\>[14]{}(\mathbin{:}){}\<[E]%
\ColumnHook
\end{hscode}\resethooks
These definitions are entirely equivalent to the original definitions; while it may be impractical
to use a type family instance \ensuremath{\id{Elem}\;\tau} without a corresponding instance \ensuremath{\id{Collects}\;\tau}, it is not
an error in either approach to do so.

Type families can express many type-level computations.  However, some useful type-level functions
cannot be expressed using open type families.  One is the type family \ensuremath{\id{Equ}\;\id{a}\;\id{b}}, as appeared in the
introduction.  We might hope to characterize \ensuremath{\id{Equ}} using the following equations:
\begin{hscode}\SaveRestoreHook
\column{B}{@{}>{\hspre}l<{\hspost}@{}}%
\column{24}{@{}>{\hspre}c<{\hspost}@{}}%
\column{24E}{@{}l@{}}%
\column{28}{@{}>{\hspre}l<{\hspost}@{}}%
\column{E}{@{}>{\hspre}l<{\hspost}@{}}%
\>[B]{}\keyword{type}\;\keyword{family}\;\id{Equ}\;\id{a}\;\id{b}\mathbin{::}\id{Bool}{}\<[E]%
\\
\>[B]{}\keyword{type}\;\keyword{instance}\;\id{Equ}\;\id{a}\;\id{a}{}\<[24]%
\>[24]{}\mathrel{=}{}\<[24E]%
\>[28]{}\id{True}{}\<[E]%
\\
\>[B]{}\keyword{type}\;\keyword{instance}\;\id{Equ}\;\id{a}\;\id{b}{}\<[24]%
\>[24]{}\mathrel{=}{}\<[24E]%
\>[28]{}\id{False}{}\<[E]%
\ColumnHook
\end{hscode}\resethooks
However, type family instances are interpreted without any ordering, arising either from their
source locations or from their relative generality.  In this case, both equations apply to a type
family application such as \ensuremath{\id{Equ}\;\id{Int}\;\id{Int}} but give different results, and so they are rejected as
inconsistent.  Closed type families~\cite{EisenbergVPJW14} address this problem by allowing ordered
overlap among the instances in a type family definition, so long as the family cannot be further
extended.  We could write the equality function using a closed type family, as we did in the
introduction.  Closed type families cannot be extended later, even if their definitions do not cover
all possible applications.  For example, consider the following definition:
\begin{hscode}\SaveRestoreHook
\column{B}{@{}>{\hspre}l<{\hspost}@{}}%
\column{3}{@{}>{\hspre}l<{\hspost}@{}}%
\column{E}{@{}>{\hspre}l<{\hspost}@{}}%
\>[B]{}\keyword{type}\;\keyword{family}\;\id{OnlyInt}\;\id{a}\mathbin{::}\id{Bool}\;\keyword{where}{}\<[E]%
\\
\>[B]{}\hsindent{3}{}\<[3]%
\>[3]{}\id{OnlyInt}\;\id{Int}\mathrel{=}\id{True}{}\<[E]%
\ColumnHook
\end{hscode}\resethooks
The type family application \ensuremath{\id{OnlyInt}\;\id{Bool}} does not rewrite to any ground type (i.e., type without
type family applications), but the programmer is still prevented from adding further equations to
\ensuremath{\id{OnlyInt}}.

In general, type families need not be injective.  However, there are cases in which it would be
useful to capture the natural injectivity of type-level definitions. For example: session types,
which provide static typing for communication protocols, depend on a naturally injective notion of
duality.  We would expect that if the duals of two session types are equal, then the session types
themselves are equal as well.  Injective type families~\cite{StolarekJE15} can express such cases;
duality could be characterized by the following type family:
\begin{hscode}\SaveRestoreHook
\column{B}{@{}>{\hspre}l<{\hspost}@{}}%
\column{E}{@{}>{\hspre}l<{\hspost}@{}}%
\>[B]{}\keyword{type}\;\keyword{family}\;\id{Dual}\;\id{s}\mathrel{=}\id{r}\mid \id{r}\to \id{s}{}\<[E]%
\ColumnHook
\end{hscode}\resethooks
where the \ensuremath{\id{s}\to \id{r}} annotation denotes the injectivity of duality. The compiler
validates that the injectivity condition is upheld by the type family's defining equations.

The most recent version of
the Glasgow Haskell Compiler, GHC~8.0, accepts all the varieties of type families described above.

\section{The Totality Trap}
\label{sec:partiality}

Recent developments in the theory and implementation of type families~\cite{EisenbergVPJW14,StolarekJE15}
have relied on increasingly technical and confusing
constraints, impeding their use in practice.  In this
section, we argue that these problems arise from a single source: an implicit assumption of totality
for type families.

\subsection{The Assumption of Totality}
\label{sec:assume-totality}

Type families are open and extensible, and so suggesting that they are assumed total seems
counterintuitive.  However, we can rephrase the question as follows.  Suppose that we have a type
family \ensuremath{\id{F}} with no equations.  Is \ensuremath{\id{F}\;\id{Bool}} a type?  It seems absurd that it should be---after all,
the meaning of a type family is given by its equations, and \ensuremath{\id{F}} has no equations. Yet we can
observe that it is:
\begin{hscode}\SaveRestoreHook
\column{B}{@{}>{\hspre}l<{\hspost}@{}}%
\column{E}{@{}>{\hspre}l<{\hspost}@{}}%
\>[B]{}\keyword{type}\;\keyword{family}\;\id{F}\;\id{a}\mathbin{::}\star{}\<[E]%
\\
\>[B]{}\id{f}\;\id{x}\mathrel{=}\id{fst}\;(\id{x},\id{undefined}\mathbin{::}\id{F}\;\id{Bool}){}\<[E]%
\ColumnHook
\end{hscode}\resethooks
This is a well-typed definition: \ensuremath{\id{f}} has type \ensuremath{\id{a}\to \id{a}} and behaves like the identity function.  So
\ensuremath{\id{F}\;\id{Bool}} must be a type, even if all we can observe about it are properties true of all Haskell
types (such as pointedness, or definition of \ensuremath{\id{seq}}).  While it is possible that \ensuremath{\id{F}\;\id{Bool}} will be
defined later in the program, this program would be equally valid were \ensuremath{\id{F}} replaced by a closed type
family, such as \ensuremath{\id{OnlyInt}} (above), in which case we could say with confidence that \ensuremath{\id{OnlyInt}\;\id{Bool}}
would never become defined.

This illustrates that our intuitive understanding of type families is flawed.  Rather than thinking
of type families as partial functions on types, where individual instances extend the definition of
the type family, a type family declaration should be thought of as initially introducing an infinite
family of distinct types, one for each possible application of the type family, and individual
instances as equating previously distinct types.  But does this distinction cause actual problems?
Consider the following definition:
\begin{hscode}\SaveRestoreHook
\column{B}{@{}>{\hspre}l<{\hspost}@{}}%
\column{E}{@{}>{\hspre}l<{\hspost}@{}}%
\>[B]{}\id{g}\;\id{x}\mathrel{=}\id{x}\mathbin{:}\id{x}{}\<[E]%
\ColumnHook
\end{hscode}\resethooks
We might expect this definition to be ill-typed: \ensuremath{\id{x}} must have both type \ensuremath{\tau} and type \ensuremath{[\mskip1.5mu \tau\mskip1.5mu]}, a
seeming contradiction.  But recall type family \ensuremath{\id{Loop}}~\pref{sec:intro}.  If we must assume that
\ensuremath{\id{Loop}} is a type, then it is clearly a satisfying instantiation of the constraint
\ensuremath{\tau\,\sim\,[\mskip1.5mu \tau\mskip1.5mu]}, and so we can assign \ensuremath{\id{g}} the type \ensuremath{\id{Loop}\to \id{Loop}}.  But worse, we expect Haskell
terms to have principal types, and since \ensuremath{\id{g}} makes no reference to \ensuremath{\id{Loop}}, \ensuremath{\id{Loop}\to \id{Loop}} cannot be
its principal type.  Instead, we conclude that \ensuremath{\id{g}} has principal type \ensuremath{\id{a}\,\sim\,[\mskip1.5mu \id{a}\mskip1.5mu]\Rightarrow \id{a}\to \id{a}}.

Now the true consequences of the totality assumption are revealed.  It is not only a gap between the
intuitive and actual meanings of type families, nor just an incompleteness in specifications of type
checking and type inference with type families.  Rather, we are left with a type system which must
accept some (but not all) apparently erroneous definitions: we can reject \ensuremath{\id{Int}\,\sim\,\id{Bool}}, even if we must
accept \ensuremath{\id{a}\,\sim\,[\mskip1.5mu \id{a}\mskip1.5mu]}.  The specification of principal types for this system remains an open question.

It might seem that the problems illustrated here are not to do with the totality assumption itself,
but rather in its interaction with the accepted equations for \ensuremath{\id{Loop}}, and therefore should be fixed
simply by rejecting \ensuremath{\id{Loop}} (and other non-terminating type family definitions).  However, this would
burden the programmer with satisfying some termination checking algorithm, and does not reflect the
realities of either type family practice or research (where significant effort has been devoted to
accounting for non-terminating type families).  Instead, we propose~\secref{predicates} an approach that restores
the intuitive interpretation of type families, preserves their current uses, and avoids introducing
new constraints, such as termination checking.

\subsection{Closed Type Families and the Infinity Problem}
\label{sec:closed-infinite}
\label{sec:old-apartness}

We have seen that assuming totality of type families introduces a variety of theoretical problems.
With the development of closed type families, the totality assumption began causing practical
problems as well, as demonstrated above in the unpleasant interaction between \ensuremath{\id{Equ}} and \ensuremath{\id{Loop}}.

Closed type families rely on a notion of \emph{apartness} to determine when an equation cannot apply
to a particular type family application.  Intuitively, two types are apart if they have no common
instantiations; for example, \ensuremath{\id{Equ}\;\id{Int}\;\id{Bool}} is apart from \ensuremath{\id{Equ}\;\id{a}\;\id{a}}, while \ensuremath{\id{Equ}\;\id{a}\;\id{b}} is not apart
from \ensuremath{\id{Equ}\;\id{a}\;\id{a}}.  This intuition can be formalized in terms of unification: two types are apart if
they have no most general unifier. The problems with \ensuremath{\id{Loop}} arise from the apartness of \ensuremath{\id{Equ}\;\id{a}\;[\mskip1.5mu \id{a}\mskip1.5mu]}
and \ensuremath{\id{Equ}\;\id{a}\;\id{a}}: while these instances do not have any most general unifier in the typical sense,
considering them apart leads to the unsoundness above.  This problem is addressed in closed type
families by defining apartness in terms of \emph{infinitary} unification.  As there is an infinite
(i.e., non-idempotent) unifier of \ensuremath{\id{Equ}\;\id{a}\;[\mskip1.5mu \id{a}\mskip1.5mu]} and \ensuremath{\id{Equ}\;\id{a}\;\id{a}} (namely $\{ [a] / a \}$) \ensuremath{\id{Equ}\;\id{a}\;[\mskip1.5mu \id{a}\mskip1.5mu]}
does not rewrite to \ensuremath{\id{False}} until \ensuremath{\id{a}} is instantiated to some concrete
type. 

While the interaction between closed and infinite type families may seem like a theoretical concern,
the solution causes confusion in practice.  Programmers
discovering that \ensuremath{\id{Equ}\;\id{a}\;[\mskip1.5mu \id{a}\mskip1.5mu]} does not rewrite to \ensuremath{\id{False}} consider this a bug in the implementation
rather than an expected behavior of the type system.\footnote{See, among others:
\begin{itemize}
\item \url{https://ghc.haskell.org/trac/ghc/ticket/9082}: Unexpected behavior involving closed type
  families and repeat occurrences of variables
\item \url{https://ghc.haskell.org/trac/ghc/ticket/9918}: {GHC} chooses an instance between two
  overlapping, but cannot resolve a clause within the similar closed type family
\end{itemize}} It can also result in programs that use closed type families to require more complex type
signatures than similar programs expressed using older techniques, like overlapping
instances~\cite{PeytonJonesJM97} or instance chains~\cite{MorrisJ10}.

\subsection{Explosive Injectivity}
\label{sec:injective}

We have seen that the totality assumption causes both theoretical and practical problems in the
ongoing development of type families.  These problems have depended on the interaction of other type
system features with non-terminating type families.  This might seem like a corner case, and one
that programmers would not expect to encounter in practice.  Recent work on injective type families
bring the problems caused by the totality assumption into starker relief, without relying on
non-terminating type families.

Some families of types are naturally injective; examples include duality relationships
\citep{PucellaT08,LindleyM16b} or the pairing between mutable and immutable vectors types in the
\texttt{vector} library.\footnote{See \url{https://github.com/haskell/vector/issues/34}: Add
  immutable type family.}  However, because type families are not injective in general, expressing
such examples required either the introduction of additional constraints to assure involutiveness or
the use of either proxy arguments or type applications~\cite{EisenbergWA16}
to fix type parameters.  Injective type families~\cite{StolarekJE15}
introduce annotations on type family declarations that characterize their injectivity.  For example,
the duality function for session types could be declared by
\begin{hscode}\SaveRestoreHook
\column{B}{@{}>{\hspre}l<{\hspost}@{}}%
\column{E}{@{}>{\hspre}l<{\hspost}@{}}%
\>[B]{}\keyword{type}\;\keyword{family}\;\id{Dual}\;\id{s}\mathrel{=}\id{r}\mid \id{r}\to \id{s}{}\<[E]%
\ColumnHook
\end{hscode}\resethooks
This declaration differs from traditional type family declarations in two ways: first, the result is
named (\ensuremath{\id{r}}), and second, the annotation \ensuremath{\id{r}\to \id{s}} specifies \ensuremath{\id{Dual}}'s injectivity: its result
determines its argument.

Unfortunately, injective type families require seemingly arcane restrictions to preserve type
safety.  For example, consider the following apparently innocuous definitions:
\begin{hscode}\SaveRestoreHook
\column{B}{@{}>{\hspre}l<{\hspost}@{}}%
\column{E}{@{}>{\hspre}l<{\hspost}@{}}%
\>[B]{}\keyword{type}\;\keyword{family}\;\id{ListElems}\;\id{a}\mathrel{=}\id{b}\mid \id{b}\to \id{a}{}\<[E]%
\\
\>[B]{}\keyword{type}\;\keyword{instance}\;\id{ListElems}\;[\mskip1.5mu \id{a}\mskip1.5mu]\mathrel{=}\id{a}{}\<[E]%
\ColumnHook
\end{hscode}\resethooks
\ensuremath{\id{ListElems}} is clearly injective: if $\ensuremath{\id{a}} \sim \ensuremath{\id{b}}$ then $\ensuremath{[\mskip1.5mu \id{a}\mskip1.5mu]} \sim \ensuremath{[\mskip1.5mu \id{b}\mskip1.5mu]}$.  Nevertheless, this
example is sufficient to derive a violation of type safety: by the definition of \ensuremath{\id{ListElems}}, we have that
$\ensuremath{\id{ListElems}\;[\mskip1.5mu \id{ListElems}\;\id{Int}\mskip1.5mu]}$ 
$\sim \ensuremath{\id{ListElems}\;\id{Int}}$, and then by injectivity, we have that
$\ensuremath{[\mskip1.5mu \id{ListElems}\;\id{Int}\mskip1.5mu]} \sim \ensuremath{\id{Int}}$, an impossibility.  In the previous sections, difficulties stemmed
from the type family application \ensuremath{\id{Loop}}, which does not correspond to any ground type.  In this
case, problems arise from the type family application \ensuremath{\id{ListElems}\;\id{Int}}, which similarly cannot
correspond to any ground type.  Suppose that we could prove that $\ensuremath{\id{ListElems}\;\id{Int}} \sim \tau$ for some
type $\tau$; as, from the definition of \ensuremath{\id{ListElems}} we also have that $\ensuremath{\id{ListElems}\;[\mskip1.5mu \tau\mskip1.5mu]} \sim \tau$,
the injectivity of \ensuremath{\id{ListElems}} has been violated.

Definitions like that of \ensuremath{\id{ListElems}} are ruled out by strict restrictions on the right-hand sides of
injective type family equations; for example, the RHS of an injective type family equation cannot
(in most cases) be a type variable or another type family application.  These restrictions are
necessary to assure the safety of injective type families, but have not yet been shown to be
sufficient.  They are also a significant limitation in expressiveness, especially in comparison with
older approaches, such as functional dependencies~\cite{Jones00}.

\section{Constraining Type Families}
\label{sec:predicates}

In the previous section, we have seen that indexed type families are implicitly assumed to be
defined at all their applications---that is, they represent total functions on types.  We have
seen how this totality assumption introduces practical and theoretical obstacles, both in preserving
totality (such as in injective type families) or in accounting for its violations (such as in the
interaction between non-terminating and closed type families).

We propose a new approach, \emph{constrained type families}, which treats type families as partial
maps between types.  Our key observation is that Haskell already supplies a mechanism to limit the
domain of polymorphism: qualified types with type classes. So we can capture partiality by
associating each type family with a type class that characterizes its domain of definition.  We will
show that this approach naturally resolves the practical and theoretical issues with type families
and restores their intuitive meaning, while adding little new complexity for programmer or
implementer.

This section describes constrained open type families; we discuss the extension of our approach to
closed type families in the following section.

\subsection{Constrained Type Families}

Our goal is a system of partial type families that sacrifices neither the expressiveness nor the
ease of use of present type families.  This introduces two challenges. First, we must retain the
applicative syntax of type families, while taking their domains of definition into account.  That
is, a type family application such as \ensuremath{\id{F}\;\tau} should be constrained by the domain of \ensuremath{\id{F}}.  In
particular, whether a type family application that contains type variables, such as \ensuremath{\id{F}\;\id{a}}, is
well-defined depends on the instantiation of the type variable \ensuremath{\id{a}}.  Second, we must keep
type families easy to define, while simultaneously characterizing their domains of definition.

We address each of these problems using features already present in modern Haskell.  Haskell already
has a mechanism suited to capturing this kind of constrained polymorphism: qualified types and type
classes~\cite{WadlerBlott89}.  Qualified types are currently used to track when type class methods are defined; for
example, the equality operator is defined at all types \ensuremath{\id{a}\to \id{a}\to \id{Bool}} such that the class
predicate \ensuremath{\id{Eq}\;\id{a}} is satisfiable.  Our intention is to reuse the qualified types mechanism to account
for partiality in type families as well.  Haskell also supports a mechanism that combines type
classes and type families: associated type synonyms.  Our intention is to rely on associated types
to simultaneously define type families and characterize their domains.

We propose combining these mechanisms to give an account of partial type families that matches both
the intuitions and usage of type families in Haskell today.  In doing so, we make two changes to the
surface language.  First, we require that type families be defined by associated types, disallowing
free-standing type family declarations.  This means that the well-definedness of type family
instances follows from the satisfiability of the corresponding class predicates.  In our previous
example~\secref{background}, the type family application \ensuremath{\id{Elem}\;\tau} is defined precisely when the
predicate \ensuremath{\id{Collects}\;\tau} is satisfiable.  Second, we require that all uses of type families be
well-defined, as enforced by their corresponding class predicates.  That is, uses of the type family
\ensuremath{\id{Elem}\;\tau} must occur in a context that is sufficient to prove \ensuremath{\id{Collects}\;\tau} (either because
\ensuremath{\id{Collects}\;\tau} is assumed or provable from the instances).

Our approach captures the natural interpretation and use of open type families.  Open type families
are already primarily useful in combination with type class constraints---we have no way to use a
value of type \ensuremath{\id{Elem}\;\tau} unless we have some additional information about that type, captured by the
class constraint \ensuremath{\id{Collects}\;\tau}.  Thus, our requirements do not reduce the expressiveness of the
language.  The remainder of the section demonstrates informally that our approach addresses the
difficulties and confusion with type families.

We begin with the behavior of undefined, or ``stuck'', type family
instances~\secref{assume-totality}.  As before, We define a type family, \ensuremath{\id{F2}}, now associated with a
class \ensuremath{\id{C2}}:
\begin{hscode}\SaveRestoreHook
\column{B}{@{}>{\hspre}l<{\hspost}@{}}%
\column{3}{@{}>{\hspre}l<{\hspost}@{}}%
\column{E}{@{}>{\hspre}l<{\hspost}@{}}%
\>[B]{}\keyword{class}\;\id{C2}\;\id{t}\;\keyword{where}{}\<[E]%
\\
\>[B]{}\hsindent{3}{}\<[3]%
\>[3]{}\keyword{type}\;\id{F2}\;\id{t}\mathbin{::}\star{}\<[E]%
\ColumnHook
\end{hscode}\resethooks
Instances of the \ensuremath{\id{F2}} type family can be added only by adding instances to the \ensuremath{\id{C2}} class:
\begin{hscode}\SaveRestoreHook
\column{B}{@{}>{\hspre}l<{\hspost}@{}}%
\column{3}{@{}>{\hspre}l<{\hspost}@{}}%
\column{E}{@{}>{\hspre}l<{\hspost}@{}}%
\>[B]{}\keyword{instance}\;\id{C2}\;\id{Int}\;\keyword{where}{}\<[E]%
\\
\>[B]{}\hsindent{3}{}\<[3]%
\>[3]{}\keyword{type}\;\id{F2}\;\id{Int}\mathrel{=}\id{Bool}{}\<[E]%
\ColumnHook
\end{hscode}\resethooks
Now, recall our function definition:
\begin{hscode}\SaveRestoreHook
\column{B}{@{}>{\hspre}l<{\hspost}@{}}%
\column{E}{@{}>{\hspre}l<{\hspost}@{}}%
\>[B]{}\id{f}\;\id{x}\mathrel{=}\id{fst}\;(\id{x},\id{undefined}\mathbin{::}\id{F2}\;\id{Bool}){}\<[E]%
\ColumnHook
\end{hscode}\resethooks
Is this definition still well-typed?  The use of \ensuremath{\id{F2}\;\id{Bool}} requires that \ensuremath{\id{C2}\;\id{Bool}} be satisfiable to
assure that it is well defined.  However, without any instances of \ensuremath{\id{C2}\;\id{Bool}} in scope, the
constraint would be unsatisfiable, so the definition would be rejected.  This account extends
naturally to polymorphism.  Suppose that we had some function \ensuremath{\id{g}} that used \ensuremath{\id{F2}}, with the following
type:
\begin{hscode}\SaveRestoreHook
\column{B}{@{}>{\hspre}l<{\hspost}@{}}%
\column{E}{@{}>{\hspre}l<{\hspost}@{}}%
\>[B]{}\id{g}\mathbin{::}\id{C2}\;\id{a}\Rightarrow \id{a}\to \id{F2}\;\id{a}{}\<[E]%
\ColumnHook
\end{hscode}\resethooks
(Note the requisite \ensuremath{\id{C2}\;\id{a}} constraint.)  Now, we could define an alternative version of \ensuremath{\id{f}} as
follows:
\begin{hscode}\SaveRestoreHook
\column{B}{@{}>{\hspre}l<{\hspost}@{}}%
\column{E}{@{}>{\hspre}l<{\hspost}@{}}%
\>[B]{}\id{f'}\;\id{x}\mathrel{=}\id{fst}\;(\id{x},\id{g}\;\id{x}){}\<[E]%
\ColumnHook
\end{hscode}\resethooks
The definition of \ensuremath{\id{f'}} is not ill-typed, but its type, \ensuremath{\id{C2}\;\id{a}\Rightarrow \id{a}\to \id{a}}, includes the \ensuremath{\id{C2}\;\id{a}}
constraint to assures that the type of \ensuremath{\id{g}\;\id{x}} is well-defined.

The complications with closed type families arose from their interaction with non-terminating type
families.  We can already see how non-terminating type family definitions would play out in our
system.  As before, we define a type family \ensuremath{\id{Loop}}, but now as an associated type to a type class
\ensuremath{\id{Loopy}}:
\begin{hscode}\SaveRestoreHook
\column{B}{@{}>{\hspre}l<{\hspost}@{}}%
\column{3}{@{}>{\hspre}l<{\hspost}@{}}%
\column{E}{@{}>{\hspre}l<{\hspost}@{}}%
\>[B]{}\keyword{class}\;\id{Loopy}\;\keyword{where}{}\<[E]%
\\
\>[B]{}\hsindent{3}{}\<[3]%
\>[3]{}\keyword{type}\;\id{Loop}\mathbin{::}\star{}\<[E]%
\ColumnHook
\end{hscode}\resethooks
As \ensuremath{\id{Loop}} is a 0-ary type family, \ensuremath{\id{Loopy}} is a 0-ary type class.  This is not problematic; in
particular, there are two canonical 0-ary type classes, one whose predicates are trivially true and
another whose predicates are unsatisfiable.  Now, suppose we want to add the equation
$\ensuremath{\id{Loop}} \sim \ensuremath{[\mskip1.5mu \id{Loop}\mskip1.5mu]}$.  We would need to do so via an instance of \ensuremath{\id{Loopy}}.  However, we cannot add
the instance
\begin{hscode}\SaveRestoreHook
\column{B}{@{}>{\hspre}l<{\hspost}@{}}%
\column{3}{@{}>{\hspre}l<{\hspost}@{}}%
\column{E}{@{}>{\hspre}l<{\hspost}@{}}%
\>[B]{}\keyword{instance}\;\id{Loopy}\;\keyword{where}{}\<[E]%
\\
\>[B]{}\hsindent{3}{}\<[3]%
\>[3]{}\keyword{type}\;\id{Loop}\mathrel{=}[\mskip1.5mu \id{Loop}\mskip1.5mu]{}\<[E]%
\ColumnHook
\end{hscode}\resethooks
as the use of \ensuremath{\id{Loop}} on the right-hand side of the type definition does not have a corresponding
constraint.  We can add the instance
\begin{hscode}\SaveRestoreHook
\column{B}{@{}>{\hspre}l<{\hspost}@{}}%
\column{4}{@{}>{\hspre}l<{\hspost}@{}}%
\column{E}{@{}>{\hspre}l<{\hspost}@{}}%
\>[B]{}\keyword{instance}\;\id{Loopy}\Rightarrow \id{Loopy}\;\keyword{where}{}\<[E]%
\\
\>[B]{}\hsindent{4}{}\<[4]%
\>[4]{}\keyword{type}\;\id{Loop}\mathrel{=}[\mskip1.5mu \id{Loop}\mskip1.5mu]{}\<[E]%
\ColumnHook
\end{hscode}\resethooks
but it is clear that the \ensuremath{\id{Loopy}} constraint cannot be satisfied.  Thus, any attempt to use this
\ensuremath{\id{Loop}} equation must be guarded by an unsatisfiable \ensuremath{\id{Loopy}} constraint, and so cannot compromise
type safety.

%

Finally, we can give an informal description of constrained injective type families.  We return to
the \ensuremath{\id{ListElems}} example, now defining it by an associated type synonym:
\begin{hscode}\SaveRestoreHook
\column{B}{@{}>{\hspre}l<{\hspost}@{}}%
\column{4}{@{}>{\hspre}l<{\hspost}@{}}%
\column{E}{@{}>{\hspre}l<{\hspost}@{}}%
\>[B]{}\keyword{class}\;\id{Listy}\;\id{t}\;\keyword{where}{}\<[E]%
\\
\>[B]{}\hsindent{4}{}\<[4]%
\>[4]{}\keyword{type}\;\id{ListElems}\;\id{t}\mathrel{=}\id{u}\mid \id{u}\to \id{t}{}\<[E]%
\\[\blanklineskip]%
\>[B]{}\keyword{instance}\;\id{Listy}\;[\mskip1.5mu \id{t}\mskip1.5mu]\;\keyword{where}{}\<[E]%
\\
\>[B]{}\hsindent{4}{}\<[4]%
\>[4]{}\keyword{type}\;\id{ListElems}\;[\mskip1.5mu \id{t}\mskip1.5mu]\mathrel{=}\id{t}{}\<[E]%
\ColumnHook
\end{hscode}\resethooks
Notice that we could not add an instance \ensuremath{\id{Listy}\;\id{Int}}, as that would require adding a corresponding
instance to the type family and any such instance would be rejected for violating the injectivity
constraint of \ensuremath{\id{ListElems}}.  Consequently, inconsistencies arising from uses of the type family
application \ensuremath{\id{ListElems}\;\id{Int}} must be guarded by the unsatisfiable class constraint \ensuremath{\id{Listy}\;\id{Int}}.

Constrained type families are not, in their simplest form, backward compatible.  We will return to
the question of compatibility with existing Haskell programs, and show how we can infer the
requisite constraints to transition from current usage to the explicit use of constrained type
families~\pref{sec:compat}.

\subsection{Validating Constrained Type Families}
\label{sec:valid-ctf}


In the previous section, we introduced an intuitive characterization of constrained type families.
Later~\secref{formal}, we will formalize a core calculus with constrained type families.  However,
our formalization will differ from Haskell-like surface languages in several significant ways. This
section bridges the intuition of constrained type families and our core language, in the context of
a simple, Haskell-like type system.

\begin{figure}
\begin{gather*}
\ottdruleSTXXVar{}
\quad
\ottdruleSTXXForall{}
\quad
\ottdruleSTXXQual{}
\\
\ottdruleSTXXTyCon{}
\qquad
\ottdruleSTXXFamily{}
\end{gather*}
\caption{Well-formedness rules for types}
\label{fig:wf-types}
\end{figure}

\figref{wf-types} gives the syntax and formation rules for our surface type system.  We omit kinds
from our account, as they are an orthogonal concern from the use of type classes and type families.
Our well-formedness judgment takes the form $ P  \mid  \Gamma   \vdash   \sigma  \ \mathsf{type} $, in which $\sigma$ is a
surface-language type, $\Gamma$ is a type variable environment, and $P$ is a predicate
context.  As we have omitted kinds, the environment $\Gamma$ is simply a list of type variables.  The
form of the judgment and use of context $P$ should be familiar from other formulations of
qualified types~\cite{Jones94}.

Our types include type variables ($\alpha$), quantified types ($\forall \, \alpha  \ottsym{.}  \tau$), qualified types
($\pi  \Rightarrow  \tau$), and applications of type constructors ($\ottnt{H} \, \overline{\tau}$) and type families ($\ottmv{F} \, \overline{\tau}$).
The rules for variables, quantifiers, and qualifiers should all be unsurprising.  Leaf nodes depend
on an auxiliary well-formedness judgment $ \vdash   P  \mid  \Gamma  \ \mathsf{ctx} $ on contexts, which is entirely
unsurprising.  Our treatments of type constructors and type families depend on an ambient signature
$\Sigma$, representing the top-level declarations.  Arity $n$ type constructors are captured by
entries $ (  \ottnt{H}  :  \ottmv{n}  ) \in  \Sigma $; the typing rule for constructors assures that they have the correct
number of arguments.  The interesting case is for type families.  Constrained type families are
represented by assertions $ ( C   \Rightarrow   \ottmv{F}  :  \ottmv{n} ) \in  \Sigma $; this denotes that type family $\ottmv{F}$ has arity
$n$, and is associated with class $C$.  Uses of the type family application \ensuremath{\id{F}\;\overline\tau}, then, should
occur in a context strong enough to prove \ensuremath{\id{C}\;\overline\tau}.  This is captured by \rul{ST\_Family}, in which
we insist that the context $P$ is strong enough to prove $C \, \overline{\tau}$; we omit the details of
the standard type class entailment relation $\cdot \Vdash \cdot$.  For a simple example, suppose
that \ensuremath{\id{F}} is a unary type family declared in class \ensuremath{\id{C}}, and class \ensuremath{\id{D}} is a subclass of \ensuremath{\id{C}}.  Then
we could prove any of the following judgments:
\[
  \ensuremath{\id{C}\;\tau} \mid \emptyset \vdash \ensuremath{\id{F}\;\tau} \ \mathsf{type} \qquad
  \ensuremath{\id{D}\;\tau} \mid \emptyset \vdash \ensuremath{\id{F}\;\tau} \ \mathsf{type} \qquad
  \emptyset \mid \emptyset \vdash \ensuremath{\id{C}\;\tau\Rightarrow \id{F}\;\tau} \ \mathsf{type}
\]
but, absent other instances of \ensuremath{\id{C}}, we could not prove $\emptyset \mid \emptyset \vdash \ensuremath{\id{F}\;\tau} \ \mathsf{type}$.




\section{Achieving Closure}
\label{sec:closed}

Closed type families are one of the most fruitful extensions of indexed type families.  They allow
type families to be specified by ordered sequences of overlapping equations, capturing
many patterns of type-level computation that were previously inexpressible or required intricate
indirect encodings.  In this section, we discuss the extension of constrained type families to
include closed type families.  This introduces two challenges.  First, there is no
existing feature of type classes that mirrors closed type families.  We introduce closed type
classes, a simplification of instance chains~\cite{MorrisJ10}, and show how they can be used
to constrain closed type families.  Second, closed type families may be total, and so
could be used without constraints.  We discuss approaches to recognizing and supporting total closed
type families.  Finally, we illustrate the simplification our approach provides over previous formulations of closed type families.

\subsection{Closed Type Classes}
\label{sec:closed-classes}

\emph{Closed type classes} are our novel approach to introducing and resolving overlap among
class instances.  They closely follow the design of closed type families: just as
closed type families allow type families to be defined by ordered sequences of overlapping equations,
closed type classes allow type classes to be defined by ordered sequences of overlapping instances.  Instance
resolution begins with the first instance in the sequence, and proceeds to subsequent instances only
if the first instance cannot match the goal predicate.  In the next section, we will show that
closed type classes can characterize the domain of definition of closed type families.  We begin,
however, by considering closed type classes as a feature on their own.

As an example, we consider heterogeneous lists, following the approach of
\citet{KiselyovLS04}.  We begin by introducing data types to represent heterogeneous lists:
\begin{hscode}\SaveRestoreHook
\column{B}{@{}>{\hspre}l<{\hspost}@{}}%
\column{E}{@{}>{\hspre}l<{\hspost}@{}}%
\>[B]{}\keyword{data}\;\id{HNil}\mathrel{=}\id{MkHNil}{}\<[E]%
\\
\>[B]{}\keyword{data}\;\id{HCons}\;\id{e}\;\id{l}\mathrel{=}\id{MkHCons}\;\id{e}\;\id{l}{}\<[E]%
\ColumnHook
\end{hscode}\resethooks
For example, the declaration
\begin{hscode}\SaveRestoreHook
\column{B}{@{}>{\hspre}l<{\hspost}@{}}%
\column{E}{@{}>{\hspre}l<{\hspost}@{}}%
\>[B]{}\id{hlst}\mathrel{=}\id{MkHCons}\;\id{True}\;(\id{MkHCons}\;\text{\tt 'c'}\;\id{MkHNil}){}\<[E]%
\ColumnHook
\end{hscode}\resethooks
defines a heterogeneous list \ensuremath{\id{hlst}} with type \ensuremath{\id{HCons}\;\id{Bool}\;(\id{HCons}\;\id{Char}\;\id{HNil})}.  Kiselyov et
al. describe a number of operations on heterogeneous lists, and show how they can be used to build
more complex data structures, such as extensible records.  We will limit ourselves to some of the
simpler operations.  One such operation is \ensuremath{\id{hOccurs}}, which projects all elements of a given type
from a heterogeneous list.  We can define \ensuremath{\id{hOccurs}} using a closed type class as follows:
\begin{hscode}\SaveRestoreHook
\column{B}{@{}>{\hspre}l<{\hspost}@{}}%
\column{3}{@{}>{\hspre}l<{\hspost}@{}}%
\column{5}{@{}>{\hspre}l<{\hspost}@{}}%
\column{E}{@{}>{\hspre}l<{\hspost}@{}}%
\>[B]{}\keyword{class}\;\id{HOccurs}\;\id{e}\;\id{l}\;\keyword{where}{}\<[E]%
\\
\>[B]{}\hsindent{3}{}\<[3]%
\>[3]{}\id{hOccurs}\mathbin{::}\id{l}\to [\mskip1.5mu \id{e}\mskip1.5mu]{}\<[E]%
\\[\blanklineskip]%
\>[B]{}\hsindent{3}{}\<[3]%
\>[3]{}\keyword{instance}\;\id{HOccurs}\;\id{e}\;\id{HNil}\;\keyword{where}{}\<[E]%
\\
\>[3]{}\hsindent{2}{}\<[5]%
\>[5]{}\id{hOccurs}\;\id{MkHNil}\mathrel{=}[\mskip1.5mu \mskip1.5mu]{}\<[E]%
\\[\blanklineskip]%
\>[B]{}\hsindent{3}{}\<[3]%
\>[3]{}\keyword{instance}\;\id{HOccurs}\;\id{e}\;\id{l}\Rightarrow \id{HOccurs}\;\id{e}\;(\id{HCons}\;\id{e}\;\id{l})\;\keyword{where}{}\<[E]%
\\
\>[3]{}\hsindent{2}{}\<[5]%
\>[5]{}\id{hOccurs}\;(\id{MkHCons}\;\id{e}\;\id{l})\mathrel{=}\id{e}\mathbin{:}\id{hOccurs}\;\id{l}{}\<[E]%
\\[\blanklineskip]%
\>[B]{}\hsindent{3}{}\<[3]%
\>[3]{}\keyword{instance}\;\id{HOccurs}\;\id{e}\;\id{l}\Rightarrow \id{HOccurs}\;\id{e}\;(\id{HCons}\;\id{e'}\;\id{l})\;\keyword{where}{}\<[E]%
\\
\>[3]{}\hsindent{2}{}\<[5]%
\>[5]{}\id{hOccurs}\;(\id{MkHCons}\;\anonymous \;\id{l})\mathrel{=}\id{hOccurs}\;\id{l}{}\<[E]%
\ColumnHook
\end{hscode}\resethooks
\ensuremath{\id{HOccurs}} is a closed type class, as indicated by the sequence of instances inside the class
declaration.  The second two instances are overlapping---for example, both apply to the predicate
\ensuremath{\id{HOccurs}\;\id{Char}\;(\id{HCons}\;\id{Char}\;\id{HNil})}---but the ordering indicates that the first instance should apply
in the common cases.  Depending on its expected return type, \ensuremath{\id{hOccurs}\;\id{hlst}} could evaluate to
\ensuremath{[\mskip1.5mu \id{True}\mskip1.5mu]}, \ensuremath{[\mskip1.5mu \text{\tt 'c'}\mskip1.5mu]}, or \ensuremath{[\mskip1.5mu \mskip1.5mu]}.

Closed type classes bear a close resemblance to overlapping instances~\cite{PeytonJonesJM97}, a
well-established extension of the Haskell class system.  However, whereas the order of instances in
closed type families is explicit in their declaration, overlapping instances have an implicit
ordering, fixed by the compiler.  This means that overlapping instances can lead to unintended
ambiguity.  For example, in \citeauthor{Swierstra08}'s~\citeyear{Swierstra08} encoding of extensible variants, he
relies on a data type of functor coproducts:
\begin{hscode}\SaveRestoreHook
\column{B}{@{}>{\hspre}l<{\hspost}@{}}%
\column{E}{@{}>{\hspre}l<{\hspost}@{}}%
\>[B]{}\keyword{data}\;(\id{f}\mathop{\oplus}\id{g})\;\id{e}\mathrel{=}\id{Inl}\;(\id{f}\;\id{e})\mid \id{Inr}\;(\id{g}\;\id{e}){}\<[E]%
\ColumnHook
\end{hscode}\resethooks
He defines a class of polymorphic injectors, as follows:
\begin{hscode}\SaveRestoreHook
\column{B}{@{}>{\hspre}l<{\hspost}@{}}%
\column{3}{@{}>{\hspre}l<{\hspost}@{}}%
\column{E}{@{}>{\hspre}l<{\hspost}@{}}%
\>[B]{}\keyword{class}\;\id{f}\mathop{\preceq}\id{g}\;\keyword{where}{}\<[E]%
\\
\>[B]{}\hsindent{3}{}\<[3]%
\>[3]{}\id{inj}\mathbin{::}\id{f}\;\id{e}\to \id{g}\;\id{e}{}\<[E]%
\\[\blanklineskip]%
\>[B]{}\keyword{instance}\;\id{f}\mathop{\preceq}\id{f}\;\keyword{where}{}\<[E]%
\\
\>[B]{}\hsindent{3}{}\<[3]%
\>[3]{}\id{inj}\mathrel{=}\id{id}{}\<[E]%
\\[\blanklineskip]%
\>[B]{}\keyword{instance}\;\id{f}\mathop{\preceq}(\id{f}\mathop{\oplus}\id{g})\;\keyword{where}{}\<[E]%
\\
\>[B]{}\hsindent{3}{}\<[3]%
\>[3]{}\id{inj}\mathrel{=}\id{Inl}{}\<[E]%
\\[\blanklineskip]%
\>[B]{}\keyword{instance}\;\id{f}\mathop{\preceq}\id{h}\Rightarrow \id{f}\mathop{\preceq}(\id{g}\mathop{\oplus}\id{h})\;\keyword{where}{}\<[E]%
\\
\>[B]{}\hsindent{3}{}\<[3]%
\>[3]{}\id{inj}\mathrel{=}\id{Inr}\hsdot{\circ }{.\,}\id{inj}{}\<[E]%
\ColumnHook
\end{hscode}\resethooks
The intuition here is simple: these instances describe a recursive search of (right-grouped)
coproduct types, in which the first two instances provide base cases and the third instance provides
the recursive case.  However, there is actually an unresolved overlap among the instances: the
predicate \ensuremath{(\id{f}\mathop{\oplus}\id{g})\mathop{\preceq}(\id{f}\mathop{\oplus}\id{g})} could be resolved by either the first or third instance, and
neither is more specific than the other.  Consequently, GHC will report an error if such a predicate
is encountered.  An implementation of this class using closed type class (written simply by
indenting the \ensuremath{\keyword{instance}} declarations to fit within the \ensuremath{\keyword{class}} body) would be unambiguous, and the
predicate \ensuremath{(\id{f}\mathop{\preceq}\id{g})\mathop{\preceq}(\id{f}\mathop{\oplus}\id{g})} would be resolved using the first instance.

\subsection{Constrained Closed Type Families}

Combining closed type classes and associated types gives us a way to introduce closed type families
while accurately characterizing their domains of definition.

For an example, we turn again to the heterogeneous lists of \citet{KiselyovLS04}.  Our
new goal is to define an operation \ensuremath{\id{hDelete}}, which will remove all values of a given type from a
heterogeneous list.  In doing so, we must simultaneously define a mapping on types describing the
type of the resulting list.  We do this by defining an associated type \ensuremath{\id{HWithout}} such that, if \ensuremath{\id{l}}
is a heterogeneous list type, then \ensuremath{\id{HWithout}\;\id{e}\;\id{l}} is the same list without any occurrences of
element type \ensuremath{\id{e}}.  Thus, we arrive at the following closed type class definition.
\begin{hscode}\SaveRestoreHook
\column{B}{@{}>{\hspre}l<{\hspost}@{}}%
\column{3}{@{}>{\hspre}l<{\hspost}@{}}%
\column{5}{@{}>{\hspre}l<{\hspost}@{}}%
\column{E}{@{}>{\hspre}l<{\hspost}@{}}%
\>[B]{}\keyword{class}\;\id{HDelete}\;\id{e}\;\id{l}\;\keyword{where}{}\<[E]%
\\
\>[B]{}\hsindent{3}{}\<[3]%
\>[3]{}\keyword{type}\;\id{HWithout}\;\id{e}\;\id{l}\mathbin{::}\star{}\<[E]%
\\
\>[B]{}\hsindent{3}{}\<[3]%
\>[3]{}\id{hDelete}\mathbin{::}\id{Proxy}\;\id{e}\to \id{l}\to \id{HWithout}\;\id{e}\;\id{l}{}\<[E]%
\\[\blanklineskip]%
\>[B]{}\hsindent{3}{}\<[3]%
\>[3]{}\keyword{instance}\;\id{HDelete}\;\id{e}\;\id{HNil}\;\keyword{where}{}\<[E]%
\\
\>[3]{}\hsindent{2}{}\<[5]%
\>[5]{}\keyword{type}\;\id{HWithout}\;\id{e}\;\id{HNil}\mathrel{=}\id{HNil}{}\<[E]%
\\
\>[3]{}\hsindent{2}{}\<[5]%
\>[5]{}\id{hDelete}\;\anonymous \;\id{MkHNil}\mathrel{=}\id{MkHNil}{}\<[E]%
\\[\blanklineskip]%
\>[B]{}\hsindent{3}{}\<[3]%
\>[3]{}\keyword{instance}\;\id{HDelete}\;\id{e}\;\id{l}\Rightarrow \id{HDelete}\;\id{e}\;(\id{HCons}\;\id{e}\;\id{l})\;\keyword{where}{}\<[E]%
\\
\>[3]{}\hsindent{2}{}\<[5]%
\>[5]{}\keyword{type}\;\id{HWithout}\;\id{e}\;(\id{HCons}\;\id{e}\;\id{l})\mathrel{=}\id{HWithout}\;\id{e}\;\id{l}{}\<[E]%
\\
\>[3]{}\hsindent{2}{}\<[5]%
\>[5]{}\id{hDelete}\;\id{ep}\;(\id{MkHCons}\;\anonymous \;\id{l})\mathrel{=}\id{hDelete}\;\id{ep}\;\id{l}{}\<[E]%
\\[\blanklineskip]%
\>[B]{}\hsindent{3}{}\<[3]%
\>[3]{}\keyword{instance}\;\id{HDelete}\;\id{e}\;\id{l}\Rightarrow \id{HDelete}\;\id{e}\;(\id{HCons}\;\id{e'}\;\id{l})\;\keyword{where}{}\<[E]%
\\
\>[3]{}\hsindent{2}{}\<[5]%
\>[5]{}\keyword{type}\;\id{HWithout}\;\id{e}\;(\id{HCons}\;\id{e'}\;\id{l})\mathrel{=}\id{HCons}\;\id{e'}\;(\id{HWithout}\;\id{e}\;\id{l}){}\<[E]%
\\
\>[3]{}\hsindent{2}{}\<[5]%
\>[5]{}\id{hDelete}\;\id{ep}\;(\id{MkHCons}\;\id{e'}\;\id{l})\mathrel{=}\id{MkHCons}\;\id{e'}\;(\id{hDelete}\;\id{ep}\;\id{l}){}\<[E]%
\ColumnHook
\end{hscode}\resethooks
The class \ensuremath{\id{HDelete}\;\id{e}\;\id{l}} has the \ensuremath{\id{hDelete}} method and the \ensuremath{\id{HWithout}} associated type synonym; to
disambiguate the type of \ensuremath{\id{hDelete}}, we capture the type \ensuremath{\id{e}} using a \ensuremath{\id{Proxy}} argument.  The \ensuremath{\id{HDelete}}
class has three instances, following the same recursion scheme we used for \ensuremath{\id{HOccurs}}; again, the
final two instances overlap.  Like conventional closed type families, the associated type synonym
equations are checked in the order in which they appear in the type class definition.  For example,
we have that \ensuremath{\id{HWithout}\;\id{Char}\;(\id{HCons}\;\id{Bool}\;(\id{HCons}\;\id{Char}\;\id{HNil}))\,\sim\,\id{HCons}\;\id{Bool}\;\id{HNil}}.  Note that \ensuremath{\id{HWithout}}
is not total: while it is defined for arbitrary \ensuremath{\id{e}}, it is only defined for \ensuremath{\id{l}} that are properly
formed heterogeneous list types.

\subsection{Closed Type Families and Totality}
\label{sec:totality-checking}
\label{sec:vector-append}

Unlike open type families, closed type families can be total.\footnote{Open type families might also
be total, with the right equations. Any such open type family can, however, be written as a closed family.
We thus consider all open type families to be partial.}  For example, we
could implement addition for type-level naturals using constrained closed type classes as follows:
\begin{hscode}\SaveRestoreHook
\column{B}{@{}>{\hspre}l<{\hspost}@{}}%
\column{3}{@{}>{\hspre}l<{\hspost}@{}}%
\column{6}{@{}>{\hspre}l<{\hspost}@{}}%
\column{E}{@{}>{\hspre}l<{\hspost}@{}}%
\>[B]{}\keyword{data}\;\id{Nat}\mathrel{=}\id{Z}\mid \id{S}\;\id{Nat}{}\<[E]%
\\
\>[B]{}\keyword{class}\;\id{PlusC}\;(\id{m}\mathbin{::}\id{Nat})\;(\id{n}\mathbin{::}\id{Nat})\;\keyword{where}{}\<[E]%
\\
\>[B]{}\hsindent{3}{}\<[3]%
\>[3]{}\keyword{type}\;\id{Plus}\;\id{m}\;\id{n}\mathbin{::}\star{}\<[E]%
\\[\blanklineskip]%
\>[B]{}\hsindent{3}{}\<[3]%
\>[3]{}\keyword{instance}\;\id{PlusC}\;\id{Z}\;\id{n}\;\keyword{where}{}\<[E]%
\\
\>[3]{}\hsindent{3}{}\<[6]%
\>[6]{}\keyword{type}\;\id{Plus}\;\id{Z}\;\id{n}\mathrel{=}\id{n}{}\<[E]%
\\[\blanklineskip]%
\>[B]{}\hsindent{3}{}\<[3]%
\>[3]{}\keyword{instance}\;\id{PlusC}\;\id{m}\;\id{n}\Rightarrow \id{PlusC}\;(\id{S}\;\id{m})\;\id{n}\;\keyword{where}{}\<[E]%
\\
\>[3]{}\hsindent{3}{}\<[6]%
\>[6]{}\keyword{type}\;\id{Plus}\;(\id{S}\;\id{m})\;\id{n}\mathrel{=}\id{S}\;(\id{Plus}\;\id{m}\;\id{n}){}\<[E]%
\ColumnHook
\end{hscode}\resethooks
This formulation behaves roughly as we expect: \ensuremath{\id{Plus}\;\id{M}\;\id{N}} evaluates to the sum of the naturals \ensuremath{\id{M}}
and \ensuremath{\id{N}}, while the predicate \ensuremath{\id{PlusC}\;\id{M}\;\id{N}} is satisfied for arbitrary naturals \ensuremath{\id{M}} and \ensuremath{\id{N}}.  However,
in this case, the \ensuremath{\id{PlusC}\;\id{M}\;\id{N}} predicates are unnecessary: \ensuremath{\id{Plus}\;\id{M}\;\id{N}} is defined for arbitrary
naturals \ensuremath{\id{M}} and \ensuremath{\id{N}}.  Furthermore, the requirement to include these predicates could significantly
complicate definitions using polymorphic recursion.  For a simple example, consider the definition
of the \ensuremath{\id{append}} function for length-indexed vectors.  We might hope to write it as follows:
\begin{hscode}\SaveRestoreHook
\column{B}{@{}>{\hspre}l<{\hspost}@{}}%
\column{3}{@{}>{\hspre}l<{\hspost}@{}}%
\column{9}{@{}>{\hspre}c<{\hspost}@{}}%
\column{9E}{@{}l@{}}%
\column{13}{@{}>{\hspre}l<{\hspost}@{}}%
\column{21}{@{}>{\hspre}l<{\hspost}@{}}%
\column{25}{@{}>{\hspre}c<{\hspost}@{}}%
\column{25E}{@{}l@{}}%
\column{28}{@{}>{\hspre}l<{\hspost}@{}}%
\column{E}{@{}>{\hspre}l<{\hspost}@{}}%
\>[B]{}\keyword{data}\;\id{Vec}\;(\id{a}\mathbin{::}\star)\;(\id{n}\mathbin{::}\id{Nat})\;\keyword{where}{}\<[E]%
\\
\>[B]{}\hsindent{3}{}\<[3]%
\>[3]{}\id{Nil}{}\<[9]%
\>[9]{}\mathbin{::}{}\<[9E]%
\>[13]{}\id{Vec}\;\id{a}\;\id{Z}{}\<[E]%
\\
\>[B]{}\hsindent{3}{}\<[3]%
\>[3]{}\id{Cons}{}\<[9]%
\>[9]{}\mathbin{::}{}\<[9E]%
\>[13]{}\id{a}\to \id{Vec}\;\id{a}\;\id{n}\to \id{Vec}\;\id{a}\;(\id{S}\;\id{n}){}\<[E]%
\\[\blanklineskip]%
\>[B]{}\id{append}\mathbin{::}\id{PlusC}\;\id{m}\;\id{n}\Rightarrow \id{Vec}\;\id{a}\;\id{m}\to \id{Vec}\;\id{a}\;\id{n}\to \id{Vec}\;\id{a}\;(\id{Plus}\;\id{m}\;\id{n}){}\<[E]%
\\
\>[B]{}\id{append}\;\id{Nil}\;{}\<[21]%
\>[21]{}\id{ys}{}\<[25]%
\>[25]{}\mathrel{=}{}\<[25E]%
\>[28]{}\id{ys}{}\<[E]%
\\
\>[B]{}\id{append}\;(\id{Cons}\;\id{x}\;\id{xs})\;{}\<[21]%
\>[21]{}\id{ys}{}\<[25]%
\>[25]{}\mathrel{=}{}\<[25E]%
\>[28]{}\id{Cons}\;\id{x}\;(\id{append}\;\id{xs}\;\id{ys}){}\<[E]%
\ColumnHook
\end{hscode}\resethooks
However, the type signature given here is not strong enough: in the second case, where we know that
\ensuremath{\id{m}} is \ensuremath{\id{S}\;\id{m'}} for some \ensuremath{\id{m'}}, we also need to know that \ensuremath{\id{PlusC}\;\id{m'}\;\id{n}} holds.  But this does not
follow from the assumption \ensuremath{\id{PlusC}\;(\id{S}\;\id{m'})\;\id{n}}.  It would seem that we would have to define \ensuremath{\id{append}}
itself via a type class:
\begin{hscode}\SaveRestoreHook
\column{B}{@{}>{\hspre}l<{\hspost}@{}}%
\column{3}{@{}>{\hspre}l<{\hspost}@{}}%
\column{5}{@{}>{\hspre}l<{\hspost}@{}}%
\column{E}{@{}>{\hspre}l<{\hspost}@{}}%
\>[B]{}\keyword{class}\;\id{PlusC}\;\id{m}\;\id{n}\Rightarrow \id{Append}\;\id{m}\;\id{n}\;\keyword{where}{}\<[E]%
\\
\>[B]{}\hsindent{3}{}\<[3]%
\>[3]{}\id{append}\mathbin{::}\id{Vec}\;\id{a}\;\id{m}\to \id{Vec}\;\id{a}\;\id{n}\to \id{Vec}\;\id{a}\;(\id{Plus}\;\id{m}\;\id{n}){}\<[E]%
\\[\blanklineskip]%
\>[B]{}\hsindent{3}{}\<[3]%
\>[3]{}\keyword{instance}\;\id{Append}\;\id{Z}\;\id{n}\;\keyword{where}{}\<[E]%
\\
\>[3]{}\hsindent{2}{}\<[5]%
\>[5]{}\id{append}\;\id{Nil}\;\id{ys}\mathrel{=}\id{ys}{}\<[E]%
\\[\blanklineskip]%
\>[B]{}\hsindent{3}{}\<[3]%
\>[3]{}\keyword{instance}\;\id{Append}\;\id{m}\;\id{n}\Rightarrow \id{Append}\;(\id{S}\;\id{m})\;\id{n}\;\keyword{where}{}\<[E]%
\\
\>[3]{}\hsindent{2}{}\<[5]%
\>[5]{}\id{append}\;(\id{Cons}\;\id{x}\;\id{xs})\;\id{ys}\mathrel{=}\id{Cons}\;\id{x}\;(\id{append}\;\id{xs}\;\id{ys}){}\<[E]%
\ColumnHook
\end{hscode}\resethooks
But this is verbose, and complicates what should be a simple definition.  It also complicates uses
of \ensuremath{\id{append}}, which will now have to include the \ensuremath{\id{Append}} constraint instead of the \ensuremath{\id{PlusC}}
constraint or (even better) just an application of the \ensuremath{\id{Plus}} type family.

In essence, having recognized that most type families are partial, \emph{some} are total, and users should
be able to take advantage of this fact. If we could recognize \ensuremath{\id{Plus}} as total, then we could allow
the following, much simpler definition of \ensuremath{\id{append}}:
\begin{hscode}\SaveRestoreHook
\column{B}{@{}>{\hspre}l<{\hspost}@{}}%
\column{21}{@{}>{\hspre}l<{\hspost}@{}}%
\column{25}{@{}>{\hspre}c<{\hspost}@{}}%
\column{25E}{@{}l@{}}%
\column{28}{@{}>{\hspre}l<{\hspost}@{}}%
\column{E}{@{}>{\hspre}l<{\hspost}@{}}%
\>[B]{}\id{append}\mathbin{::}\id{Vec}\;\id{a}\;\id{m}\to \id{Vec}\;\id{a}\;\id{n}\to \id{Vec}\;\id{a}\;(\id{Plus}\;\id{m}\;\id{n}){}\<[E]%
\\
\>[B]{}\id{append}\;\id{Nil}\;{}\<[21]%
\>[21]{}\id{ys}{}\<[25]%
\>[25]{}\mathrel{=}{}\<[25E]%
\>[28]{}\id{ys}{}\<[E]%
\\
\>[B]{}\id{append}\;(\id{Cons}\;\id{x}\;\id{xs})\;{}\<[21]%
\>[21]{}\id{ys}{}\<[25]%
\>[25]{}\mathrel{=}{}\<[25E]%
\>[28]{}\id{Cons}\;\id{x}\;(\id{append}\;\id{xs}\;\id{ys}){}\<[E]%
\ColumnHook
\end{hscode}\resethooks
This definition needs no constraints, as the type-checker is aware that \ensuremath{\id{Plus}} is total, with no
possibility for a usage outside its domain of definition.

We now have a new, challenging question: how do we know when a type family is total?
Totality checking of functional programs is a hard problem, one we do not propose to solve here.
This problem is well studied both in the context of dependently-typed programming\footnote{E.g., \url{https://coq.inria.fr/cocorico/CoqTerminationDiscussion}} and
partial evaluation~\cite{LeeJB01,SereniJ05}. In practice, an implementation of our ideas would use a
totality checker to discover or check the totality of type families. Users could also have the capability to
(unsafely) assert the totality of functions that 
lie beyond the abilities of the checker.


We can extend our type formation rules~\secref{valid-ctf} to take account of total type families.
Intuitively, we can think of a total type family as a constrained type family for which the
constraint is trivially provable.  To formalize this notion, we extend our top-level environment
$\Sigma$ to include total type families $\top \Rightarrow F : n$ as well as partial type families
$C \Rightarrow F : n$.  Then, we can add a new rule that allows total type families regardless of
the context:
\[
\ottdruleSTXXTFamily{}
\]
While this rule is superficially similar to the rule for type constructors, it will have a different
elaboration into our core calculus, which must explicitly account for the totality of \ensuremath{\id{F}}.

\subsection{Simplifying Apartness}
\label{sec:apartness}

As introduced above~\pref{sec:old-apartness}, closed type family reduction critically relies on a
notion of apartness on types. The existing definition of apartness~\cite[\S3.3]{EisenbergVPJW14} is
subtle, requiring both infinitary unification and a \emph{flattening} operation to account for the
possibility of type family applications in the arguments to another type family. Because type
families cannot appear directly as arguments to other type families, the flattening
operation---whose details thankfully no longer concern us---becomes redundant.  In addition, because
we require the caller of a function to provide the ground type to which a type family reduces at
every call site, we no longer have to worry about infinite types and infinitary unification.  Thus,
we can define apartness very simply: as the inverse of unifiability. Indeed, our formal
development~\pref{sec:formal} no longer contains a first-class notion of apartness, using
unification directly.

\section{Type Safety of Constrained Type Families}
\label{sec:formal}

For over a decade, GHC has compiled its variant of Haskell into System FC~\cite{SulzmannCJD07}, a
variant of System F~\cite{GirardTL89,Reynolds74} that supports explicit \emph{coercions}, or proofs
of equality between types. As type family instances introduce new such equalities (via axioms), type
families are integrated into FC. Accordingly, proving the type safety of System FC
requires careful reasoning about type family reduction. As the safety of Haskell itself rests on the
safety of FC,\footnote{We are unaware of a precise semantics for the surface Haskell language that
  accounts for all the features of modern GHC/Haskell.} we must now show that our extension of
constrained type families retains soundness.

Indeed we go further: by adding constrained type families and a new treatment of axioms, we can now
prove that all type family reduction chains in System FC terminate, thus closing the gap in the
proof presented by \citet{EisenbergVPJW14}, which was unable to cope with the
interaction of non-linear patterns and non-terminating type families.


This section presents an overview of our formalism and a sketch of our proofs. The full definitions
and proofs can be found in \auxiliarymatter.

\subsection{System CFC}


\newcommand{\typeregstatement}{If $\Gamma  \vdash  \ottnt{e}  \ottsym{:}  \tau$, then $ \Gamma   \vdash   \tau  \ \mathsf{type} $.}

\newcommand{\presstatement}{If $\emptyset  \vdash  \ottnt{e}  \ottsym{:}  \tau$ and $\ottnt{e}  \longrightarrow  \ottnt{e'}$, then $\emptyset  \vdash  \ottnt{e'}  \ottsym{:}  \tau$.}

\newcommand{\progstatement}{If $\emptyset  \vdash  \ottnt{e}  \ottsym{:}  \tau$, then either $\ottnt{e}$ is a value $\ottnt{v}$, $\ottnt{e}$ is
a coerced value $\ottnt{v}  \triangleright  \gamma$, or $\ottnt{e}  \longrightarrow  \ottnt{e'}$ for some $\ottnt{e'}$.}

\newcommand{\consstatement}{If $\emptyset  \vdash  \gamma  \ottsym{:}  \tau_{{\mathrm{1}}}  \sim  \tau_{{\mathrm{2}}}$, $ \emptyset   \vdash   \tau_{{\mathrm{1}}}  \ \mathsf{type} $, and $ \emptyset   \vdash   \tau_{{\mathrm{2}}}  \ \mathsf{type} $,
then $\tau_{{\mathrm{1}}} \, \ottsym{=} \, \tau_{{\mathrm{2}}}$.}

\newcommand{\completestatement}{If $\emptyset  \vdash  \gamma  \ottsym{:}  \tau_{{\mathrm{1}}}  \sim  \tau_{{\mathrm{2}}}$, then there exists $\tau_{{\mathrm{3}}}$ such that $\tau_{{\mathrm{1}}}  \rightsquigarrow^*  \tau_{{\mathrm{3}}}  \leftsquigarrow^*  \tau_{{\mathrm{2}}}$.}

\newcommand{\noredstatement}{If $ \Gamma   \vdash   \tau  \ \mathsf{type} $, then there exists no $\tau'$ such that $\tau  \rightsquigarrow  \tau'$.}

\newcommand{\termstatement}{For all types $\tau$, there exists a type $\sigma$ such that $\tau  \rightsquigarrow^*  \sigma$ and
$\sigma$ cannot reduce.}

\newcommand{\localconfstatement}{If $\tau_{{\mathrm{1}}}  \leftsquigarrow  \tau_{{\mathrm{0}}}  \rightsquigarrow  \tau_{{\mathrm{2}}}$, then there exists $\tau_{{\mathrm{3}}}$ such that
$\tau_{{\mathrm{1}}}  \rightsquigarrow^*  \tau_{{\mathrm{3}}}  \leftsquigarrow^*  \tau_{{\mathrm{2}}}$.}

\newcommand{\confstatement}{If $\tau_{{\mathrm{1}}}  \leftsquigarrow^*  \tau_{{\mathrm{0}}}  \rightsquigarrow^*  \tau_{{\mathrm{2}}}$, then there exists $\tau_{{\mathrm{3}}}$ such that $\tau_{{\mathrm{1}}}  \rightsquigarrow^*  \tau_{{\mathrm{3}}}  \leftsquigarrow^*  \tau_{{\mathrm{2}}}$.}

\newcommand{\totaltypefamilies}{
\begin{property}[Total type families]
\label{prop:total}
For all $\ottmv{F} \,  \mathop{ :_{\top} }  \, \ottmv{n}  \in  \Sigma$
 and all $\ottcompu{\tau_{\ottmv{i}}}{\ottmv{i}}{\ottmv{n}}$ such
that $\overline{ \emptyset   \vdash   \tau_{\ottmv{i}}  \ \mathsf{type} }$, there exists $\ottnt{q}$ such that
$\emptyset  \vdash  \ottnt{q}  \ottsym{:}   (  \alpha  \pipe  \ottmv{c}  :  \ottmv{F} \, \overline{\tau}   \sim   \alpha  ) $.
Define $ \ottmv{F} \  \overline{\tau}  \Downarrow  \ottnt{q} $ to witness
the above fact.
\end{property}
}

\begin{figure}
\begin{tabular}{c@{\hspace{2em}}c}
\begin{minipage}{.47\textwidth}
\paragraph{Metavariables.}
\[
\begin{array}{rlrl@{}}
\alpha & \text{type variables}  &
\ottmv{x} & \text{term variables} \\
\ottmv{c} & \text{coercion variables} &
\xi & \text{axioms} \\
\ottmv{F} & \text{type families} &
\ottnt{H} & \text{type constants} \\
\ottmv{K} & \multicolumn{3}{l}{\text{term constants (constructors)}}
\end{array}
\]
\end{minipage}
&
\begin{minipage}{.45\textwidth}
\paragraph{Notations.}
\begin{itemize}
\setlength{\itemindent}{-5ex}
\item Substitutions application: $\tau  \ottsym{[}  \theta  \ottsym{]}$
\item Substitutions composition: $\theta \, \ottsym{=} \, \theta_{{\mathrm{1}}}  \circ  \theta_{{\mathrm{2}}}$
\item $\ottmv{F} \, : \, \ottmv{n}$ stands for either $\ottmv{F} \,  \mathop{:_{\not\top} }  \, \ottmv{n}$ or $\ottmv{F} \,  \mathop{ :_{\top} }  \, \ottmv{n}$
\item Free variables of constructs: $\mathit{fv}\! \, \ottsym{(}   \cdot   \ottsym{)}$
\item $\mathit{tvs}\! \, \ottsym{(}  \overline{\chi}  \ottsym{)}$: bound type variables of $\overline{\chi}$
\item Domains of contexts are denoted $\mathit{dom}\! \, \ottsym{(}  \Gamma  \ottsym{)}$
\end{itemize}
\end{minipage}
\end{tabular}\\[2ex]


\paragraph{Grammar.}
\[
\begin{array}{rcll}
\tau,\sigma, \rho &\bnfeq& \ottnt{H} \, \overline{\tau} \bnfor \tau_{{\mathrm{1}}}  \to  \tau_{{\mathrm{2}}} \bnfor \alpha \bnfor \forall \, \alpha  \ottsym{.}  \tau \bnfor
                            \phi  \Rightarrow  \tau \bnfor \ottmv{F} \, \overline{\tau} & \text{types} \\
\phi &\bnfeq& \tau_{{\mathrm{1}}}  \sim  \tau_{{\mathrm{2}}} & \text{constraints} \\
\gamma, \eta &\bnfeq&  \langle  \tau  \rangle  \bnfor \ottkw{sym} \, \gamma \bnfor \gamma_{{\mathrm{1}}}  \fatsemi  \gamma_{{\mathrm{2}}} \bnfor \ottnt{H} \, \overline{\gamma}
                              \bnfor \gamma_{{\mathrm{1}}}  \to  \gamma_{{\mathrm{2}}} \bnfor \forall \, \alpha  \ottsym{.}  \gamma  & \text{coercions} \\
 & \bnfor & \gamma_{{\mathrm{1}}}  \sim  \gamma_{{\mathrm{2}}}  \Rightarrow  \gamma_{{\mathrm{3}}} \bnfor \ottmv{F} \, \overline{\gamma}
                              \bnfor  \ottkw{nth} _{ \ottnt{i} }\  \gamma  \bnfor \gamma  \at  \tau \bnfor \ottmv{c} \bnfor  \xi _{ \ottnt{i} }\  \overline{\tau} \  \overline{\ottnt{q} }  \\
\ottnt{e} &\bnfeq& \ottmv{x} \bnfor \ottmv{K} \bnfor \lambda  \ottmv{x}  \ottsym{:}  \tau  \ottsym{.}  \ottnt{e} \bnfor \ottnt{e_{{\mathrm{1}}}} \, \ottnt{e_{{\mathrm{2}}}} \bnfor \Lambda  \alpha  \ottsym{.}  \ottnt{e}
                     \bnfor \ottnt{e} \, \tau & \text{expressions} \\
      &\bnfor& \lambda  \ottmv{c}  \ottsym{:}  \phi  \ottsym{.}  \ottnt{e} \bnfor \ottnt{e} \, \gamma \bnfor \ottnt{e}  \triangleright  \gamma
                     \bnfor \ottkw{assume} \, \chi \, \ottkw{in} \, \ottnt{e} \\
\ottnt{v} &\bnfeq& \ottmv{K} \bnfor \lambda  \ottmv{x}  \ottsym{:}  \tau  \ottsym{.}  \ottnt{e} \bnfor \Lambda  \alpha  \ottsym{.}  \ottnt{e} \bnfor \lambda  \ottmv{c}  \ottsym{:}  \phi  \ottsym{.}  \ottnt{e} & \text{values} \\
\chi &\bnfeq&  (  \alpha  \pipe  \ottmv{c}  :  \ottmv{F} \, \overline{\tau}   \sim   \alpha  )  & \text{evaluation assumption} \\
\ottnt{q} &\bnfeq&  ( \tau  \pipe  \gamma )  & \text{evaluation resolution} \\[1ex]
\ottnt{E} &\bnfeq& \forall \, \overline{\alpha} \, \overline{\chi}  \ottsym{.}  \ottmv{F} \, \overline{\tau}  \sim  \tau_{{\mathrm{0}}} & \text{type family equations} \\
\Sigma &\bnfeq&  \emptyset  \bnfor \Sigma  \ottsym{,}  \ottmv{F} \,  \mathop{:_{\not\top} }  \, \ottmv{n} \bnfor \Sigma  \ottsym{,}  \ottmv{F} \,  \mathop{ :_{\top} }  \, \ottmv{n} \bnfor \Sigma  \ottsym{,}  \xi  \ottsym{:}  \overline{\ottnt{E} } & \text{signatures} \\
\delta &\bnfeq& \alpha \bnfor  \ottmv{c} {:} \phi  \bnfor  \ottmv{x} {:} \tau  & \text{bindings} \\
\Gamma &\bnfeq&  \emptyset  \bnfor \Gamma  \ottsym{,}  \delta & \text{typing contexts} \\[1ex]
\theta &\bnfeq&  \emptyset  \bnfor \theta  \ottsym{,}  \tau  \ottsym{/}  \alpha \bnfor \theta  \ottsym{,}  \gamma  \ottsym{/}  \ottmv{c} \bnfor \theta  \ottsym{,}  \ottnt{e}  \ottsym{/}  \ottmv{x} & \text{substitutions} \\
\mathcal{V} &\bnfeq& \ldots & \text{sets of variables} \\
\mathcal{C}  \ottsym{[}   \cdot   \ottsym{]} &\bnfeq& \ldots & \text{one-hole type contexts}
\end{array}
\]

\caption{System CFC Design}
\label{fig:fc-grammar}
\end{figure}

\newcommand{\nl}{\\[1ex]}
\begin{figure}
\begin{judg}{$ \Gamma   \vdash   \tau  \ \mathsf{type} $}{Type validity}
\ottdruleTXXTyCon{} \qquad
\ottdruleTXXArrow{} \qquad
\ottdruleTXXVar{} \nl
\ottdruleTXXForall{} \qquad
\ottdruleTXXQual{}
\end{judg}\nl
\begin{judg}{$ \Gamma   \vdash   \phi  \ \mathsf{prop} $}{Proposition validity}
\ottdrulePXXTypes{} \qquad
\ottdrulePXXFamily{}
\end{judg}\nl
\begin{judg}{$ \vdash   \Gamma  \ \mathsf{ctx} $}{Context validity}
\ottdruleGXXNil{} \qquad
\ottdruleGXXTyVar{} \qquad
\ottdruleGXXCoVar{} \qquad
\ottdruleGXXVar{}
\end{judg}
\caption{Type validity judgments}
\label{fig:type-judg}
\label{fig:prop-judg}
\label{fig:judg-begin}
\end{figure}

\begin{figure}
\begin{judg}{$\Gamma  \vdash  \ottnt{e}  \ottsym{:}  \tau$}{Expression typing}
\ottdruleEXXVar{} \qquad
\ottdruleEXXConst{} \nl
\ottdruleEXXLam{} \qquad
\ottdruleEXXTLam{} \qquad
\ottdruleEXXCLam{} \nl
\ottdruleEXXApp{} \qquad
\ottdruleEXXTApp{} \qquad
\ottdruleEXXCApp{} \nl
\ottdruleEXXCast{} \qquad
\ottdruleEXXAssume{}
\end{judg}\nl
\begin{judg}{$\ottnt{e}  \longrightarrow  \ottnt{e'}$}{Small-step operational semantics}
\ottdruleSXXApp{} \quad
\ottdruleSXXTApp{} \quad
\ottdruleSXXCApp{} \quad
\ottdruleSXXCast{} \nl
\ottdruleSXXBeta{} \qquad
\ottdruleSXXTBeta{} \nl
\ottdruleSXXCBeta{} \qquad
\ottdruleSXXCPush{} \nl
\ottdruleSXXPush{} \quad
\ottdruleSXXTPush{} \nl
\ottdruleSXXTrans{} \qquad
\ottdruleSXXResolve{}
\end{judg}
\caption{Expression judgments}
\label{fig:expr-judgments}
\label{fig:expr-type-judg}
\end{figure}

\begin{figure}
\begin{judg}{$\Gamma  \vdash  \gamma  \ottsym{:}  \phi$}{Coercion validity}
\ottdruleCXXRefl{} \qquad
\ottdruleCXXSym{} \qquad
\ottdruleCXXTrans{} \nl
\ottdruleCXXApp{} \qquad
\ottdruleCXXFun{} \nl
\ottdruleCXXFam{} \qquad
\ottdruleCXXForall{} \nl
\ottdruleCXXQual{} \nl
\ottdruleCXXNth{} \qquad
\ottdruleCXXNthArrow{} \nl
\ottdruleCXXNthQual{} \quad
\ottdruleCXXInst{} \nl
\ottdruleCXXVar{} \qquad
\ottdruleCXXAxiom{}
\end{judg}\nl
\begin{judg}{$\Gamma  \vdash  \overline{\ottnt{q} }  \ottsym{:}  \overline{\chi}$}{Evaluation resolution validity}
\ottdruleAXXNil{} \qquad
\ottdruleAXXCons{}
\end{judg}\nl
\caption{Coercion validity judgments}
\label{fig:co-judg}
\label{fig:qs-fs-judg}
\end{figure}

\begin{figure}
\begin{center}
\begin{judg}{$ \vdash   \Sigma  \ \mathsf{ok} $}{Signature validity}
\begin{array}[b]{@{}c@{}}
\ottdruleDXXNil{} \\[3ex]
\ottdruleDXXPartial{} \\[3ex]
\ottdruleDXXTotal{}
\end{array} \qquad
\ottdruleDXXAxiom{}
\end{judg}\nl
\begin{judg}{$ \Gamma   \vdash   \overline{\chi} \ \mathsf{assumps} $}{Evaluation assumptions validity}
\ottdruleXXXNil{} \qquad
\ottdruleXXXCons{}
\end{judg}
\end{center}
\caption{Signature validity}
\label{fig:decl-judgments}
\end{figure}

\begin{figure}
\begin{center}
\begin{judg}{$ \mathsf{compat} ( \ottnt{E_{{\mathrm{1}}}} ,  \ottnt{E_{{\mathrm{2}}}} ) $}{Equation compatibility}
\ottdruleCoXXCoinc{} \quad
\ottdruleCoXXDistinct{}
\end{judg}\nl
\begin{judg}{$ \mathsf{no\_conflict} ( \overline{\ottnt{E} } ,  \ottnt{i} ,  \overline{\tau} ,  \ottnt{j} ) $}{Check for equation conflicts}
\ottdruleNCXXApart{} \qquad
\ottdruleNCXXCompatible{}
\end{judg}\nl
\end{center}
\caption{Closed type family auxiliary judgments}
\label{fig:ctfs}
\label{fig:judg-end}
\end{figure}

We will study a simplified version of System FC, called CFC (``constrained FC''). The grammar for the
language is presented in \pref{fig:fc-grammar} and is checked by the judgments in Figures \ref{fig:judg-begin}--\ref{fig:judg-end}.
Broadly speaking, CFC is like System F, but with explicit coercions witnessing equality between types and
usable in type conversions (see rule \rul{E\_Cast}, \pref{fig:expr-type-judg}). The features in this system
beyond those in System F are all driven by these coercions. Before describing the novelty of CFC, we take a quick tour of the grounds. Novel components are indicated in the following
discussion; the rest of System CFC follows previous work (e.g., \cite{BreitnerEPW16,EisenbergVPJW14}).

Types in CFC are like those in System F, but with three additions: $\ottnt{H} \, \overline{\tau}$ is a fully applied
type constant $\ottnt{H}$ (allowing partial application would require reasoning about kinds), $\phi  \Rightarrow  \tau$
is a type $\tau$ qualified by an equality assumption $\phi$, and $\ottmv{F} \, \overline{\tau}$ is a fully applied type family
$\ottmv{F}$.
Perhaps unexpectedly,
classes are not included. The novel constrained nature of type families arises from CFC's
 differentiation between pretypes (any production of metavariable $\tau$)
and types (as validated by
$ \Gamma   \vdash   \tau  \ \mathsf{type} $, \pref{fig:type-judg}); proper types may mention type families only
in a proposition $\phi$.
Examine the judgment $ \Gamma   \vdash   \phi  \ \mathsf{prop} $ (\pref{fig:prop-judg}). Its rule \rul{P\_Types} allows the proposition
to relate two proper types, while the rule \rul{P\_Family} allows a saturated type family application to
be related to a type. Thus, in CFC, we would write \ensuremath{\id{insert}\mathbin{::}\forall\!\! \hsforall \;\id{a}\;\id{c}\hsdot{\circ }{.\,}\id{Elem}\;\id{c}\,\sim\,\id{a}\Rightarrow \id{a}\to \id{c}\to \id{c}} instead
of the more typical \ensuremath{\id{insert}\mathbin{::}\forall\!\! \hsforall \;\id{c}\hsdot{\circ }{.\,}\id{Collects}\;\id{c}\Rightarrow \id{Elem}\;\id{c}\to \id{c}\to \id{c}}. In effect, the type family
equality assumption
\ensuremath{\id{Elem}\;\id{c}\,\sim\,\id{a}} takes the place of the class constraint \ensuremath{\id{Collects}\;\id{c}}: both assert that \ensuremath{\id{Elem}\;\id{c}} can evaluate
to a proper (type family-free) type.

The language omits any consideration of kinds, as the
complexity of kinds does not illuminate the invention of constrained type families.

Expressions $\ottnt{e}$ are checked by the judgment $\Gamma  \vdash  \ottnt{e}  \ottsym{:}  \tau$ (\pref{fig:expr-type-judg}). There are two leaf forms, for variables $\ottmv{x}$ and constants (such as data
constructors) $\ottmv{K}$. In addition to System F's two forms of abstraction and application (over
expressions and types), CFC contains abstraction and application over coercions. Accordingly,
a function may
assume an equality proposition $\phi$ relating two types. The feature can be seen in the rules \rul{E\_CLam} and \rul{E\_CApp} (\pref{fig:expr-type-judg}). Though this language omits datatypes, generalized algebraic datatypes (GADTs) can
be encoded using coercion abstractions~\cite[\S3.2]{SulzmannCJD07}. Coercions are
used in casts $\ottnt{e}  \triangleright  \gamma$, which use the coercion to change the type of an expression
(\rul{E\_Cast}, \pref{fig:expr-type-judg}).
Lastly, expressions also
contain a novel form $\ottkw{assume} \, \chi \, \ottkw{in} \, \ottnt{e}$ used in our account of total type
families~\secref{cfc-totality}.

The small-step operational semantics (\pref{fig:expr-judgments}) provides the relation
$\ottnt{e}  \longrightarrow  \ottnt{e'}$.  The definition for $ \longrightarrow $ contains congruence forms to allow evaluation in
applications and casts, $\beta$-reductions over the three application forms, and four push rules
(counting \rul{S\_Trans} as a push rule for casts). The push rules allow us to move casts around
when they get in the way---for example when a cast prevents us from reducing an applied
$\lambda$-expression. Though somewhat intricate, these rules are derived straightforwardly simply by
making choices in order to have the output expression preserve the type of the input expression. The
novel rule \rul{S\_Resolve} is discussed with $ \ottkw{assume} $~\pref{sec:cfc-totality}. Values in CFC
are unsurprisingly constants and abstractions.

Of the main productions in the grammar, we are left with coercions $\gamma$, checked by the judgment
$\Gamma  \vdash  \gamma  \ottsym{:}  \phi$ (\pref{fig:co-judg}). A coercion is a witness of type equality; thus, the rules
for coercion formation determine the equality relation underlying the type system.\footnote{In
a similar system that leaves coercions out but has a conversion rule, the rules for
$\Gamma  \vdash  \gamma  \ottsym{:}  \phi$ would correspond to rules for definitional equality, often rendered with
$\equiv$.} The critical
property of this relation is \emph{consistency}---that we can never prove, for example, that \ensuremath{\id{Int}}
equals \ensuremath{\id{Bool}}. We return to consistency and our proof thereof later in this section~\pref{sec:cfc-consistency}. The equality
relation as witnessed by these coercions has several properties:
\begin{itemize}
\item Our equality relation is indeed an equivalence, as witnessed by coercion forms for reflexivity
($ \langle  \tau  \rangle $), symmetry ($\ottkw{sym} \, \gamma$), and transitivity ($\gamma_{{\mathrm{1}}}  \fatsemi  \gamma_{{\mathrm{2}}}$).
\item Equality is congruent,
as witnessed by a coercion for each recursive type form.
\item  Equality can be decomposed via the $ \ottkw{nth} _{ \ottnt{i} }\  \gamma $ and $\gamma  \at  \tau$ coercions. The former
extracts equalities from applied type constants (\rul{C\_Nth}), function arrows (\rul{C\_NthArrow}),
and qualified
types (\rul{C\_NthQual}). The latter instantiates an equality between polytypes (\rul{C\_Inst}), giving
us an equality between the two polytype bodies.
\item Equality can be assumed, as witnessed by coercion variables $\ottmv{c}$.
\item Crucially, equality witnesses the reduction of type families through the form $ \xi _{ \ottnt{i} }\  \overline{\tau} \  \overline{\ottnt{q} } $
and the rule \rul{C\_Axiom}, as discussed in the next subsection.
\end{itemize}
Unlike in other developments of System FC, this system does \emph{not} support a coercion regularity
lemma; that is, $\Gamma  \vdash  \gamma  \ottsym{:}  \phi$ does \emph{not} imply that $ \Gamma   \vdash   \phi  \ \mathsf{prop} $. In other words,
the two types related by a coercion may mention type families at arbitrary depths. The lemma was
used primarily for convenience in prior proofs; its omission here does not bite.

\subsection{Type Family Axioms and Signatures}
Following prior work on System FC (initially that of \citet{SulzmannCJD07}), we use
\emph{axioms} $\xi$ to witness type family reductions. That is, if there is an equation \ensuremath{\keyword{type}\;\id{F}\;\id{Int}\mathrel{=}\id{Bool}} in scope, then we have an axiom $\xi$ that proves \ensuremath{\id{F}\;\id{Int}\,\sim\,\id{Bool}}. An
expression can then use this axiom to cast an expression of type \ensuremath{\id{Bool}} to one of type \ensuremath{\id{F}\;\id{Int}}.


In System CFC, axioms exist in an ambient signature $\Sigma$ (which, more formally, should appear
in every judgment; we omit this to reduce clutter). Signatures contain both declarations for
type families $\ottmv{F} \, : \, \ottmv{n}$ and axiom declarations $\xi  \ottsym{:}  \overline{\ottnt{E} }$. The former has two forms:
$\ottmv{F} \,  \mathop{:_{\not\top} }  \, \ottmv{n}$ declares a \emph{partial} type family $\ottmv{F}$ that takes $\ottmv{n}$ arguments, and
$\ottmv{F} \,  \mathop{ :_{\top} }  \, \ottmv{n}$ declares a \emph{total} type family. The difference is in the treatment of the
$ \ottkw{assume} $ construct~\pref{sec:cfc-totality}.

An axiom $\xi$ is classified by a list of equations $\overline{\ottnt{E} }$, where each equation has
the form $\forall \, \overline{\alpha} \, \overline{\chi}  \ottsym{.}  \ottmv{F} \, \overline{\tau}  \sim  \tau_{{\mathrm{0}}}$. Using a list of equations, as opposed to only one equation,
is necessary to support closed type families, with their ordered lists of equations. However,
the intricacies of closed type families do not affect our main contribution to this formalism
(i.e., the constraining of type family applications via the distinction between pretypes and
types). We will thus consider only singleton lists of equations $\ottnt{E}$ for now. We return
to the full generality of closed type families below.

In an equation $\ottnt{E}$, the types $\overline{\tau}$ and the
type $\tau_{{\mathrm{0}}}$ are proper types, with no type family applications; the lack of type family
application on the right-hand side ($\tau_{{\mathrm{0}}}$) is new in this work.
As in prior work on type families, equations
can be quantified over type variables $\overline{\alpha}$; this allows the equations to work at many types.
For example, the equation \ensuremath{\id{F}\;(\id{Maybe}\;\id{a})\mathrel{=}\id{a}} is quantified over the variable \ensuremath{\id{a}}.

Also novel in this
work is quantification over \emph{evaluation assumptions} $\overline{\chi}$. The form for $\chi$ is
$ (  \alpha  \pipe  \ottmv{c}  :  \ottmv{F} \, \overline{\tau}   \sim   \alpha  ) $, read ``$\alpha$ such that $\ottmv{c}$ witnesses that $\ottmv{F} \, \overline{\tau}$ reduces
to $\alpha$''. Quantification over evaluation assumptions is necessary to support type families
that reduce to other type families. For example, we might have \ensuremath{\id{F}\;(\id{Maybe}\;\id{a})\mathrel{=}\id{G}\;\id{a}}; such an
equation would compile to \ensuremath{\forall\!\! \hsforall \;\id{a}\;(\id{b}\mid \id{c}\mathbin{:}\id{G}\;\id{a}\,\sim\,\id{b})\hsdot{\circ }{.\,}\id{F}\;(\id{Maybe}\;\id{a})\,\sim\,\id{b}}. Because of evaluation
assumptions, we can continue to support equations such as \ensuremath{\id{F}\;(\id{Maybe}\;\id{a})\mathrel{=}\id{G}\;\id{a}} even while disallowing
type families on the right-hand sides of axioms. The assumptions in a type family equation
bind a coercion variable $\ottmv{c}$, though this variable is not used; the use of $\chi$ here (instead
of a construct that does not bind $\ottmv{c}$) is for simplicity and parallelism with the $\chi$ in
the $ \ottkw{assume} $ construct. Note that evaluation assumptions are more specific than arbitrary
equality assumptions $\phi$, requiring a type family on the left and requiring that the right-hand
side be a fresh type variable. This restrictive form is critical in proving that type family
reduction is confluent~\pref{sec:cfc-consistency}.

Signatures, with their type family equations, are validated by the judgment $ \vdash   \Sigma  \ \mathsf{ok} $ and its
auxiliary judgment $ \Gamma   \vdash   \overline{\chi} \ \mathsf{assumps} $, both in \pref{fig:decl-judgments}.

The use of an axiom $\xi$ to form a coercion has the form $ \xi _{ \ottnt{i} }\  \overline{\tau} \  \overline{\ottnt{q} } $, supplying the index $\ottnt{i}$ of
the equation to use (for now, $\ottnt{i}$ will always be 0),
a list of types $\overline{\tau}$ used to instantiate the type variables $\overline{\alpha}$, and a
list of \emph{evaluation resolutions} $\overline{\ottnt{q} }$ used to instantiate the evaluation assumptions
$\overline{\chi}$. An evaluation resolution $\ottnt{q}$ has the form $ ( \tau  \pipe  \gamma ) $, where the type $\tau$ can
instantiate the type variable $\alpha$ in $ (  \alpha  \pipe  \ottmv{c}  :  \ottmv{F} \, \overline{\tau}   \sim   \alpha  ) $, and the coercion $\gamma$ proves
the equality and instantiates the $\ottmv{c}$. We write $\ottnt{q}  \ottsym{/}  \chi$ to mean a substitution that maps the
type and coercion, respectively.

To understand the daunting rule \rul{C\_Axiom}, let's first simplify it to eliminate the possibility
of multiple equations for the given axiom. Here is the simplified version:

\[
\ottdrule{\ottpremise{\xi  \ottsym{:}  \forall \, \overline{\alpha} \, \overline{\chi}  \ottsym{.}  \ottmv{F} \, \overline{\tau}  \sim  \tau_{{\mathrm{0}}}  \in  \Sigma \quad \hspace{-.6em} \quad  \vdash   \Gamma  \ \mathsf{ctx} }%
\ottpremise{\ottcomp{ \Gamma   \vdash   \sigma_{\ottmv{j}}  \ \mathsf{type} }{\ottmv{j}} \quad \hspace{-.6em} \quad \Gamma  \vdash  \overline{\ottnt{q} }  \ottsym{:}  \overline{\chi}  \ottsym{[}  \overline{\sigma}  \ottsym{/}  \overline{\alpha}  \ottsym{]}}}%
{\Gamma  \vdash   \xi _{ \ottsym{0} }\  \overline{\sigma} \  \overline{\ottnt{q} }   \ottsym{:}  \ottmv{F} \, \overline{\tau}  \ottsym{[}  \overline{\sigma}  \ottsym{/}  \overline{\alpha}  \ottsym{]}  \sim  \tau_{{\mathrm{0}}}  \ottsym{[}  \overline{\sigma}  \ottsym{/}  \overline{\alpha}  \ottsym{,}  \overline{\ottnt{q} }  \ottsym{/}  \overline{\chi}  \ottsym{]}}%
{\rul{C\_Axiom} (simplified)}
\]

The rule looks up the axiom in the signature,
checks to make sure the $\overline{\sigma}$ are proper (type family-free) types and that the $\overline{\ottnt{q} }$ satisfy the
assumptions $\overline{\chi}$ (using the auxiliary judgment $\Gamma  \vdash  \overline{\ottnt{q} }  \ottsym{:}  \overline{\chi}$, \pref{fig:qs-fs-judg}).
The $\Gamma  \vdash  \overline{\ottnt{q} }  \ottsym{:}  \overline{\chi}$ judgment is straightforward, matching up
the $\overline{\ottnt{q} }$ with the corresponding $\overline{\chi}$ and checking that the coercions in $\overline{\ottnt{q} }$ prove
the correct propositions.

Let's now generalize to full closed type families with an ordered list of equations.\footnote{The
inclusion of closed type families in the formalization is to support our claim of a consistency proof in
the presence of closed type families. However, the treatment of these families here is not novel,
and our contributions have a minimal impact on the presentation of closed type families---it is the
\emph{metatheory} that is affected, not the \emph{theory}. The intricacies of closed type families
may therefore be skipped by readers not interested in reproducing our proof.} Here is the
full rule for axioms:

\[
\ottdruleCXXAxiom{}
\]

Compared to the rule above, this rule uses the index $\ottnt{i}$ to look up the right equation; it
also adds an invocation of the $ \mathsf{no\_conflict} $ judgment (\pref{fig:ctfs}).
This check is substantively identical to the existing check for closed type families
but with our simplified notion of apartness (see \pref{sec:apartness}); the
two necessary judgments appear in \pref{fig:ctfs}. The only change from prior work is in the
use of the $ \mathit{subst}\! $ operator in the premise to \rul{Co\_Coinc}. This rule detects when two type
family equations are \emph{compatible}. Recalling Eisenberg et al.~\cite{EisenbergVPJW14},
two equations are compatible if, whenever they are both applicable to the same type, they will yield
the same result. This can happen in two ways: if the two equations' left-hand sides are unifiable, then
the right-hand sides coincide under the unifying substitution (\rul{Co\_Coinc}); or the two equations' left-hand sides
have no overlap (\rul{Co\_Distinct}). In the former case, we must be careful, as the true right-hand sides
of the equations may mention type families; we thus use $ \mathit{subst}\! $ to generate a substitution over the
evaluation assumptions $\overline{\chi}$, expanding out the variables bound in the $\overline{\chi}$ to the type
family applications they equal.

\subsection{Totality and Assumptions}
\label{sec:cfc-totality}

The challenge to totality in CFC is best understood by example. Consider again the \ensuremath{\id{append}}
operation on length-indexed vectors~\pref{sec:vector-append}, repeated here:
\begin{hscode}\SaveRestoreHook
\column{B}{@{}>{\hspre}l<{\hspost}@{}}%
\column{21}{@{}>{\hspre}l<{\hspost}@{}}%
\column{25}{@{}>{\hspre}c<{\hspost}@{}}%
\column{25E}{@{}l@{}}%
\column{28}{@{}>{\hspre}l<{\hspost}@{}}%
\column{E}{@{}>{\hspre}l<{\hspost}@{}}%
\>[B]{}\id{append}\mathbin{::}\id{Vec}\;\id{a}\;\id{m}\to \id{Vec}\;\id{a}\;\id{n}\to \id{Vec}\;\id{a}\;(\id{Plus}\;\id{m}\;\id{n}){}\<[E]%
\\
\>[B]{}\id{append}\;\id{Nil}\;{}\<[21]%
\>[21]{}\id{ys}{}\<[25]%
\>[25]{}\mathrel{=}{}\<[25E]%
\>[28]{}\id{ys}{}\<[E]%
\\
\>[B]{}\id{append}\;(\id{Cons}\;\id{x}\;\id{xs})\;{}\<[21]%
\>[21]{}\id{ys}{}\<[25]%
\>[25]{}\mathrel{=}{}\<[25E]%
\>[28]{}\id{Cons}\;\id{x}\;(\id{append}\;\id{xs}\;\id{ys}){}\<[E]%
\ColumnHook
\end{hscode}\resethooks
In CFC, the type of \ensuremath{\id{append}} would be rewritten to become
\begin{hscode}\SaveRestoreHook
\column{B}{@{}>{\hspre}l<{\hspost}@{}}%
\column{E}{@{}>{\hspre}l<{\hspost}@{}}%
\>[B]{}\id{append}\mathbin{::}\id{Plus}\;\id{m}\;\id{n}\,\sim\,\id{p}\Rightarrow \id{Vec}\;\id{a}\;\id{m}\to \id{Vec}\;\id{a}\;\id{n}\to \id{Vec}\;\id{a}\;\id{p}{}\<[E]%
\ColumnHook
\end{hscode}\resethooks
But now we have a problem.  In the \ensuremath{\id{Cons}} case, we have learned
that \ensuremath{\id{m}\,\sim\,\id{Succ}\;\id{m'}} for some \ensuremath{\id{m'}}; \ensuremath{\id{xs}} has type \ensuremath{\id{Vec}\;\id{a}\;\id{m'}}.
When we make the recursive call to \ensuremath{\id{append}}, we must provide a
\ensuremath{\id{p'}} such that \ensuremath{\id{Plus}\;\id{m'}\;\id{n}\,\sim\,\id{p'}}. However, there is no way to get such a \ensuremath{\id{p'}}
from the information to hand.

The solution to this problem is the $ \ottkw{assume} $ construct. The idea of
$\ottkw{assume} \, \chi \, \ottkw{in} \, \ottnt{e}$ is that we are allowed to assume that arbitrary applications of
a total type family reduce to proper types. Indeed, that's what \emph{total} means!

Let's now examine the typing rule for assumptions:
\[
\ottdruleEXXAssume{}
\]
This rule requires that the type family be total, according to the $\top$ subscript in the
$\ottmv{F} \,  \mathop{ :_{\top} }  \, \ottmv{n}  \in  \Sigma$ premise. It then checks the body $\ottnt{e}$ in a context where we have
a type $\alpha$ and coercion $\ottmv{c}$, as bound by $\chi$.  Finally, $\alpha$ is essentially
existential, so the rule also does a skolem escape check to assure that $\alpha$ does
not leak into the type of $\ottnt{e}$.

Discharging such assumptions is straightforward:
\[
\ottdruleSXXResolve{}
\]
When an $ \ottkw{assume} $ construct is ready to reduce, we are in an empty context---meaning that
all type variables have concrete values. At this point, we simply evaluate the type family
application at the concrete values. We are sure that this evaluation is possible, due to the
totality of the type function. The $ \ottmv{F} \  \overline{\tau}  \Downarrow  \ottnt{q} $ operation does the work for us, as defined in this
property of total type families:

\totaltypefamilies{}
This property must hold for any total type family, as accepted by any totality checker.

\subsection{Metatheory: Consistency of Equality}
\label{sec:cfc-consistency}

System CFC admits the usual preservation and progress theorems.

\begin{theorem}[Preservation]
\label{thm:preservation}
\presstatement
\end{theorem}

\begin{theorem}[Progress]
\label{thm:progress}
\progstatement
\end{theorem}

The proof of preservation is uninteresting. The hardest part is verifying that the push rules are
correct, but the only challenge is attention to detail. The unusual choice to make the context
empty in this proof is to support the \rul{S\_Resolve} rule, whose premise $ \ottmv{F} \  \overline{\tau}  \Downarrow  \ottnt{q} $ is well-defined
only in an empty context, according to \pref{prop:total}.

On the other hand, proving progress requires reasoning about the consistency of our equality
relation. This need arises in the case, among others, for \rul{E\_App}:
\[
\ottdrule{\ottpremise{\Gamma  \vdash  \ottnt{e_{{\mathrm{1}}}}  \ottsym{:}  \tau_{{\mathrm{1}}}  \to  \tau_{{\mathrm{2}}} \quad \hspace{-.6em} \quad \Gamma  \vdash  \ottnt{e_{{\mathrm{2}}}}  \ottsym{:}  \tau_{{\mathrm{1}}}}}{\Gamma  \vdash  \ottnt{e_{{\mathrm{1}}}} \, \ottnt{e_{{\mathrm{2}}}}  \ottsym{:}  \tau_{{\mathrm{2}}}}{\rul{E\_App}}
\]
We use the induction hypothesis to say that $\ottnt{e_{{\mathrm{1}}}}$ is a value $\ottnt{v_{{\mathrm{1}}}}$, a coerced value $\ottnt{v_{{\mathrm{1}}}}  \triangleright  \gamma$,
or steps to $\ottnt{e'_{{\mathrm{1}}}}$. In the case where $\ottnt{e_{{\mathrm{1}}}} \, \ottsym{=} \, \ottnt{v_{{\mathrm{1}}}}  \triangleright  \gamma$, we then wish to use \rul{S\_Push} to show
that the overall expression can step. However, this rule requires that $\ottnt{v_{{\mathrm{1}}}}$ have the form $\lambda  \ottmv{x}  \ottsym{:}  \tau  \ottsym{.}  \ottnt{e_{{\mathrm{0}}}}$.
The only way to show this is that the coercion $\gamma$ relates two functions.

The consistency lemma is what we need:
\begin{lemma}[Consistency]
\label{lem:consistency}
\consstatement
\end{lemma}
In an empty context and when two types are type family free, if they are related by a coercion, then they
must be the same. Using the following regularity lemma about expression typing, we can use consistency
in the proof of progress to finish the \rul{E\_App} case, among others.


\subsubsection{The route to consistency}
Broadly speaking, we prove consistency in the same manner as in previous work.\footnote{The best point
of comparison is with \citet{EisenbergVPJW14}, as that proof considers closed type families, as does
ours here.} First, we must restrict the set of available axioms to obey the following syntactic rules:

\begin{assumption}[Good signature]
\label{assn:good}
We assume that our implicit signature $\Sigma$ conforms to the following rules,
adapted from \citet[Definition 18]{EisenbergVPJW14}:
\begin{enumerate}
\item For all $\xi  \ottsym{:}  \overline{\ottnt{E} }  \in  \Sigma$ where
$\ottnt{E_{\ottmv{i}}} \, \ottsym{=} \, \forall \, \overline{\alpha}_{\ottmv{i}} \, \overline{\chi}_{\ottmv{i}}  \ottsym{.}  \ottmv{F_{\ottmv{i}}} \, \overline{\tau}_{\ottmv{i}}  \sim  \tau_{{\mathrm{0}}\,\ottmv{i}}$, there exists $\ottmv{F}$ such that,
for all $i$, $\ottmv{F_{\ottmv{i}}} \, \ottsym{=} \, \ottmv{F}$. That is, every equation listed within one axiom is over the same
type family $\ottmv{F}$.
\item For all $\xi  \ottsym{:}  \overline{\ottnt{E} }  \in  \Sigma$ where $\ottnt{E_{\ottmv{i}}} \, \ottsym{=} \, \forall \, \overline{\alpha}_{\ottmv{i}} \, \overline{\chi}_{\ottmv{i}}  \ottsym{.}  \ottmv{F_{\ottmv{i}}} \, \overline{\tau}_{\ottmv{i}}  \sim  \tau_{{\mathrm{0}}\,\ottmv{i}}$, for
all $\ottmv{i}$, $\mathit{fv}\! \, \ottsym{(}  \overline{\tau}_{\ottmv{i}}  \ottsym{)} \, \ottsym{=} \, \overline{\alpha}_{\ottmv{i}}$.
That is, every
quantified type variable in an equation is mentioned free in a type on the equation's
left-hand side.
\item For all $\xi  \ottsym{:}  \overline{\ottnt{E} }  \in  \Sigma$, if $\mathsf{length}(\overline{\ottnt{E} }) > 1$ and the equations are over
type family $\ottmv{F}$, then no other axiom $\xi'  \ottsym{:}  \overline{\ottnt{E} }'  \in  \Sigma$ is over the same type
family $\ottmv{F}$. That is, all axioms with multiple equations are for \emph{closed}
type families.
\item For all $\xi_{{\mathrm{1}}}  \ottsym{:}  \ottnt{E_{{\mathrm{1}}}}  \in  \Sigma$ and $\xi_{{\mathrm{2}}}  \ottsym{:}  \ottnt{E_{{\mathrm{2}}}}  \in  \Sigma$ (each with only one equation),
if $\ottnt{E_{{\mathrm{1}}}}$ and $\ottnt{E_{{\mathrm{2}}}}$ are over the same type family $\ottmv{F}$, then
$ \mathsf{compat} ( \ottnt{E_{{\mathrm{1}}}} ,  \ottnt{E_{{\mathrm{2}}}} ) $. That is, equations for open type families are all pairwise compatible.
\end{enumerate}
\end{assumption}

The conditions above are identical to the conditions in \citet[Definition 18]{EisenbergVPJW14},
but with one change: we here do not need to restrict the left-hand types of equations
not to mention type families, because of the $\ottcompu{ \Gamma   \vdash   \tau_{\ottmv{i}}  \ \mathsf{type} }{\ottmv{i}}{\ottmv{n}}$ premise
to \rul{D\_Axiom} describing the validity of axioms in the signature. Type family
applications
are not types.

Then, we define a non-deterministic rewrite relation on types $\tau_{{\mathrm{1}}}  \rightsquigarrow  \tau_{{\mathrm{2}}}$
and prove both of the following:

\begin{lemma}[Completeness of the rewrite relation]
\label{lem:complete-red}
\completestatement
\end{lemma}

\begin{lemma}[Proper types do not reduce]
\label{lem:type-no-red}
\noredstatement
\end{lemma}

Taken together, these quickly prove the consistency lemma.

\subsubsection{Type reduction relation}

\begin{figure}
\begin{minipage}{.74\textwidth}
\begin{center}
\begin{ottdefnblock}{$\tau_{{\mathrm{1}}}  \rightsquigarrow_{\top}  \tau_{{\mathrm{2}}}$}{Type family application reduction}
\[
\ottdrule{\begin{array}{@{}c@{}c@{}}
\begin{array}{@{}l@{}}
\xi  \ottsym{:}  \overline{\ottnt{E} }  \in  \Sigma \qquad \overline{\tau} \, \ottsym{=} \, \overline{\sigma}  \ottsym{[}  \overline{\rho}  \ottsym{/}  \overline{\alpha}  \ottsym{]} \\
\ottcomp{ \emptyset   \vdash   \rho_{\ottmv{k}}  \ \mathsf{type} }{\ottmv{k}} \\
\ottnt{E_{\ottmv{i}}} \, \ottsym{=} \, \forall \, \overline{\alpha} \, \overline{\chi}  \ottsym{.}  \ottmv{F} \, \overline{\sigma}  \sim  \sigma_{{\mathrm{0}}}  \\
\tau' \, \ottsym{=} \, \sigma_{{\mathrm{0}}}  \ottsym{[}  \overline{\rho}  \ottsym{/}  \overline{\alpha}  \ottsym{,}  \overline{\rho}'  \ottsym{/}  \mathit{tvs}\! \, \ottsym{(}  \overline{\chi}  \ottsym{)}  \ottsym{]} \\
\forall \, \ottmv{j}  \ottsym{<}  \ottmv{i}  \ottsym{,}   \mathsf{no\_conflict} ( \overline{\ottnt{E} } ,  \ottmv{i} ,  \overline{\rho} ,  \ottmv{j} )
\end{array}
&
\begin{array}{@{}r@{\;}l@{}}
\forall \, \ottmv{n}  \ottsym{:} & \chi_{\ottmv{n}} \, \ottsym{=} \,  (  \alpha'  \pipe  \ottmv{c'}  :  \ottmv{F'} \, \overline{\sigma}'   \sim   \alpha'  )  \\
&  \emptyset   \vdash   \rho'_{\ottmv{n}}  \ \mathsf{type}  \\
& \theta_{\ottmv{n}} \, \ottsym{=} \, \overline{\rho}  \ottsym{/}  \overline{\alpha}  \ottsym{,} \, \ottcomplu{\rho'_{\ottmv{m}}  \ottsym{/}  \mathit{tv}\! \, \ottsym{(}  \chi_{\ottmv{m}}  \ottsym{)}}{\ottmv{m}}{{\mathrm{1}}}{..}{{\ottmv{n}-1}} \\
& \ottmv{F'} \, \overline{\sigma}'  \ottsym{[}  \theta_{\ottmv{n}}  \ottsym{]}  \rightsquigarrow_{\top}  \rho'_{\ottmv{n}} \\
\end{array}
\end{array}}{\ottmv{F} \, \overline{\tau}  \rightsquigarrow_{\top}  \tau'}{\rul{RTop}}
\]
\end{ottdefnblock}
\end{center}
\end{minipage}
\vrule
\begin{minipage}{.24\textwidth}
\begin{center}
\ottdefnRed{}
\end{center}
\end{minipage}
\caption{Non-deterministic type reduction}
\label{fig:rewrite-rel}
\end{figure}

The type reduction relation $ \rightsquigarrow $ is captured by the judgments in \pref{fig:rewrite-rel}.
Rule \rul{Red} says that a type $\sigma$ can reduce by reducing a type family application
occurring anywhere within $\sigma$. (The metavariable $\mathcal{C}$ denotes one-hole type contexts.)
The intimidating \rul{RTop} rule matches up with \rul{C\_Axiom}. The complication in the
rule is in dealing with the evaluation assumptions $\overline{\chi}$ in a given type family equation; each
needs to be satisfied with an evaluation resolution of a type paired with a coercion. The premises
under the $\forall \, \ottmv{n}  \ottsym{:}$ roughly simulate the $\Gamma  \vdash  \overline{\ottnt{q} }  \ottsym{:}  \overline{\chi}$ judgment.

Unlike in prior proofs of the consistency of versions of System FC, when $\tau_{{\mathrm{1}}}  \rightsquigarrow  \tau_{{\mathrm{2}}}$, there must
be precisely one fewer type family application in $\tau_{{\mathrm{2}}}$ than in $\tau_{{\mathrm{1}}}$. This fact is borne of
the use of evaluation assumptions $\overline{\chi}$ to model type family applications in the right-hand side
of a type family equation instead of using type families there directly. It leads to this critical lemma:

\begin{lemma}[Termination]
\label{lem:termination}
\termstatement
\end{lemma}

The fact that the reduction relation terminates means that we can use Newman's Lemma to prove confluence
via local confluence,
a necessary precursor to the proof of the completeness of the rewrite relation (\pref{lem:complete-red}):

\begin{lemma}[Local confluence]
\label{lem:local-confluence}
\localconfstatement
\end{lemma}

\begin{lemma}[Confluence]
\label{lem:confluence}
\confstatement
\end{lemma}

\citet{EisenbergVPJW14} also prove confluence via local confluence, but that proof
must assume termination. The formulation here allows us to prove termination instead of assume it. The local confluence proof
in the current work is also a simplification over the previous proof, as the location of occurrences
of type family applications is restricted.

\paragraph{Conclusion}
By using evaluation assumptions in our treatment of type families, we can easily prove the termination of
type reduction and simplify the proof of confluence. The intricate definition of apartness from
\citet{EisenbergVPJW14} is gone, as well. In short, our approach leads to a substantial simplification
to the metatheory of type families.

\section{Practicalities}
\label{sec:compat}

We believe that constrained type families provide significant benefits compared to the previous approach
to type families, with its underlying, implicit assumption of totality.
As we are changing the type system of a language,
not all current Haskell code is immediately supported in our design.
For example, existing code may make use of non-associated open type
families, or use incomplete type families as if they were total.  In this section, we describe an
approach for inferring constrained type families, and the corresponding constraints, from current
declarations and uses of indexed type families.  This is intended to allow a transition from current
practice to the explicit use of constrained type families.

\subsection{Inferring Type Family Constraints}

We first consider uses of type families in types.  Here, our approach is to read the well-formedness
restrictions for constrained type families~\secref{valid-ctf} as inference rules rather than as a
checking relation.  Because the typing rules are syntax directed, given type environments $\Sigma$
and $\Gamma$ (known in advance), and a type $\tau$, we can follow the rules to generate a
$P$ such that $ P  \mid  \Gamma   \vdash   \tau  \ \mathsf{type} $, if such a $P$ exists.  While there is not
necessarily a unique $P$ such that the derivation exists, it is easy to pick a minimal
one such that it does.  (In essence, we view the well-formedness rules as an attribute
grammar, in which $\Sigma,\Gamma$ and $\tau$ are given, and $P$ is synthesized.)  Then, we
interpret each qualified type $\sigma$ in context $\Gamma$ in the program as instead denoting the type
$P \Rightarrow \sigma$ where $P$ is the minimal set of additional constraints such that
$ P  \mid  \Gamma   \vdash   \sigma  \ \mathsf{type} $.  Note that some programs may still fail to type check under this approach,
if they explicitly make use of undefined type family applications.  However, we view this as an
acceptable trade-off, as those programs arguably already contained (admittedly unreported) type
errors.

\subsection{Making Associations}

We must also interpret top-level type family syntax in terms of constrained type families.  Type
family declarations themselves can be straightforwardly interpreted as declarations of constrained
type families; for example, the declaration
\begin{hscode}\SaveRestoreHook
\column{B}{@{}>{\hspre}l<{\hspost}@{}}%
\column{E}{@{}>{\hspre}l<{\hspost}@{}}%
\>[B]{}\keyword{type}\;\keyword{family}\;\id{F}\;\id{t}\;\id{u}\mathbin{::}\star{}\<[E]%
\ColumnHook
\end{hscode}\resethooks
would be interpreted as
\begin{hscode}\SaveRestoreHook
\column{B}{@{}>{\hspre}l<{\hspost}@{}}%
\column{3}{@{}>{\hspre}l<{\hspost}@{}}%
\column{E}{@{}>{\hspre}l<{\hspost}@{}}%
\>[B]{}\keyword{class}\;\id{CF}\;\id{t}\;\id{u}\;\keyword{where}{}\<[E]%
\\
\>[B]{}\hsindent{3}{}\<[3]%
\>[3]{}\keyword{type}\;\id{F}\;\id{t}\;\id{u}\mathbin{::}\star{}\<[E]%
\ColumnHook
\end{hscode}\resethooks
where any other kind restrictions in the original declaration of \ensuremath{\id{F}} can be transferred
straightforwardly to the declaration of \ensuremath{\id{CF}}. Connecting \ensuremath{\id{F}} to the compiler-generated \ensuremath{\id{CF}}
would be a new special form \ensuremath{(\keyword{class}\;\id{F})}, entirely equivalent to \ensuremath{\id{CF}}.

Instance declarations are more interesting.  For
example, consider the instance declaration
\begin{hscode}\SaveRestoreHook
\column{B}{@{}>{\hspre}l<{\hspost}@{}}%
\column{E}{@{}>{\hspre}l<{\hspost}@{}}%
\>[B]{}\keyword{type}\;\keyword{instance}\;\id{F}\;\id{Int}\;(\id{Maybe}\;\id{t})\mathrel{=}\id{G}\;\id{Int}\;\id{t}{}\<[E]%
\ColumnHook
\end{hscode}\resethooks
where we assume that \ensuremath{\id{G}} is a binary type family.  We could not simply
interpret this as the instance declaration
\begin{hscode}\SaveRestoreHook
\column{B}{@{}>{\hspre}l<{\hspost}@{}}%
\column{3}{@{}>{\hspre}l<{\hspost}@{}}%
\column{E}{@{}>{\hspre}l<{\hspost}@{}}%
\>[B]{}\keyword{instance}\;\id{CF}\;\id{Int}\;(\id{Maybe}\;\id{t})\;\keyword{where}{}\<[E]%
\\
\>[B]{}\hsindent{3}{}\<[3]%
\>[3]{}\keyword{type}\;\id{F}\;\id{Int}\;(\id{Maybe}\;\id{t})\mathrel{=}\id{G}\;\id{Int}\;\id{t}{}\<[E]%
\ColumnHook
\end{hscode}\resethooks
as the use of type family \ensuremath{\id{G}} lacks a suitable guarding constraint.  Again, however, we can rely on
interpreting the well-formedness rules for types to infer the necessary constraints.  In this case,
we would interpret the type instance as denoting the instance declaration
\begin{hscode}\SaveRestoreHook
\column{B}{@{}>{\hspre}l<{\hspost}@{}}%
\column{3}{@{}>{\hspre}l<{\hspost}@{}}%
\column{E}{@{}>{\hspre}l<{\hspost}@{}}%
\>[B]{}\keyword{instance}\;\id{P}\Rightarrow \id{CF}\;\id{Int}\;(\id{Maybe}\;\id{t})\;\keyword{where}{}\<[E]%
\\
\>[B]{}\hsindent{3}{}\<[3]%
\>[3]{}\keyword{type}\;\id{F}\;\id{Int}\;(\id{Maybe}\;\id{t})\mathrel{=}\id{G}\;\id{Int}\;\id{t}{}\<[E]%
\ColumnHook
\end{hscode}\resethooks
where \ensuremath{\id{P}} is the minimal set of constraints such that $P \mid \ensuremath{\id{t}} \vdash \ensuremath{\id{G}\;\id{Int}\;\id{t}} \
\mathsf{type}$ holds.  Again, so long as the original type instance declaration did not rely on
undefined type family applications, the resulting instance declaration will be well-formed.

Finally, we turn to closed type families.  Given a closed type family declaration, we
initially check its totality~\secref{totality-checking}.  If it is not total, we can then interpret
it as a constrained closed type family, following the same approach as for open type families.  For
example, consider the following closed type family declaration:
\begin{hscode}\SaveRestoreHook
\column{B}{@{}>{\hspre}l<{\hspost}@{}}%
\column{3}{@{}>{\hspre}l<{\hspost}@{}}%
\column{18}{@{}>{\hspre}c<{\hspost}@{}}%
\column{18E}{@{}l@{}}%
\column{21}{@{}>{\hspre}l<{\hspost}@{}}%
\column{E}{@{}>{\hspre}l<{\hspost}@{}}%
\>[B]{}\keyword{type}\;\keyword{family}\;\id{F}\;\id{t}\mathbin{::}\star\;\keyword{where}{}\<[E]%
\\
\>[B]{}\hsindent{3}{}\<[3]%
\>[3]{}\id{F}\;(\id{Maybe}\;\id{Int}){}\<[18]%
\>[18]{}\mathrel{=}{}\<[18E]%
\>[21]{}\id{Bool}{}\<[E]%
\\
\>[B]{}\hsindent{3}{}\<[3]%
\>[3]{}\id{F}\;(\id{Maybe}\;\id{t}){}\<[18]%
\>[18]{}\mathrel{=}{}\<[18E]%
\>[21]{}\id{G}\;\id{t}{}\<[E]%
\ColumnHook
\end{hscode}\resethooks
This declaration is clearly not total.  We would interpret this as a closed type family declaration:
\begin{hscode}\SaveRestoreHook
\column{B}{@{}>{\hspre}l<{\hspost}@{}}%
\column{3}{@{}>{\hspre}l<{\hspost}@{}}%
\column{5}{@{}>{\hspre}l<{\hspost}@{}}%
\column{E}{@{}>{\hspre}l<{\hspost}@{}}%
\>[B]{}\keyword{class}\;\id{CF}\;\id{t}\;\keyword{where}{}\<[E]%
\\
\>[B]{}\hsindent{3}{}\<[3]%
\>[3]{}\keyword{type}\;\id{F}\;\id{t}\mathbin{::}\star{}\<[E]%
\\[\blanklineskip]%
\>[B]{}\hsindent{3}{}\<[3]%
\>[3]{}\keyword{instance}\;\id{CF}\;(\id{Maybe}\;\id{Int})\;\keyword{where}{}\<[E]%
\\
\>[3]{}\hsindent{2}{}\<[5]%
\>[5]{}\keyword{type}\;\id{F}\;(\id{Maybe}\;\id{Int})\mathrel{=}\id{Bool}{}\<[E]%
\\[\blanklineskip]%
\>[B]{}\hsindent{3}{}\<[3]%
\>[3]{}\keyword{instance}\;\id{P}\Rightarrow \id{CF}\;(\id{Maybe}\;\id{t})\;\keyword{where}{}\<[E]%
\\
\>[3]{}\hsindent{2}{}\<[5]%
\>[5]{}\keyword{type}\;\id{F}\;(\id{Maybe}\;\id{t})\mathrel{=}\id{G}\;\id{t}{}\<[E]%
\ColumnHook
\end{hscode}\resethooks
where \ensuremath{\id{P}} is the minimal set of constraints such that $P \mid t \vdash \ensuremath{\id{G}\;\id{t}} \ \mathsf{type}$ holds.

The decision of whether or not to treat a top-level closed type family as constrained is based
on the output from the totality checker. We expect users will want to override the compiler's
decision in this matter, as any totality checker will be incomplete. We propose the new syntax
\ensuremath{\keyword{type}\;\keyword{family}\;\keyword{total}\;\id{F}\;\id{a}\;\keyword{where}\mathbin{...}} to denote that \ensuremath{\id{F}} is intended to be total. Such a declaration
would still be checked, but would never be packaged into an enclosing class. (A non-total definition
would be reported as an error.) The user could additionally add a pragma \{-\#~TOTAL~\ensuremath{\id{F}}~\#-\} to (unsafely)
assert that \ensuremath{\id{F}} is total, circumventing the totality checker.




\subsection{Runtime Efficiency}

Constrained type families may also seem to have a non-trivial efficiency impact.  For a simple example,
suppose we have a type family \ensuremath{\id{F}}, and consider an existentially-packaged type family application:
\begin{hscode}\SaveRestoreHook
\column{B}{@{}>{\hspre}l<{\hspost}@{}}%
\column{3}{@{}>{\hspre}l<{\hspost}@{}}%
\column{E}{@{}>{\hspre}l<{\hspost}@{}}%
\>[B]{}\keyword{data}\;\id{FPack}\;\id{a}\;\keyword{where}{}\<[E]%
\\
\>[B]{}\hsindent{3}{}\<[3]%
\>[3]{}\id{FPack}\mathbin{::}\id{F}\;\id{a}\to \id{FPack}\;\id{a}{}\<[E]%
\ColumnHook
\end{hscode}\resethooks
We might expect an \ensuremath{\id{FPack}\;\id{a}} value to contain exactly a value of type \ensuremath{\id{F}\;\id{a}}.  With constrained type
families, however, the declaration above would be incorrect; we would need to add a predicate for its
constraining class, say \ensuremath{\id{C}}:
\begin{hscode}\SaveRestoreHook
\column{B}{@{}>{\hspre}l<{\hspost}@{}}%
\column{3}{@{}>{\hspre}l<{\hspost}@{}}%
\column{E}{@{}>{\hspre}l<{\hspost}@{}}%
\>[B]{}\keyword{data}\;\id{FPack1}\;\id{a}\;\keyword{where}{}\<[E]%
\\
\>[B]{}\hsindent{3}{}\<[3]%
\>[3]{}\id{FPack1}\mathbin{::}\id{C}\;\id{a}\Rightarrow \id{F}\;\id{a}\to \id{FPack}\;\id{a}{}\<[E]%
\ColumnHook
\end{hscode}\resethooks
Now, a value of type \ensuremath{\id{FPack1}\;\id{a}} does not just contain an \ensuremath{\id{F}\;\id{a}} value, but must also carry a \ensuremath{\id{C}\;\id{a}}
dictionary, and uses of \ensuremath{\id{FPack1}} will be responsible for constructing, packing, and unpacking these
dictionaries.  Over sufficiently many uses of \ensuremath{\id{FPack1}}, this additional cost could be noticeable.

This efficiency impact can be mitigated, however. This issue can crop up only when we have a value
of type \ensuremath{\id{F}\;\id{a}} (or other type family application) without an instance of the associated class \ensuremath{\id{C}\;\id{a}}.
But in order for the value of type \ensuremath{\id{F}\;\id{a}} to be useful, parametricity tells us that \ensuremath{\id{C}\;\id{a}}, or some
other class with a similar structure to the equations for \ensuremath{\id{F}\;\id{a}} must be in scope. Barring this,
it must be that \ensuremath{\id{F}\;\id{a}} is used as a phantom type. In this case, we would want a ``phantom dictionary''
for \ensuremath{\id{C}\;\id{a}}, closely paralleling existing work on proof irrelevance in the dependently-typed programming
community~(e.g., \citet{MishraLingerS08,BarrasB08,TejiscakB15,Eisenberg16}): the \ensuremath{\id{C}\;\id{a}} dictionary essentially
represents a proof that will never be
examined.  While we do not propose here a new solution to this problem, we believe that existing
work will be applicable in our case as well.

\section{Related Work}
\label{sec:related}

The literature on type-level computation and the type system of Haskell is extensive; here, we
summarize those parts most relevant to our work.

\paragraph{Type classes and functional dependencies}

Partial type-level computation in Haskell was arguably first introduced with Jones's notion of
functional dependencies~\cite{Jones00}, which extended type classes with a notion of determined
parameters. Indeed our treatment of requiring a class constraint to use type-level computation
is inspired by functional dependencies. Functional dependencies build on Jones's theory of improvement for qualified
types~\cite{Jones95}, which allows the satisfiability of predicates to influence typing.  While
Jones's work does not focus on the computational interpretation of functional dependencies, many early
examples highlighted it, such as those of \citet{Hallgren01} or
\citet{KiselyovLS04}.  \citet{MorrisJ10} later introduced instance chains---closely related
to our closed type classes---which
combined functional dependencies with explicit notions of negation and alternatives in class
instances.

\paragraph{Associated types and type families}

\citet{ChakravartyKJ05} introduced associated type synonyms to provide a more intuitive syntax for
type-level computation in Haskell, while also addressing infelicities in the implementations of
functional dependencies.  Their type system requires that associated types appear only in contexts
where their class predicates can be satisfied, matching our approach. However, this
requirement was never implemented; GHC's translation to System~FC~\cite{SulzmannCJD07} showed
that the constraint was never used at runtime and was thus deemed superfluous.
The class constraints---that is, instance dictionaries~\cite{HallHPJW96}---are not needed at runtime,
in contrast to ordinary class method invocation. Our work does not refute this conclusion, but instead
observes that the design of type families and their metatheory are greatly simplified when we require
the class constraint.


Recent work has focused on extending the expressiveness of type families themselves.
\citet{EisenbergVPJW14} introduced closed type families, which allow overlapping equations in
type family definitions, and \citet{StolarekJE15} introduced injective type families,
recovering additional equalities from applications of injective type families.  These features,
particularly closed type families, have seen significant practical application.



\paragraph{Type classes and modules}

An alternative approach to supporting type classes directly is to encode them using
modules~\citep{Dreyer07} or objects~\citep{OLiveiraMO10}.  These approaches replace class predicates
with module (or object) arguments, and use mechanisms for canonical values or implicit arguments to
simulate instance resolution.  Associated types arise naturally in these approaches, as type members
of modules, and, as in our approach, can only appear when a suitable module is in scope.  However,
these approaches require significantly different underlying formalisms, and so it is not apparent
how well they would accommodate other extensions, like closed and total type families, or closed
classes.

\paragraph{Partial functions in logic}

An interesting---and unexpected---parallel to our work arises in Scott's examination of identity and
existence in intuitionistic logic~\cite{Scott79}.  Scott considers the cases in which (first-order)
terms in a logic may not be defined for arbitrary instantiations of their variables.  For example,
the term $1/a$ is not defined if $a$ is instantiated to $0$.  Scott addresses this problem by
introducing an additional predicate $E(\cdot)$ to track the existence of first-order terms, which
plays a similar role to our requirement that uses of constrained type families mention their
defining class predicates.

\section{Conclusions}
\label{sec:conclusion}

We have presented a new approach to type-level computation, relevant to any partial language,
in which we permit partiality in types
by using qualified types to capture their domains of definition.
We have applied our approach to
indexed type families in Haskell, showing that it aligns naturally with the intuitive semantics of
type families and that it resolves many of the complexities in recent developments of type families.
We have formalized our approach, and given the first complete proof of consistency for Haskell with
closed type families.

Since their introduction, the theory and practice of functional
dependencies and type families have diverged, although some uses of functional dependencies continue
to seem more expressive than similar uses of type families.  Our current work reunites type families
with type classes. We believe it should provide an impetus to re-examine the role of functional
dependencies.  In particular, the use of equality constraints in our core language to prove that type
families applications are well-defined is evocative of the role that class predicates would play in
a core calculus based on functional dependencies.

\begin{acks}
Thanks to the anonymous referees for their helpful feedback.
Morris was funded by \grantsponsor{}{EPSRC}{https://epsrc.ac.uk} grant number \grantnum{}{EP/K034413/1}.
\end{acks}


\ifextended
\pagebreak
\appendix

\section{Proofs}

\subsection{Evaluation assumptions}
\label{sec:eval-assump-defns}

\begin{definition}[Evaluation assumption substitution]
Define the substitution $\ottnt{q}  \ottsym{/}  \chi$ to mean $\sigma  \ottsym{/}  \alpha  \ottsym{,}  \gamma  \ottsym{/}  \ottmv{c}$, where
$\ottnt{q} \, \ottsym{=} \,  ( \sigma  \pipe  \gamma ) $ and $\chi \, \ottsym{=} \,  (  \alpha  \pipe  \ottmv{c}  :  \ottmv{F} \, \overline{\tau}   \sim   \alpha  ) $.
\end{definition}

\begin{definition}[Evaluation assumption scoping]
Define $\mathit{fv}\! \, \ottsym{(}  \overline{\chi}  \ottsym{)}$ inductively as follows:
\begin{align*}
\mathit{fv}\! \, \ottsym{(}  \emptyset  \ottsym{)} \, &= \, \emptyset \\
\mathit{fv}\! \, \ottsym{(}   (  \alpha  \pipe  \ottmv{c}  :  \ottmv{F} \, \overline{\tau}   \sim   \alpha  )   \ottsym{,}  \overline{\chi}  \ottsym{)} \, &= \, \mathit{fv}\! \, \ottsym{(}  \overline{\tau}  \ottsym{)}  \cup  \ottsym{(}   \mathit{fv}\! \, \ottsym{(}  \overline{\chi}  \ottsym{)}  \backslash  \ottsym{\{}  \alpha  \ottsym{\}}   \ottsym{)}
\end{align*}
\end{definition}

\begin{definition}[Evaluation assumption substitution]
Define $\mathit{subst}\! \, \ottsym{(}  \overline{\chi}  \ottsym{)}$ inductively as follows:
\begin{align*}
\mathit{subst}\! \, \ottsym{(}  \emptyset  \ottsym{)} \, &= \, \emptyset \\
\mathit{subst}\! \, \ottsym{(}  \overline{\chi}  \ottsym{,}   (  \alpha  \pipe  \ottmv{c}  :  \ottmv{F} \, \overline{\tau}   \sim   \alpha  )   \ottsym{)} \, &= \, \mathit{subst}\! \, \ottsym{(}  \overline{\chi}  \ottsym{)}  \circ  \ottsym{(}  \ottmv{F} \, \overline{\tau}  \ottsym{)}  \ottsym{/}  \alpha
\end{align*}
\end{definition}

\subsection{Assumptions about environment}

\begin{assumption}[Declarations]
\label{assn:declarations}
We assume that if $\ottnt{decl}  \in  \Sigma$, then $ \vdash   \ottnt{decl}  \ \mathsf{ok} $.
\end{assumption}

Furthermore, we assume \pref{assn:good}.

\subsection{Unification}

\begin{property}[$ \mathsf{unify} $ correct]
\label{prop:unify-correct}
(\citet[Property 11]{EisenbergVPJW14})
If there exists a substitution $\theta$ such that $\overline{\sigma}  \ottsym{[}  \theta  \ottsym{]} \, \ottsym{=} \, \overline{\tau}  \ottsym{[}  \theta  \ottsym{]}$,
then $ \mathsf{unify} ( \overline{\sigma} ;\, \overline{\tau} ) $ succeeds. If $ \mathsf{unify} ( \overline{\sigma} ;\, \overline{\tau} )  \, \ottsym{=} \, \mathsf{Just} \, \theta$, then
$\theta$ is a most general unifier of $\overline{\sigma}$ and $\overline{\tau}$.
\end{property}

\begin{lemma}[Apartness is stable under type substitution]
\label{lem:apart-subst}
(\citet[Property 12]{EisenbergVPJW14}) If $ \mathsf{unify} ( \overline{\tau}_{{\mathrm{1}}} ;\, \overline{\tau}_{{\mathrm{2}}} )  \, \ottsym{=} \, \mathsf{Nothing}$, then for
all substitutions $\theta$, $ \mathsf{unify} ( \overline{\tau}_{{\mathrm{1}}} ;\, \overline{\tau}_{{\mathrm{2}}}  \ottsym{[}  \theta  \ottsym{]} )  \, \ottsym{=} \, \mathsf{Nothing}$.
\end{lemma}

\begin{proof}
We prove the contrapositive: that if $ \mathsf{unify} ( \overline{\tau}_{{\mathrm{1}}} ;\, \overline{\tau}_{{\mathrm{2}}}  \ottsym{[}  \theta  \ottsym{]} )  \, \ottsym{=} \, \mathsf{Just} \, \theta'$, then
there exists $\theta''$ such that $ \mathsf{unify} ( \overline{\tau}_{{\mathrm{1}}} ;\, \overline{\tau}_{{\mathrm{2}}} )  \, \ottsym{=} \, \mathsf{Just} \, \theta''$. By
\pref{prop:unify-correct}, we have
$\overline{\tau}_{{\mathrm{1}}}  \ottsym{[}  \theta'  \ottsym{]} \, \ottsym{=} \, \overline{\tau}_{{\mathrm{2}}}  \ottsym{[}  \theta'  \circ  \theta  \ottsym{]}$. Since we assume that $\mathit{fv}\! \, \ottsym{(}  \overline{\tau}_{{\mathrm{1}}}  \ottsym{)}  \cap  \mathit{fv}\! \, \ottsym{(}  \overline{\tau}_{{\mathrm{2}}}  \ottsym{)} \, \ottsym{=} \, \emptyset$,
we can rewrite this as $\overline{\tau}_{{\mathrm{1}}}  \ottsym{[}  \theta'  \circ  \theta  \ottsym{]} \, \ottsym{=} \, \overline{\tau}_{{\mathrm{2}}}  \ottsym{[}  \theta'  \circ  \theta  \ottsym{]}$. Thus,
$\theta'' \, \ottsym{=} \, \theta'  \circ  \theta$ and we are done.
\end{proof}

\subsection{Structural properties}

\begin{definition}[Subset on contexts]
Define $\Gamma  \subseteq  \Gamma'$ to mean that $\Gamma'$ contains at least all the bindings
in $\Gamma$, possibly in a different order.
\end{definition}

\begin{lemma}[Weakening/Permutation] \leavevmode
\label{lem:weakening}
Suppose $\Gamma  \subseteq  \Gamma'$ and $ \vdash   \Gamma'  \ \mathsf{ctx} $.
\begin{enumerate}
\item If $ \Gamma   \vdash   \tau  \ \mathsf{type} $, then $ \Gamma'   \vdash   \tau  \ \mathsf{type} $.
\item If $ \Gamma   \vdash   \phi  \ \mathsf{prop} $, then $ \Gamma'   \vdash   \phi  \ \mathsf{prop} $.
\item If $\Gamma  \vdash  \gamma  \ottsym{:}  \phi$, then $\Gamma'  \vdash  \gamma  \ottsym{:}  \phi$.
\item If $\Gamma  \vdash  \overline{\ottnt{q} }  \ottsym{:}  \overline{\chi}$, then $\Gamma'  \vdash  \overline{\ottnt{q} }  \ottsym{:}  \overline{\chi}$.
\item If $\Gamma  \vdash  \ottmv{x}  \ottsym{:}  \tau$, then $\Gamma'  \vdash  \ottmv{x}  \ottsym{:}  \tau$.
\end{enumerate}
\end{lemma}

\begin{proof}
By straightforward mutual induction, renaming bound variables
if necessary to satisfy freshness conditions.
\end{proof}

\begin{lemma}[Context strengthening] \leavevmode
\label{lem:strengthening}
Suppose $ \mathit{dom}\! ( \delta )   \not\in  \mathit{fv}\! \, \ottsym{(}  \Gamma'  \ottsym{)}$.
\begin{enumerate}
\item If $ \Gamma  \ottsym{,}  \delta  \ottsym{,}  \Gamma'   \vdash   \tau  \ \mathsf{type} $, and $ \mathit{dom}\! ( \delta )   \not\in  \mathit{fv}\! \, \ottsym{(}  \tau  \ottsym{)}$, then $ \Gamma  \ottsym{,}  \Gamma'   \vdash   \tau  \ \mathsf{type} $.
\item If $ \Gamma  \ottsym{,}  \delta  \ottsym{,}  \Gamma'   \vdash   \phi  \ \mathsf{prop} $ and $ \mathit{dom}\! ( \delta )   \not\in  \mathit{fv}\! \, \ottsym{(}  \phi  \ottsym{)}$, then $ \Gamma  \ottsym{,}  \Gamma'   \vdash   \phi  \ \mathsf{prop} $.
\item If $ \vdash   \Gamma  \ottsym{,}  \delta  \ottsym{,}  \Gamma'  \ \mathsf{ctx} $, then $ \vdash   \Gamma  \ottsym{,}  \Gamma'  \ \mathsf{ctx} $.
\end{enumerate}
\end{lemma}

\begin{proof}
Straightforward mutual induction.
\end{proof}

\begin{lemma}[Scoping] \leavevmode
\label{lem:scoping}
\begin{enumerate}
\item If $ \Gamma   \vdash   \tau  \ \mathsf{type} $, then $\mathit{fv}\! \, \ottsym{(}  \tau  \ottsym{)}  \subseteq  \mathit{dom}\! \, \ottsym{(}  \Gamma  \ottsym{)}$.
\item If $ \Gamma   \vdash   \phi  \ \mathsf{prop} $, then $\mathit{fv}\! \, \ottsym{(}  \phi  \ottsym{)}  \subseteq  \mathit{dom}\! \, \ottsym{(}  \Gamma  \ottsym{)}$.
\item If $\Gamma  \vdash  \gamma  \ottsym{:}  \phi$, then $\mathit{fv}\! \, \ottsym{(}  \gamma  \ottsym{)}  \subseteq  \mathit{dom}\! \, \ottsym{(}  \Gamma  \ottsym{)}$ and $\mathit{fv}\! \, \ottsym{(}  \phi  \ottsym{)}  \subseteq  \mathit{dom}\! \, \ottsym{(}  \Gamma  \ottsym{)}$.
\item If $\Gamma  \vdash  \overline{\ottnt{q} }  \ottsym{:}  \overline{\chi}$, then $\mathit{fv}\! \, \ottsym{(}  \overline{\ottnt{q} }  \ottsym{)}  \subseteq  \mathit{dom}\! \, \ottsym{(}  \Gamma  \ottsym{)}$ and $\mathit{fv}\! \, \ottsym{(}  \overline{\chi}  \ottsym{)}  \subseteq  \mathit{dom}\! \, \ottsym{(}  \Gamma  \ottsym{)}$.
\item If $\Gamma  \vdash  \ottnt{e}  \ottsym{:}  \tau$, then $\mathit{fv}\! \, \ottsym{(}  \ottnt{e}  \ottsym{)}  \subseteq  \mathit{dom}\! \, \ottsym{(}  \Gamma  \ottsym{)}$ and $\mathit{fv}\! \, \ottsym{(}  \tau  \ottsym{)}  \subseteq  \mathit{dom}\! \, \ottsym{(}  \Gamma  \ottsym{)}$.
\item If $ \vdash   \Gamma  \ \mathsf{ctx} $, then $\mathit{fv}\! \, \ottsym{(}  \Gamma  \ottsym{)} \, \ottsym{=} \, \emptyset$.
\end{enumerate}
\end{lemma}

\begin{proof}
By induction, using a mutual induction between $ \Gamma   \vdash   \tau  \ \mathsf{type} $, $ \Gamma   \vdash   \phi  \ \mathsf{prop} $, followed
by the others. We appeal to \pref{assn:declarations} in the \rul{C\_Axiom} case.
\end{proof}

\subsection{Regularity I}

\begin{lemma}[Context formation] \leavevmode
\label{lem:ctx-form}
If $ \vdash   \Gamma  \ \mathsf{ctx} $ and $\Gamma'$ is a prefix of $\Gamma$, then $ \vdash   \Gamma'  \ \mathsf{ctx} $.
\end{lemma}

\begin{proof}
Straightforward induction on $\Gamma$.
\end{proof}

\begin{lemma}[Context regularity] \leavevmode
\label{lem:ctx-reg}
If any of
\begin{enumerate}
\item $ \Gamma   \vdash   \tau  \ \mathsf{type} $, or
\item $ \Gamma   \vdash   \phi  \ \mathsf{prop} $, or
\item $\Gamma  \vdash  \gamma  \ottsym{:}  \phi$, or
\item $\Gamma  \vdash  \overline{\ottnt{q} }  \ottsym{:}  \overline{\chi}$, or
\item $\Gamma  \vdash  \ottmv{x}  \ottsym{:}  \tau$,
\end{enumerate}
then $ \vdash   \Gamma  \ \mathsf{ctx} $.
\end{lemma}

\begin{proof}
Straightforward mutual induction on typing judgments, appealing to
\pref{lem:ctx-form} in the cases that bind a new variable.
\end{proof}

\subsection{Substitution}

\begin{lemma}[Type substitution in $ \mathsf{no\_conflict} $]
\label{lem:ty-nc-subst}
If $\xi  \ottsym{:}  \overline{\ottnt{E} }  \in  \Sigma$ and $ \mathsf{no\_conflict} ( \overline{\ottnt{E} } ,  \ottmv{i} ,  \overline{\sigma} ,  \ottmv{k} ) $, then
$ \mathsf{no\_conflict} ( \overline{\ottnt{E} } ,  \ottmv{i} ,  \overline{\sigma}  \ottsym{[}  \tau_{{\mathrm{0}}}  \ottsym{/}  \alpha  \ottsym{]} ,  \ottmv{k} ) $.
\end{lemma}

\begin{proof}
By case analysis on $ \mathsf{no\_conflict} ( \overline{\ottnt{E} } ,  \ottmv{i} ,  \overline{\sigma} ,  \ottmv{k} ) $. We have two cases:
\begin{description}
\item[Case \rul{NC\_Apart}:]
\[
\ottdruleNCXXApart{}
\]
Inversion tells us $ \mathsf{apart} ( \overline{\tau}_{{\mathrm{2}}} ;\, \overline{\tau}_{{\mathrm{1}}}  \ottsym{[}  \overline{\sigma}  \ottsym{/}  \overline{\alpha}_{{\mathrm{1}}}  \ottsym{]} ) $.
We must show $ \mathsf{apart} ( \overline{\tau}_{{\mathrm{2}}} ;\, \overline{\tau}_{{\mathrm{1}}}  \ottsym{[}  \overline{\sigma}  \ottsym{[}  \tau_{{\mathrm{0}}}  \ottsym{/}  \alpha  \ottsym{]}  \ottsym{/}  \overline{\alpha}_{{\mathrm{1}}}  \ottsym{]} ) $.
From the assumption that equations are well-scoped (\pref{assn:declarations}
and \rul{D\_Axiom}), we see that $\alpha  \not\in  \mathit{fv}\! \, \ottsym{(}  \overline{\tau}_{{\mathrm{1}}}  \ottsym{)}$ and we can thus
rewrite as
$ \mathsf{apart} ( \overline{\tau}_{{\mathrm{2}}} ;\, \overline{\tau}_{{\mathrm{1}}}  \ottsym{[}  \overline{\sigma}  \ottsym{/}  \overline{\alpha}_{{\mathrm{1}}}  \ottsym{]}  \ottsym{[}  \tau_{{\mathrm{0}}}  \ottsym{/}  \alpha  \ottsym{]} ) $.
We have this by \pref{lem:apart-subst} and we are thus done.
\item[Case \rul{NC\_Compatible}]:
\[
\ottdruleNCXXCompatible{}
\]
The substitution has no effect on the premise and we are thus done.
\end{description}
\end{proof}

\begin{lemma}[Type substitution in types]
\label{lem:ty-ty-subst}
Assume $ \Gamma   \vdash   \tau  \ \mathsf{type} $.
\begin{enumerate}
\item If $ \Gamma  \ottsym{,}  \alpha  \ottsym{,}  \Gamma'   \vdash   \sigma  \ \mathsf{type} $, then
$ \Gamma  \ottsym{,}  \Gamma'  \ottsym{[}  \tau  \ottsym{/}  \alpha  \ottsym{]}   \vdash   \sigma  \ottsym{[}  \tau  \ottsym{/}  \alpha  \ottsym{]}  \ \mathsf{type} $.
\item If $ \Gamma  \ottsym{,}  \alpha  \ottsym{,}  \Gamma'   \vdash   \phi  \ \mathsf{prop} $, then $ \Gamma  \ottsym{,}  \Gamma'  \ottsym{[}  \tau  \ottsym{/}  \alpha  \ottsym{]}   \vdash   \phi  \ottsym{[}  \tau  \ottsym{/}  \alpha  \ottsym{]}  \ \mathsf{prop} $.
\item If $ \vdash   \Gamma  \ottsym{,}  \alpha  \ottsym{,}  \Gamma'  \ \mathsf{ctx} $, then $ \vdash   \Gamma  \ottsym{,}  \Gamma'  \ottsym{[}  \tau  \ottsym{/}  \alpha  \ottsym{]}  \ \mathsf{ctx} $.
\end{enumerate}
\end{lemma}

\begin{proof}
Straightforward mutual induction, with the usual reasoning in the variable
case, appealing to \pref{lem:scoping} to show that a variable bound in
$\Gamma$ cannot be affected by the substitution and to \pref{lem:weakening}
to extend the contexts.
\end{proof}

\begin{lemma}[Type substitution in coercions] ~
\label{lem:ty-co-subst}
Assume $ \Gamma   \vdash   \tau  \ \mathsf{type} $.
\begin{enumerate}
\item If $\Gamma  \ottsym{,}  \alpha  \ottsym{,}  \Gamma'  \vdash  \gamma  \ottsym{:}  \phi$, then $\Gamma  \ottsym{,}  \Gamma'  \ottsym{[}  \tau  \ottsym{/}  \alpha  \ottsym{]}  \vdash  \gamma  \ottsym{[}  \tau  \ottsym{/}  \alpha  \ottsym{]}  \ottsym{:}  \phi  \ottsym{[}  \tau  \ottsym{/}  \alpha  \ottsym{]}$.
\item If $\Gamma  \ottsym{,}  \alpha  \ottsym{,}  \Gamma'  \vdash  \overline{\ottnt{q} }  \ottsym{:}  \overline{\chi}$, then $\Gamma  \ottsym{,}  \Gamma'  \ottsym{[}  \tau  \ottsym{/}  \alpha  \ottsym{]}  \vdash  \overline{\ottnt{q} }  \ottsym{[}  \tau  \ottsym{/}  \alpha  \ottsym{]}  \ottsym{:}  \overline{\chi}  \ottsym{[}  \tau  \ottsym{/}  \alpha  \ottsym{]}$.
\end{enumerate}
\end{lemma}

\begin{proof}
By mutual induction, appealing to \pref{lem:ty-ty-subst},
\pref{lem:ty-nc-subst},
and
\pref{lem:scoping}.
\end{proof}

\begin{lemma}[Type substitution]
\label{lem:ty-subst}
If $\Gamma  \ottsym{,}  \alpha  \ottsym{,}  \Gamma'  \vdash  \ottnt{e}  \ottsym{:}  \sigma$ and $ \Gamma   \vdash   \tau  \ \mathsf{type} $, then
$\Gamma  \ottsym{,}  \Gamma'  \ottsym{[}  \tau  \ottsym{/}  \alpha  \ottsym{]}  \vdash  \ottnt{e}  \ottsym{[}  \tau  \ottsym{/}  \alpha  \ottsym{]}  \ottsym{:}  \sigma  \ottsym{[}  \tau  \ottsym{/}  \alpha  \ottsym{]}$.
\end{lemma}

\begin{proof}
By induction, appealing to \pref{lem:ty-ty-subst}, \pref{lem:ty-co-subst},
and \pref{lem:scoping}.
\end{proof}

\begin{lemma}[Coercion substitution in coercions]
\label{lem:co-co-subst}
Assume $\Gamma  \vdash  \gamma'  \ottsym{:}  \phi'$.
\begin{enumerate}
\item If $\Gamma  \ottsym{,}   \ottmv{c} {:} \phi'   \ottsym{,}  \Gamma'  \vdash  \gamma  \ottsym{:}  \phi$, then $\Gamma  \ottsym{,}  \Gamma'  \vdash  \gamma  \ottsym{[}  \gamma'  \ottsym{/}  \ottmv{c}  \ottsym{]}  \ottsym{:}  \phi$.
\item If $\Gamma  \ottsym{,}   \ottmv{c} {:} \phi'   \ottsym{,}  \Gamma'  \vdash  \overline{\ottnt{q} }  \ottsym{:}  \overline{\chi}$, then $\Gamma  \ottsym{,}  \Gamma'  \vdash  \overline{\ottnt{q} }  \ottsym{[}  \gamma'  \ottsym{/}  \ottmv{c}  \ottsym{]}  \ottsym{:}  \overline{\chi}$.
\end{enumerate}
\end{lemma}

\begin{proof}
By mutual induction, with the usual reasoning in the variable case, appealing to
\pref{lem:weakening} and \pref{lem:strengthening}.
\end{proof}

\begin{lemma}[Coercion substitution]
\label{lem:co-subst}
If $\Gamma  \ottsym{,}   \ottmv{c} {:} \phi   \ottsym{,}  \Gamma'  \vdash  \ottnt{e}  \ottsym{:}  \tau$ and $\Gamma  \vdash  \gamma  \ottsym{:}  \phi$, then $\Gamma  \ottsym{,}  \Gamma'  \vdash  \ottnt{e}  \ottsym{[}  \gamma  \ottsym{/}  \ottmv{c}  \ottsym{]}  \ottsym{:}  \tau$.
\end{lemma}

\begin{proof}
By induction, appealing to \pref{lem:co-co-subst}, \pref{lem:scoping},
and \pref{lem:strengthening}.
\end{proof}

\begin{lemma}[Resolution substitution]
\label{lem:res-subst}
If $\Gamma  \ottsym{,}  \alpha  \ottsym{,}   \ottmv{c} {:} \ottmv{F} \, \overline{\tau}  \sim  \alpha   \vdash  \ottnt{e}  \ottsym{:}  \sigma$ and $\Gamma  \vdash  \ottnt{q}  \ottsym{:}  \chi$, then
$\Gamma  \vdash  \ottnt{e}  \ottsym{[}  \ottnt{q}  \ottsym{/}  \chi  \ottsym{]}  \ottsym{:}  \sigma  \ottsym{[}  \ottnt{q}  \ottsym{/}  \chi  \ottsym{]}$.
\end{lemma}

\begin{proof}
Corollary of \pref{lem:ty-subst} and \pref{lem:co-subst}.
\end{proof}

\begin{lemma}[Substitution]
\label{lem:subst}
If $\Gamma  \ottsym{,}   \ottmv{x} {:} \tau'   \ottsym{,}  \Gamma'  \vdash  \ottnt{e}  \ottsym{:}  \tau$ and $\Gamma  \vdash  \ottnt{e'}  \ottsym{:}  \tau'$, then
$\Gamma  \ottsym{,}  \Gamma'  \vdash  \ottnt{e}  \ottsym{[}  \ottnt{e'}  \ottsym{/}  \ottmv{x}  \ottsym{]}  \ottsym{:}  \tau$.
\end{lemma}

\begin{proof}
By induction, with the usual reasoning in the variable case, appealing to
\pref{lem:weakening} and \pref{lem:strengthening}.
\end{proof}

\subsection{Regularity II}

\begin{lemma}[Context types]
\label{lem:ctx-type}
If $ \ottmv{x} {:} \tau   \in  \Gamma$ and $ \vdash   \Gamma  \ \mathsf{ctx} $, then
$ \Gamma   \vdash   \tau  \ \mathsf{type} $.
\end{lemma}

\begin{proof}
By induction on the structure of $\Gamma$, using
\pref{lem:weakening} at the end to fix the context in the
conclusion.
\end{proof}

\begin{lemma}[Context props]
\label{lem:ctx-prop}
If $ \ottmv{c} {:} \phi   \in  \Gamma$ and $ \vdash   \Gamma  \ \mathsf{ctx} $, then
$ \Gamma   \vdash   \phi  \ \mathsf{prop} $.
\end{lemma}

\begin{proof}
Similar to previous proof.
\end{proof}

\begin{assumption}[Constant types]
\label{assn:const-type}
If $\ottmv{K}  \ottsym{:}  \ottnt{H} \, \overline{\tau}$, then $ \emptyset   \vdash   \ottnt{H} \, \overline{\tau}  \ \mathsf{type} $.
\end{assumption}

\begin{lemma}[Type regularity]
\label{lem:ty-reg}
\typeregstatement
\end{lemma}

\begin{proof}
By induction on the derivation of $\Gamma  \vdash  \ottnt{e}  \ottsym{:}  \tau$,
appealing to \pref{lem:ctx-type}, \pref{assn:const-type},
\pref{lem:ty-ty-subst} (in the \rul{E\_TApp} case),
\pref{lem:ctx-prop} (in the \rul{E\_CLam} case),
and \pref{lem:strengthening} (in the \rul{E\_Assume} case).
Note that we need the $ \Gamma   \vdash   \tau_{{\mathrm{2}}}  \ \mathsf{type} $ premise in
the \rul{E\_Cast} case.
\end{proof}

\subsection{Preservation}

\begin{rpttheorem}{thm:preservation}{Preservation}
\presstatement
\end{rpttheorem}

\begin{remark}
Note that this theorem requires an \emph{empty} context, in contrast to many
statements of type preservation. This choice is necessary in order to support
the \rul{S\_Resolve} rule and its use of $\Downarrow$, available only in an empty
context. (See \ifaec \S6.4\else \pref{sec:cfc-totality}\fi of the paper.) If we wanted a statement of type preservation
that worked in non-empty contexts, we would need to make $\Downarrow$ partial
and then assert in the premise to \rul{S\_Resolve} that it succeeds.
\end{remark}

\begin{proof}
By induction on the derivation of $\emptyset  \vdash  \ottnt{e}  \ottsym{:}  \tau$.

\begin{description}
\item[Case \rul{E\_Var}:] Impossible.
\item[Case \rul{E\_Const}:] Impossible.
\item[Case \rul{E\_Lam}:] Impossible.
\item[Case \rul{E\_App}:] We have several cases, depending on how $\ottnt{e}$ has stepped:
\begin{description}
\item[Case \rul{S\_App}:] By induction.
\item[Case \rul{S\_Beta}:] By \pref{lem:subst}.
\item[Case \rul{S\_Push}:] We adopt the metavariable names from the rule:
\[
\ottdruleSXXPush{}
\]
Inversion tells us the following (for some $\tau'$, $\sigma$, and $\sigma'$):
\begin{itemize}
\item $\emptyset  \vdash  \gamma  \ottsym{:}  \ottsym{(}  \tau  \to  \tau'  \ottsym{)}  \sim  \ottsym{(}  \sigma  \to  \sigma'  \ottsym{)}$
\item $\emptyset  \vdash  \ottnt{v}  \ottsym{:}  \tau  \to  \tau'$
\item $\emptyset  \vdash  \ottnt{e}  \ottsym{:}  \sigma$
\item $\emptyset  \vdash  \ottsym{(}  \ottnt{v}  \triangleright  \gamma  \ottsym{)} \, \ottnt{e}  \ottsym{:}  \sigma'$
\item $ \emptyset   \vdash   \sigma  \ \mathsf{type} $
\item $ \emptyset   \vdash   \sigma'  \ \mathsf{type} $
\end{itemize}
\pref{lem:ty-reg} and inversion then gives us:
\begin{itemize}
\item $ \emptyset   \vdash   \tau  \ \mathsf{type} $
\item $ \emptyset   \vdash   \tau'  \ \mathsf{type} $
\end{itemize}
We now show that $\gamma_{{\mathrm{1}}}$ and $\gamma_{{\mathrm{2}}}$ are well typed:
\begin{description}
\item[$\emptyset  \vdash  \gamma_{{\mathrm{1}}}  \ottsym{:}  \sigma  \sim  \tau$:] By \rul{C\_NthArrow} and \rul{C\_Sym}.
\item[$\emptyset  \vdash  \gamma_{{\mathrm{2}}}  \ottsym{:}  \tau'  \sim  \sigma'$:] By \rul{C\_NthArrow}.
\end{description}
Thus:
\begin{itemize}
\item $\emptyset  \vdash  \ottnt{e}  \triangleright  \gamma_{{\mathrm{1}}}  \ottsym{:}  \tau$ (by \rul{E\_Cast})
\item $\emptyset  \vdash  \ottnt{v} \, \ottsym{(}  \ottnt{e}  \triangleright  \gamma_{{\mathrm{1}}}  \ottsym{)}  \ottsym{:}  \tau'$ (by \rul{E\_App})
\item $\emptyset  \vdash  \ottnt{v} \, \ottsym{(}  \ottnt{e}  \triangleright  \gamma_{{\mathrm{1}}}  \ottsym{)}  \triangleright  \gamma_{{\mathrm{2}}}  \ottsym{:}  \sigma'$ (by \rul{E\_Cast})
\end{itemize}
The final derivation is what we seek, and thus we are done.
\end{description}
\item[Case \rul{E\_TLam}:] Impossible.
\item[Case \rul{E\_TApp}:] We have several cases, depending on how $\ottnt{e}$ has stepped:
\begin{description}
\item[Case \rul{S\_TApp}:] By induction.
\item[Case \rul{S\_TBeta}:] By \pref{lem:ty-subst}.
\item[Case \rul{S\_TPush}:] We adopt the metavariable names from the rule:
\[
\ottdruleSXXTPush{}
\]
Inversion tells us the following (for some $\sigma$ and $\sigma'$):
\begin{itemize}
\item $\emptyset  \vdash  \gamma  \ottsym{:}  \ottsym{(}  \forall \, \alpha  \ottsym{.}  \sigma  \ottsym{)}  \sim  \ottsym{(}  \forall \, \alpha  \ottsym{.}  \sigma'  \ottsym{)}$
\item $\emptyset  \vdash  \ottnt{v}  \ottsym{:}  \forall \, \alpha  \ottsym{.}  \sigma$
\item $\emptyset  \vdash  \ottsym{(}  \ottnt{v}  \triangleright  \gamma  \ottsym{)} \, \tau  \ottsym{:}  \sigma'  \ottsym{[}  \tau  \ottsym{/}  \alpha  \ottsym{]}$
\item $ \emptyset   \vdash   \forall \, \alpha  \ottsym{.}  \sigma'  \ \mathsf{type} $
\item $ \alpha   \vdash   \sigma'  \ \mathsf{type} $
\item $ \emptyset   \vdash   \tau  \ \mathsf{type} $
\end{itemize}
We see that $\emptyset  \vdash  \gamma'  \ottsym{:}  \sigma  \ottsym{[}  \tau  \ottsym{/}  \alpha  \ottsym{]}  \sim  \sigma'  \ottsym{[}  \tau  \ottsym{/}  \alpha  \ottsym{]}$ by \rul{C\_Inst}.
Thus:
\begin{itemize}
\item $\emptyset  \vdash  \ottnt{v} \, \tau  \ottsym{:}  \sigma  \ottsym{[}  \tau  \ottsym{/}  \alpha  \ottsym{]}$ (by \rul{E\_TApp})
\item $ \emptyset   \vdash   \sigma'  \ottsym{[}  \tau  \ottsym{/}  \alpha  \ottsym{]}  \ \mathsf{type} $ (by \pref{lem:ty-ty-subst})
\item $\emptyset  \vdash  \ottnt{v} \, \tau  \triangleright  \gamma'  \ottsym{:}  \sigma'  \ottsym{[}  \tau  \ottsym{/}  \alpha  \ottsym{]}$
\end{itemize}
The final derivation is what we seek, and thus we are done.
\end{description}
\item[Case \rul{E\_CLam}:] Impossible.
\item[Case \rul{E\_CApp}:] We have several cases, depending on how $\ottnt{e}$ has stepped:
\begin{description}
\item[Case \rul{S\_CApp}:] By induction.
\item[Case \rul{S\_CBeta}:] By \pref{lem:co-subst}.
\item[Case \rul{S\_CPush}:] We adopt the metavariable names from the rule:
\[
\ottdruleSXXCPush{}
\]
Let $\phi \, \ottsym{=} \, \tau_{{\mathrm{0}}}  \sim  \tau_{{\mathrm{1}}}$. Inversion tells us the following (for some $\tau_{{\mathrm{2}}}$,
$\sigma_{{\mathrm{0}}}$, $\sigma_{{\mathrm{1}}}$, and $\sigma_{{\mathrm{2}}}$):
\begin{itemize}
\item $\emptyset  \vdash  \ottnt{v}  \ottsym{:}  \tau_{{\mathrm{0}}}  \sim  \tau_{{\mathrm{1}}}  \Rightarrow  \tau_{{\mathrm{2}}}$
\item $\emptyset  \vdash  \eta  \ottsym{:}  \ottsym{(}  \tau_{{\mathrm{0}}}  \sim  \tau_{{\mathrm{1}}}  \Rightarrow  \tau_{{\mathrm{2}}}  \ottsym{)}  \sim  \ottsym{(}  \sigma_{{\mathrm{0}}}  \sim  \sigma_{{\mathrm{1}}}  \Rightarrow  \sigma_{{\mathrm{2}}}  \ottsym{)}$
\item $\emptyset  \vdash  \gamma  \ottsym{:}  \sigma_{{\mathrm{0}}}  \sim  \sigma_{{\mathrm{1}}}$
\item $\emptyset  \vdash  \ottsym{(}  \ottnt{v}  \triangleright  \eta  \ottsym{)} \, \gamma  \ottsym{:}  \sigma_{{\mathrm{2}}}$
\item $ \emptyset   \vdash   \sigma_{{\mathrm{2}}}  \ \mathsf{type} $
\end{itemize}
We can now deduce:
\begin{itemize}
\item $\emptyset  \vdash  \eta_{{\mathrm{0}}}  \ottsym{:}  \tau_{{\mathrm{0}}}  \sim  \sigma_{{\mathrm{0}}}$ (by \rul{C\_NthQual})
\item $\emptyset  \vdash  \eta_{{\mathrm{1}}}  \ottsym{:}  \sigma_{{\mathrm{1}}}  \sim  \tau_{{\mathrm{1}}}$ (by \rul{C\_NthQual} and \rul{C\_Sym})
\item $\emptyset  \vdash  \eta_{{\mathrm{2}}}  \ottsym{:}  \tau_{{\mathrm{2}}}  \sim  \sigma_{{\mathrm{2}}}$ (by \rul{C\_NthQual})
\item $\emptyset  \vdash  \eta_{{\mathrm{0}}}  \fatsemi  \gamma  \fatsemi  \eta_{{\mathrm{1}}}  \ottsym{:}  \tau_{{\mathrm{0}}}  \sim  \tau_{{\mathrm{1}}}$ (by \rul{C\_Trans})
\item $\emptyset  \vdash  \ottnt{v} \, \ottsym{(}  \eta_{{\mathrm{0}}}  \fatsemi  \gamma  \fatsemi  \eta_{{\mathrm{1}}}  \ottsym{)}  \ottsym{:}  \tau_{{\mathrm{2}}}$ (by \rul{E\_CApp})
\item $\emptyset  \vdash  \ottnt{v} \, \ottsym{(}  \eta_{{\mathrm{0}}}  \fatsemi  \gamma  \fatsemi  \eta_{{\mathrm{1}}}  \ottsym{)}  \triangleright  \eta_{{\mathrm{2}}}  \ottsym{:}  \sigma_{{\mathrm{2}}}$ (by \rul{E\_Cast})
\end{itemize}
The final derivation is what we seek, and thus we are done.
\end{description}
\item[Case \rul{E\_Cast}:] We have two possibilities:
\begin{description}
\item[Case \rul{S\_Cast}:] By induction.
\item[Case \rul{S\_Trans}:]
We adopt the metavariable names from the rule:
\[
\ottdruleSXXTrans{}
\]
Inversion tells us the following (for some $\tau_{{\mathrm{1}}}$, $\tau_{{\mathrm{2}}}$, and $\tau_{{\mathrm{3}}}$):
\begin{itemize}
\item $\emptyset  \vdash  \ottnt{v}  \ottsym{:}  \tau_{{\mathrm{1}}}$
\item $\emptyset  \vdash  \gamma_{{\mathrm{1}}}  \ottsym{:}  \tau_{{\mathrm{1}}}  \sim  \tau_{{\mathrm{2}}}$
\item $\emptyset  \vdash  \ottnt{v}  \triangleright  \gamma_{{\mathrm{1}}}  \ottsym{:}  \tau_{{\mathrm{2}}}$
\item $\emptyset  \vdash  \gamma_{{\mathrm{2}}}  \ottsym{:}  \tau_{{\mathrm{2}}}  \sim  \tau_{{\mathrm{3}}}$
\item $\emptyset  \vdash  \ottsym{(}  \ottnt{v}  \triangleright  \gamma_{{\mathrm{1}}}  \ottsym{)}  \triangleright  \gamma_{{\mathrm{2}}}  \ottsym{:}  \tau_{{\mathrm{3}}}$
\item $ \emptyset   \vdash   \tau_{{\mathrm{3}}}  \ \mathsf{type} $
\end{itemize}
We can thus deduce:
\begin{itemize}
\item $\emptyset  \vdash  \gamma_{{\mathrm{1}}}  \fatsemi  \gamma_{{\mathrm{2}}}  \ottsym{:}  \tau_{{\mathrm{1}}}  \sim  \tau_{{\mathrm{3}}}$ (by \rul{C\_Trans})
\item $\emptyset  \vdash  \ottnt{v}  \triangleright  \ottsym{(}  \gamma_{{\mathrm{1}}}  \fatsemi  \gamma_{{\mathrm{2}}}  \ottsym{)}  \ottsym{:}  \tau_{{\mathrm{3}}}$ (by \rul{E\_Cast})
\end{itemize}
The final derivation is what we seek, and thus we are done.
\end{description}
\item[Case \rul{E\_Assume}:] We must step by \rul{S\_Resolve}.
\[
\ottdruleEXXAssume{}
\]
\[
\ottdruleSXXResolve{}
\]
We can assume the premises of \rul{E\_Assume}. We thus invoke
\pref{prop:total} to see that $\emptyset  \vdash  \ottnt{q}  \ottsym{:}  \chi$.
We are done by \pref{lem:res-subst}.
\end{description}
\end{proof}

\subsection{Consistency}

\ifaec
\ottdefnRedT{}
\ottdefnRed{}
\fi

\begin{lemma}[Top-level reduction]
\label{lem:top-level-red}
If $\ottmv{F} \, \overline{\tau}  \rightsquigarrow_{\top}  \sigma_{{\mathrm{1}}}$ and $\ottmv{F} \, \overline{\tau}  \rightsquigarrow_{\top}  \sigma_{{\mathrm{2}}}$, then $\sigma_{{\mathrm{1}}} \, \ottsym{=} \, \sigma_{{\mathrm{2}}}$.
\end{lemma}

\begin{proof}
We proceed by induction on the sum of the sizes of the derivations $\ottmv{F} \, \overline{\tau}  \rightsquigarrow_{\top}  \sigma_{{\mathrm{1}}}$
and $\ottmv{F} \, \overline{\tau}  \rightsquigarrow_{\top}  \sigma_{{\mathrm{2}}}$. We may thus assume the induction hypothesis for any
top-level reduction that appears in the premises of either $\ottmv{F} \, \overline{\tau}  \rightsquigarrow_{\top}  \sigma_{{\mathrm{1}}}$
or $\ottmv{F} \, \overline{\tau}  \rightsquigarrow_{\top}  \sigma_{{\mathrm{2}}}$.

From \pref{assn:good}, we see that every type family $\ottmv{F}$ is either open
with potentially multiple, single-equations axioms or closed with at most one,
potentially many-equationed axiom. We handle these cases separately:

\begin{description}
\item[Open family:]
Let $\xi_{{\mathrm{1}}}  \ottsym{:}  \ottnt{E_{{\mathrm{1}}}}  \in  \Sigma$ and $\xi_{{\mathrm{2}}}  \ottsym{:}  \ottnt{E_{{\mathrm{2}}}}  \in  \Sigma$ be the two axioms from
the two reductions, respectively, and let $\ottnt{E_{{\mathrm{1}}}} \, \ottsym{=} \, \forall \, \overline{\alpha}_{{\mathrm{1}}} \, \overline{\chi}_{{\mathrm{1}}}  \ottsym{.}  \ottmv{F} \, \overline{\sigma}_{{\mathrm{1}}}  \sim  \sigma_{{\mathrm{01}}}$
and $\ottnt{E_{{\mathrm{2}}}} \, \ottsym{=} \, \forall \, \overline{\alpha}_{{\mathrm{2}}} \, \overline{\chi}_{{\mathrm{2}}}  \ottsym{.}  \ottmv{F} \, \overline{\sigma}_{{\mathrm{2}}}  \sim  \sigma_{{\mathrm{02}}}$. Furthermore, we know
$\overline{\sigma}_{{\mathrm{1}}}  \ottsym{[}  \overline{\rho}_{{\mathrm{1}}}  \ottsym{/}  \overline{\alpha}_{{\mathrm{1}}}  \ottsym{]} = \overline{\tau} = \overline{\sigma}_{{\mathrm{2}}}  \ottsym{[}  \overline{\rho}_{{\mathrm{2}}}  \ottsym{/}  \overline{\alpha}_{{\mathrm{2}}}  \ottsym{]}$ (for the $\overline{\rho}_{{\mathrm{1}}}$ and $\overline{\rho}_{{\mathrm{2}}}$
learned by inversion). We know (by \pref{assn:good}), that $ \mathsf{compat} ( \ottnt{E_{{\mathrm{1}}}} ,  \ottnt{E_{{\mathrm{2}}}} ) $.
We now have two cases, depending on how $ \mathsf{compat} ( \ottnt{E_{{\mathrm{1}}}} ,  \ottnt{E_{{\mathrm{2}}}} ) $
has been established:
\begin{description}
\item[Case \rul{Co\_Coinc}:]
\[
\ottdruleCoXXCoinc{}
\]
Consider the substitution $\theta \, \ottsym{=} \, \overline{\rho}_{{\mathrm{1}}}  \ottsym{/}  \overline{\alpha}_{{\mathrm{1}}}  \ottsym{,}  \overline{\rho}_{{\mathrm{2}}}  \ottsym{/}  \overline{\alpha}_{{\mathrm{2}}}$. This is a unifier of
$\overline{\sigma}_{{\mathrm{1}}}$ and $\overline{\sigma}_{{\mathrm{2}}}$. Thus, by \pref{prop:unify-correct}, $ \mathsf{unify} ( \overline{\sigma}_{{\mathrm{1}}} ;\, \overline{\sigma}_{{\mathrm{2}}} )  \, \ottsym{=} \, \mathsf{Just} \, \theta_{{\mathrm{0}}}$
where $\theta \, \ottsym{=} \, \theta'  \circ  \theta_{{\mathrm{0}}}$. We see above that
$\sigma_{{\mathrm{01}}}  \ottsym{[}  \theta_{{\mathrm{0}}}  \circ  \mathit{subst}\! \, \ottsym{(}  \overline{\chi}_{{\mathrm{1}}}  \ottsym{)}  \ottsym{]} \, \ottsym{=} \, \sigma_{{\mathrm{02}}}  \ottsym{[}  \theta_{{\mathrm{0}}}  \circ  \mathit{subst}\! \, \ottsym{(}  \overline{\chi}_{{\mathrm{2}}}  \ottsym{)}  \ottsym{]}$
and thus $\sigma_{{\mathrm{01}}}  \ottsym{[}  \theta  \circ  \mathit{subst}\! \, \ottsym{(}  \overline{\chi}_{{\mathrm{1}}}  \ottsym{)}  \ottsym{]} \, \ottsym{=} \, \sigma_{{\mathrm{02}}}  \ottsym{[}  \theta  \circ  \mathit{subst}\! \, \ottsym{(}  \overline{\chi}_{{\mathrm{2}}}  \ottsym{)}  \ottsym{]}$.
Because the free variables in $\sigma_{{\mathrm{01}}}$ and $\sigma_{{\mathrm{02}}}$ are distinct, we can simplify to
$\sigma_{{\mathrm{01}}}  \ottsym{[}  \overline{\rho}_{{\mathrm{1}}}  \ottsym{/}  \overline{\alpha}_{{\mathrm{1}}}  \circ  \mathit{subst}\! \, \ottsym{(}  \overline{\chi}_{{\mathrm{1}}}  \ottsym{)}  \ottsym{]} \, \ottsym{=} \, \sigma_{{\mathrm{02}}}  \ottsym{[}  \overline{\rho}_{{\mathrm{2}}}  \ottsym{/}  \overline{\alpha}_{{\mathrm{2}}}  \circ  \mathit{subst}\! \, \ottsym{(}  \overline{\chi}_{{\mathrm{2}}}  \ottsym{)}  \ottsym{]}$. Call that type
$\rho_{{\mathrm{0}}}$.

\RAE{This next bit is a skosh hand-wavy. But I actually think it's \emph{more}
convincing than the more formal approach, which requires a delicate induction in
the presence of the |mapAccumL| that is in the premises to \rul{RTop}.}
Let's consider the shapes of $\rho_{{\mathrm{0}}}$ and $\sigma_{{\mathrm{1}}}$, one of the top-level reducts of
$\ottmv{F} \, \overline{\tau}$. The former is $\sigma_{{\mathrm{01}}}  \ottsym{[}  \overline{\rho}_{{\mathrm{1}}}  \ottsym{/}  \overline{\alpha}_{{\mathrm{1}}}  \circ  \mathit{subst}\! \, \ottsym{(}  \overline{\chi}_{{\mathrm{1}}}  \ottsym{)}  \ottsym{]}$ and the latter is
$\sigma_{{\mathrm{01}}}  \ottsym{[}  \overline{\rho}_{{\mathrm{1}}}  \ottsym{/}  \overline{\alpha}_{{\mathrm{1}}}  \ottsym{,}  \overline{\rho}'_{{\mathrm{1}}}  \ottsym{/}  \mathit{tvs}\! \, \ottsym{(}  \overline{\chi}_{{\mathrm{1}}}  \ottsym{)}  \ottsym{]}$. Thus, the only difference between $\rho_{{\mathrm{0}}}$ and
$\sigma_{{\mathrm{1}}}$ is the choice for the instantiation of the $\mathit{tvs}\! \, \ottsym{(}  \overline{\chi}_{{\mathrm{1}}}  \ottsym{)}$---$\rho_{{\mathrm{0}}}$ replaces
these with type family applications, and $\sigma_{{\mathrm{1}}}$ replaces them with proper (type-family-free)
types. However, note that in the premises to \rul{RTop}, the choice of these types (the
$\overline{\rho}'$ in the rule) is determined by the type family applications, using a $ \rightsquigarrow_{\top} $
reduction. These reductions are in a premise to a rule we are performing induction on,
and therefore we may assume that the type family application uniquely determines the reduct.
Thus there exists precisely one $\sigma_{{\mathrm{1}}}$ that corresponds to the $\rho_{{\mathrm{0}}}$, and
accordingly $\sigma_{{\mathrm{2}}}$
must be that same $\sigma_{{\mathrm{1}}}$. We are done with this case.

\item[Case \rul{Co\_Distinct}:]
In this case, there is no unifier between $\overline{\sigma}_{{\mathrm{1}}}$ and $\overline{\sigma}_{{\mathrm{2}}}$, a contradiction
with \pref{prop:unify-correct}.
\end{description}

\item[Closed family:]
In this case, we have the possibility that $\ottmv{F} \, \overline{\tau}$ is reducible by more than one
equation of the single applicable axiom $\xi$ with equations $\overline{\ottnt{E_{\ottmv{n}}} \, \ottsym{=} \, \forall \, \overline{\alpha}_{\ottmv{n}} \, \overline{\chi}_{\ottmv{n}}  \ottsym{.}  \ottmv{F} \, \overline{\sigma}_{\ottmv{n}}  \sim  \sigma_{{\mathrm{0}}\,\ottmv{n}}}$.
Number the two equations $\ottmv{i}$
and $\ottmv{j}$. If $\ottmv{i} \, \ottsym{=} \, \ottmv{j}$, we are done.\footnote{We still need to be sure that the
evaluation assumptions are satisfied by the same types when running the reduction twice,
but we can get this fact by a similar argument as given in the open type family case.}
We thus assume, without loss of generality, that $\ottmv{j} < \ottmv{i}$. We see as a premise
to \rul{RTop} that $ \mathsf{no\_conflict} ( \overline{\ottnt{E} } ,  \ottmv{i} ,  \overline{\rho} ,  \ottmv{j} ) $, where $\overline{\tau} \, \ottsym{=} \, \overline{\sigma}_{\ottmv{i}}  \ottsym{[}  \overline{\rho}  \ottsym{/}  \overline{\alpha}_{\ottmv{i}}  \ottsym{]}$. Thus,
either the two equations are apart (\rul{NC\_Apart}) or they are compatible. We'll
handle these cases separately:
\begin{description}
\item[Case \rul{NC\_Apart}:] In this case, we know that $ \mathsf{apart} ( \overline{\sigma}_{\ottmv{j}} ;\, \overline{\sigma}_{\ottmv{i}}  \ottsym{[}  \overline{\rho}  \ottsym{/}  \overline{\alpha}_{\ottmv{i}}  \ottsym{]} ) $---that
is, $ \mathsf{apart} ( \overline{\sigma}_{\ottmv{j}} ;\, \overline{\tau} ) $. By \pref{prop:unify-correct}, we can conclude that \rul{RTop}
cannot apply at equation $\ottmv{j}$, a contradiction.
\item[Case \rul{NC\_Compatible}:] Here, equations $\ottnt{E_{\ottmv{i}}}$ and $\ottnt{E_{\ottmv{j}}}$ are compatible;
follow the logic used in the open-type-family case.
\end{description}
\end{description}
\end{proof}

This is very closely based on the similar proof by \citet{EisenbergVPJW14}.

\begin{definition}[Type contexts] ~
\begin{enumerate}
\item Let $\mathcal{C}  \ottsym{[}   \cdot   \ottsym{]}$ be a type with exactly one hole.
\item Let $\mathcal{C}\!\!\!\mathcal{C}  \ottsym{[}   \cdot   \ottsym{]}$ be a list of types with exactly one hole (in the whole list).
\item Let $ \mathcal{C}  \llbracket   \cdot   \rrbracket $ be a type with any number of holes.
\item Let $ \mathcal{C}\!\!\!\mathcal{C}  \llbracket   \cdot   \rrbracket $ be a list of types with any number of holes.
\end{enumerate}
\end{definition}

\begin{lemma}[One step/many holes context substitution]
\label{lem:subst_1step_manyholes}
If $\tau  \rightsquigarrow  \tau'$, then $ \mathcal{C}  \llbracket  \tau  \rrbracket   \rightsquigarrow^*   \mathcal{C}  \llbracket  \tau'  \rrbracket $.
\end{lemma}

\begin{proof}
Straightforward induction on the structure of $ \mathcal{C}  \llbracket   \cdot   \rrbracket $.
\end{proof}

\begin{lemma}[Multistep/many holes context substitution]
\label{lem:subst_manysteps_manyholes}
If $\tau  \rightsquigarrow^*  \tau'$, then $ \mathcal{C}  \llbracket  \tau  \rrbracket   \rightsquigarrow^*   \mathcal{C}  \llbracket  \tau'  \rrbracket $.
\end{lemma}

\begin{proof}
Straightforward induction on the length of the reduction $\tau  \rightsquigarrow^*  \tau'$,
appealing to \pref{lem:subst_1step_manyholes}.
\end{proof}

\begin{rptlemma}{lem:local-confluence}{Local confluence}
\localconfstatement
\end{rptlemma}

\begin{proof}
Proceed by induction on the structure of $\tau_{{\mathrm{0}}}$.

\begin{description}
\item[Case $\tau_{{\mathrm{0}}} \, \ottsym{=} \, \sigma_{{\mathrm{1}}}  \to  \sigma_{{\mathrm{2}}}$:] Inversion on \rul{Red} tells
us $\mathcal{C}_{{\mathrm{1}}}  \ottsym{[}  \ottmv{F_{{\mathrm{1}}}} \, \overline{\rho}_{{\mathrm{1}}}  \ottsym{]} = \sigma_{{\mathrm{1}}}  \to  \sigma_{{\mathrm{2}}} = \mathcal{C}_{{\mathrm{2}}}  \ottsym{[}  \ottmv{F_{{\mathrm{2}}}} \, \overline{\rho}_{{\mathrm{2}}}  \ottsym{]}$, where
$\tau_{{\mathrm{1}}} \, \ottsym{=} \, \mathcal{C}_{{\mathrm{1}}}  \ottsym{[}  \rho'_{{\mathrm{1}}}  \ottsym{]}$ and $\tau_{{\mathrm{2}}} \, \ottsym{=} \, \mathcal{C}_{{\mathrm{2}}}  \ottsym{[}  \rho'_{{\mathrm{2}}}  \ottsym{]}$. Proceed by case analysis
on $\mathcal{C}_{{\mathrm{1}}}  \ottsym{[}   \cdot   \ottsym{]}$ and $\mathcal{C}_{{\mathrm{2}}}  \ottsym{[}   \cdot   \ottsym{]}$:
\begin{description}
\item[Case $\mathcal{C}_{{\mathrm{1}}}  \ottsym{[}   \cdot   \ottsym{]} \, \ottsym{=} \, \mathcal{C}'_{{\mathrm{1}}}  \ottsym{[}   \cdot   \ottsym{]}  \to  \sigma_{{\mathrm{2}}}$, $\mathcal{C}_{{\mathrm{2}}}  \ottsym{[}   \cdot   \ottsym{]} \, \ottsym{=} \, \mathcal{C}'_{{\mathrm{2}}}  \ottsym{[}   \cdot   \ottsym{]}  \to  \sigma_{{\mathrm{2}}}$:]
(We know the right-hand types to $( \to )$ must match from
$\mathcal{C}_{{\mathrm{1}}}  \ottsym{[}  \ottmv{F_{{\mathrm{1}}}} \, \overline{\rho}_{{\mathrm{1}}}  \ottsym{]} \, \ottsym{=} \, \mathcal{C}_{{\mathrm{2}}}  \ottsym{[}  \ottmv{F_{{\mathrm{2}}}} \, \overline{\rho}_{{\mathrm{2}}}  \ottsym{]}$.)
Let $\sigma_{{\mathrm{3}}} \, \ottsym{=} \, \mathcal{C}'_{{\mathrm{1}}}  \ottsym{[}  \rho'_{{\mathrm{1}}}  \ottsym{]}$ and $\sigma_{{\mathrm{4}}} \, \ottsym{=} \, \mathcal{C}'_{{\mathrm{2}}}  \ottsym{[}  \rho'_{{\mathrm{2}}}  \ottsym{]}$. Noting that
$\mathcal{C}$ appears only in the conclusion of \rul{Red}, and not in any
premise, we can see that $\sigma_{{\mathrm{1}}}  \rightsquigarrow  \sigma_{{\mathrm{3}}}$ and $\sigma_{{\mathrm{1}}}  \rightsquigarrow  \sigma_{{\mathrm{4}}}$.
The induction hypothesis thus gives us $\sigma_{{\mathrm{5}}}$ such that
$\sigma_{{\mathrm{3}}}  \rightsquigarrow^*  \sigma_{{\mathrm{5}}}  \leftsquigarrow^*  \sigma_{{\mathrm{4}}}$. We then say that $\tau_{{\mathrm{3}}}$, our common reduct,
is $\sigma_{{\mathrm{5}}}  \to  \sigma_{{\mathrm{2}}}$, appealing to \pref{lem:subst_manysteps_manyholes}.
\item[Case $\mathcal{C}_{{\mathrm{1}}}  \ottsym{[}   \cdot   \ottsym{]} \, \ottsym{=} \, \mathcal{C}'_{{\mathrm{1}}}  \ottsym{[}   \cdot   \ottsym{]}  \to  \sigma_{{\mathrm{2}}}$, $\mathcal{C}_{{\mathrm{2}}}  \ottsym{[}   \cdot   \ottsym{]} \, \ottsym{=} \, \sigma_{{\mathrm{1}}}  \to  \mathcal{C}'_{{\mathrm{2}}}  \ottsym{[}   \cdot   \ottsym{]}$:]
Let $\tau_{{\mathrm{3}}} \, \ottsym{=} \, \mathcal{C}'_{{\mathrm{1}}}  \ottsym{[}  \rho'_{{\mathrm{1}}}  \ottsym{]}  \to  \mathcal{C}'_{{\mathrm{2}}}  \ottsym{[}  \rho'_{{\mathrm{2}}}  \ottsym{]}$. Because we have $\sigma_{{\mathrm{1}}} \, \ottsym{=} \, \mathcal{C}'_{{\mathrm{1}}}  \ottsym{[}  \ottmv{F} \, \overline{\rho}  \ottsym{]}$
and $\sigma_{{\mathrm{2}}} \, \ottsym{=} \, \mathcal{C}'_{{\mathrm{2}}}  \ottsym{[}  \ottmv{F} \, \overline{\rho}  \ottsym{]}$, we can see that $\ottsym{(}  \mathcal{C}'_{{\mathrm{1}}}  \ottsym{[}  \rho'_{{\mathrm{1}}}  \ottsym{]}  \to  \sigma_{{\mathrm{2}}}  \ottsym{)}  \rightsquigarrow  \tau_{{\mathrm{3}}}  \leftsquigarrow  \ottsym{(}  \sigma_{{\mathrm{1}}}  \to  \mathcal{C}'_{{\mathrm{2}}}  \ottsym{[}  \rho'_{{\mathrm{2}}}  \ottsym{]}  \ottsym{)}$
as desired.
\item[Other cases:] Similar to the two previous cases.
\end{description}
\item[Case $\tau_{{\mathrm{0}}} \, \ottsym{=} \, \ottnt{H} \, \overline{\sigma}$:] Similar to previous case.
\item[Case $\tau_{{\mathrm{0}}} \, \ottsym{=} \, \alpha$:] Impossible.
\item[Case $\tau_{{\mathrm{0}}} \, \ottsym{=} \, \forall \, \alpha  \ottsym{.}  \sigma$:] Similar to first sub-case of the
$\tau_{{\mathrm{0}}} \, \ottsym{=} \, \sigma_{{\mathrm{1}}}  \to  \sigma_{{\mathrm{2}}}$ case.
\item[Case $\tau_{{\mathrm{0}}} \, \ottsym{=} \, \ottsym{(}  \sigma_{{\mathrm{1}}}  \sim  \sigma_{{\mathrm{2}}}  \ottsym{)}  \Rightarrow  \sigma_{{\mathrm{3}}}$:] Similar to $\tau_{{\mathrm{0}}} \, \ottsym{=} \, \sigma_{{\mathrm{1}}}  \to  \sigma_{{\mathrm{2}}}$ case.
\item[Case $\tau_{{\mathrm{0}}} \, \ottsym{=} \, \ottmv{F} \, \overline{\sigma}$:]
We have $\mathcal{C}_{{\mathrm{1}}}  \ottsym{[}  \ottmv{F_{{\mathrm{1}}}} \, \overline{\rho}_{{\mathrm{1}}}  \ottsym{]} = \ottmv{F} \, \overline{\sigma} = \mathcal{C}_{{\mathrm{2}}}  \ottsym{[}  \ottmv{F_{{\mathrm{2}}}} \, \overline{\rho}_{{\mathrm{2}}}  \ottsym{]}$, where $\tau_{{\mathrm{1}}} \, \ottsym{=} \, \mathcal{C}_{{\mathrm{1}}}  \ottsym{[}  \rho'_{{\mathrm{1}}}  \ottsym{]}$ and $\tau_{{\mathrm{2}}} \, \ottsym{=} \, \mathcal{C}_{{\mathrm{2}}}  \ottsym{[}  \rho'_{{\mathrm{2}}}  \ottsym{]}$.
We have several cases,
depending on the structure of $\mathcal{C}_{{\mathrm{1}}}  \ottsym{[}   \cdot   \ottsym{]}$ and $\mathcal{C}_{{\mathrm{2}}}  \ottsym{[}   \cdot   \ottsym{]}$:
\begin{description}
\item[Case $\mathcal{C}_{{\mathrm{1}}}  \ottsym{[}   \cdot   \ottsym{]} \, \neq \,  \cdot $, $\mathcal{C}_{{\mathrm{2}}}  \ottsym{[}   \cdot   \ottsym{]} \, \neq \,  \cdot $:] The top-level function
$\ottmv{F}$ is not involved in the reductions. Proceed similarly to the $\tau_{{\mathrm{0}}} \, \ottsym{=} \, \sigma_{{\mathrm{1}}}  \to  \sigma_{{\mathrm{2}}}$ case.
\item[Case $\mathcal{C}_{{\mathrm{1}}}  \ottsym{[}   \cdot   \ottsym{]} \, \ottsym{=} \, \ottmv{F} \, \mathcal{C}\!\!\!\mathcal{C}'_{{\mathrm{1}}}  \ottsym{[}   \cdot   \ottsym{]}$, $\mathcal{C}_{{\mathrm{2}}}  \ottsym{[}   \cdot   \ottsym{]} \, \ottsym{=} \,  \cdot $:]
This case cannot happen. Examine the \rul{Red} rule:
\[
\ottdruleRed{}
\]
The $\overline{\tau}$---that is, the arguments to the function $\ottmv{F}$---equal $\overline{\sigma}  \ottsym{[}  \overline{\rho}  \ottsym{/}  \overline{\alpha}  \ottsym{]}$. But
$\overline{\sigma}$ are type-family free (by \pref{assn:declarations} and \rul{D\_Axiom}) and the $\overline{\rho}$
are type-family free (by $\ottcomp{ \emptyset   \vdash   \rho_{\ottmv{k}}  \ \mathsf{type} }{\ottmv{k}}$). Thus, $\overline{\tau}$ must also
be type-family free, meaning that there is no way $\tau_{{\mathrm{0}}}$ could have stepped to $\tau_{{\mathrm{1}}}$.
\item[Case $\mathcal{C}_{{\mathrm{1}}}  \ottsym{[}   \cdot   \ottsym{]} \, \ottsym{=} \,  \cdot $, $\mathcal{C}_{{\mathrm{2}}}  \ottsym{[}   \cdot   \ottsym{]} \, \ottsym{=} \, \ottmv{F} \, \mathcal{C}\!\!\!\mathcal{C}'_{{\mathrm{2}}}  \ottsym{[}   \cdot   \ottsym{]}$:] Similar to previous case.
\item[Case $\mathcal{C}_{{\mathrm{1}}}  \ottsym{[}   \cdot   \ottsym{]} \, \ottsym{=} \,  \cdot $, $\mathcal{C}_{{\mathrm{2}}}  \ottsym{[}   \cdot   \ottsym{]} \, \ottsym{=} \,  \cdot $:]
In this case, we have two top-level reductions. \pref{lem:top-level-red} gives us
something stronger than what we seek: that $\tau_{{\mathrm{1}}} \, \ottsym{=} \, \tau_{{\mathrm{2}}}$ in this case. We are done.
\end{description}
\end{description}
\end{proof}

\begin{rptlemma}{lem:termination}{Termination}
\termstatement
\end{rptlemma}

\begin{proof}
Examine \rul{RTop}:
\[
\ottdruleRTop{}
\]
We see that the reduct type $\tau'$ equals $\sigma_{{\mathrm{0}}}  \ottsym{[}  \overline{\rho}  \ottsym{/}  \overline{\alpha}  \ottsym{,}  \overline{\rho}'  \ottsym{/}  \mathit{tvs}\! \, \ottsym{(}  \overline{\chi}  \ottsym{)}  \ottsym{]}$.
By \pref{assn:declarations}, $\sigma_{{\mathrm{0}}}$ has no type families.
By $\ottcomp{ \emptyset   \vdash   \rho_{\ottmv{k}}  \ \mathsf{type} }{\ottmv{k}}$, we see that $\overline{\rho}$ has no type
families. By $\ottcomp{ \emptyset   \vdash   \rho'_{\ottmv{n}}  \ \mathsf{type} }{\ottmv{n}}$, we see that $\overline{\rho}'$ has
no type families. Thus $\tau'$ can mention no type families.

Yet, in \rul{Red},
\[
\ottdruleRed{}
\]
we replace a type family application $\ottmv{F} \, \overline{\tau}$ with $\tau'$. We have thus
replaced a type family application with a type that contains no type families.
Accordingly, in every use of \rul{Red}, the reduct must have exactly one
fewer type family application than the redex. Given that types are an inductively
defined structure, we cannot have an infinite number of type family applications
in a type.

We have thus identified a decreasing measure: the number of type family applications
in a type. Accordingly, all reduction chains terminate.
\end{proof}

\begin{rptlemma}{lem:confluence}{Confluence}
\confstatement
\end{rptlemma}

\begin{proof}
By \pref{lem:local-confluence}, \pref{lem:termination}, and Newman's Lemma.
\end{proof}

\begin{lemma}[Rigid reduction]
\label{lem:rigid-red}
If $\mathcal{C}  \ottsym{[}  \tau_{{\mathrm{1}}}  \ottsym{]}  \rightsquigarrow^*  \tau_{{\mathrm{0}}}  \leftsquigarrow^*  \mathcal{C}  \ottsym{[}  \tau_{{\mathrm{2}}}  \ottsym{]}$, such that the path to the hole in $\mathcal{C}  \ottsym{[}   \cdot   \ottsym{]}$ does not
go through any type family arguments, then there exists $\tau_{{\mathrm{3}}}$ such that $\tau_{{\mathrm{1}}}  \rightsquigarrow^*  \tau_{{\mathrm{3}}}  \leftsquigarrow^*  \tau_{{\mathrm{2}}}$.
\end{lemma}

\begin{proof}
Straightforward induction on the structure of $\mathcal{C}  \ottsym{[}   \cdot   \ottsym{]}$.
\end{proof}

\begin{rptlemma}{lem:complete-red}{Completeness of the rewrite relation}
\completestatement
\end{rptlemma}

\begin{proof}
By induction on the derivation of $\emptyset  \vdash  \gamma  \ottsym{:}  \tau_{{\mathrm{1}}}  \sim  \tau_{{\mathrm{2}}}$.

\begin{description}
\item[Case \rul{C\_Refl}:] Trivial.
\item[Case \rul{C\_Sym}:] By induction hypothesis.
\item[Case \rul{C\_Trans}:] By \pref{lem:confluence}.
\item[Case \rul{C\_App}:] By \pref{lem:subst_manysteps_manyholes}.
\item[Case \rul{C\_Fun}:] By \pref{lem:subst_manysteps_manyholes}.
\item[Case \rul{C\_Fam}:] By \pref{lem:subst_manysteps_manyholes}.
\item[Case \rul{C\_Forall}:] By \pref{lem:subst_manysteps_manyholes}.
\item[Case \rul{C\_Qual}:] By \pref{lem:subst_manysteps_manyholes}.
\item[Case \rul{C\_Nth}:] By \pref{lem:rigid-red}.
\item[Case \rul{C\_NthArrow}:] By \pref{lem:rigid-red}.
\item[Case \rul{C\_NthQual}:] By \pref{lem:rigid-red}.
\item[Case \rul{C\_Inst}:] By \pref{lem:rigid-red} and substitution.
\item[Case \rul{C\_Axiom}:]
\[
\ottdruleCXXAxiom{}
\]
The left-hand type steps to the right-hand type by \rul{Red}.
\item[Case \rul{C\_Var}:] Impossible in an empty context.
\end{description}
\end{proof}

\begin{rptlemma}{lem:type-no-red}{Proper types do not reduce}
\noredstatement
\end{rptlemma}

\begin{proof}
Direct from the definition of $ \rightsquigarrow $, \rul{Red}.
\end{proof}

\begin{rptlemma}{lem:consistency}{Consistency}
\consstatement
\end{rptlemma}

\begin{proof}
\pref{lem:complete-red} tells us that there exists $\tau_{{\mathrm{3}}}$ such that $\tau_{{\mathrm{1}}}  \rightsquigarrow^*  \tau_{{\mathrm{3}}}  \leftsquigarrow^*  \tau_{{\mathrm{2}}}$.
But \pref{lem:type-no-red} tells us that $\tau_{{\mathrm{1}}}$ and $\tau_{{\mathrm{2}}}$ do not reduce. We must conclude
that $\tau_{{\mathrm{1}}} \, \ottsym{=} \, \tau_{{\mathrm{3}}} = \tau_{{\mathrm{2}}}$, and we are done.
\end{proof}

\subsection{Progress}

\begin{lemma}[Canonical forms] \leavevmode
\label{lem:canon-form}
\begin{enumerate}
\item If $\emptyset  \vdash  \ottnt{v}  \ottsym{:}  \tau_{{\mathrm{1}}}  \to  \tau_{{\mathrm{2}}}$, then $\ottnt{v} \, \ottsym{=} \, \lambda  \ottmv{x}  \ottsym{:}  \tau_{{\mathrm{1}}}  \ottsym{.}  \ottnt{e}$ for some $\ottnt{e}$.
\item If $\emptyset  \vdash  \ottnt{v}  \ottsym{:}  \forall \, \alpha  \ottsym{.}  \tau$, then $\ottnt{v} \, \ottsym{=} \, \Lambda  \alpha  \ottsym{.}  \ottnt{e}$ for some $\ottnt{e}$.
\item If $\emptyset  \vdash  \ottnt{v}  \ottsym{:}  \phi  \Rightarrow  \tau$, then $\ottnt{v} \, \ottsym{=} \, \lambda  \ottmv{c}  \ottsym{:}  \phi  \ottsym{.}  \ottnt{e}$ for some $\ottnt{e}$.
\end{enumerate}
\end{lemma}

\begin{proof}
By case analysis on the typing derivation.
\end{proof}

\begin{rpttheorem}{thm:progress}{Progress}
\progstatement
\end{rpttheorem}

\begin{proof}
By induction on the derivation of $\emptyset  \vdash  \ottnt{e}  \ottsym{:}  \tau$.

\begin{description}
\item[Case \rul{E\_Var}:] Impossible.
\item[Case \rul{E\_Const}:] Trivial.
\item[Case \rul{E\_Lam}:] Trivial.
\item[Case \rul{E\_App}:] We know that $\ottnt{e} \, \ottsym{=} \, \ottnt{e_{{\mathrm{1}}}} \, \ottnt{e_{{\mathrm{2}}}}$. A use of the induction
hypothesis on $\ottnt{e_{{\mathrm{1}}}}$ gives us three possibilities:
\begin{description}
\item[Case $\ottnt{e_{{\mathrm{1}}}} \, \ottsym{=} \, \ottnt{v_{{\mathrm{1}}}}$:] \pref{lem:canon-form} tells us that
$\ottnt{e_{{\mathrm{1}}}} \, \ottsym{=} \, \lambda  \ottmv{x}  \ottsym{:}  \tau_{{\mathrm{0}}}  \ottsym{.}  \ottnt{e_{{\mathrm{0}}}}$ and we are thus done by \rul{E\_Beta}.
\item[Case $\ottnt{e_{{\mathrm{1}}}} \, \ottsym{=} \, \ottnt{v_{{\mathrm{1}}}}  \triangleright  \gamma$:] Inversion tells us that
$\emptyset  \vdash  \gamma  \ottsym{:}  \tau_{{\mathrm{0}}}  \sim  \ottsym{(}  \sigma_{{\mathrm{1}}}  \to  \sigma_{{\mathrm{2}}}  \ottsym{)}$, with
$ \Gamma   \vdash   \ottsym{(}  \sigma_{{\mathrm{1}}}  \to  \sigma_{{\mathrm{2}}}  \ottsym{)}  \ \mathsf{type} $.
\pref{lem:ty-reg} (on $\emptyset  \vdash  \ottnt{v_{{\mathrm{1}}}}  \ottsym{:}  \tau_{{\mathrm{0}}}$)
tells us that $ \Gamma   \vdash   \tau_{{\mathrm{0}}}  \ \mathsf{type} $.
\pref{lem:consistency}
then proves that $\tau_{{\mathrm{0}}} \, \ottsym{=} \, \sigma_{{\mathrm{1}}}  \to  \sigma_{{\mathrm{2}}}$.
We can thus use \pref{lem:canon-form} to get $\ottnt{e_{{\mathrm{1}}}} \, \ottsym{=} \, \lambda  \ottmv{x}  \ottsym{:}  \sigma_{{\mathrm{1}}}  \ottsym{.}  \ottnt{e_{{\mathrm{0}}}}$
and we can step by \rul{S\_Push}.
\item[Case $\ottnt{e_{{\mathrm{1}}}}  \longrightarrow  \ottnt{e'_{{\mathrm{1}}}}$:] We are done by \rul{S\_App}.
\end{description}
\item[Case \rul{E\_TLam}:] Trivial.
\item[Case \rul{E\_TApp}:] Similar to \rul{E\_App} case, using
\rul{S\_TBeta} and \rul{S\_TPush}.
\item[Case \rul{E\_CLam}:] Trivial.
\item[Case \rul{E\_CApp}:] Similar to \rul{E\_App} case, using
\rul{S\_CBeta} and \rul{S\_CPush}.
\item[Case \rul{E\_Cast}:] We know that $\ottnt{e} \, \ottsym{=} \, \ottnt{e_{{\mathrm{0}}}}  \triangleright  \gamma$. A use of the induction
hypothesis on $\ottnt{e_{{\mathrm{0}}}}$ gives us three possibilities:
\begin{description}
\item[Case $\ottnt{e_{{\mathrm{0}}}} \, \ottsym{=} \, \ottnt{v_{{\mathrm{0}}}}$:] Trivial.
\item[Case $\ottnt{e_{{\mathrm{0}}}} \, \ottsym{=} \, \ottnt{v_{{\mathrm{0}}}}  \triangleright  \gamma_{{\mathrm{0}}}$:] We are done by \rul{S\_Trans}.
\item[Case $\ottnt{e_{{\mathrm{0}}}}  \longrightarrow  \ottnt{e'_{{\mathrm{0}}}}$:] We are done by \rul{S\_Cast}.
\end{description}
\item[Case \rul{E\_Assume}:]
We are done by \rul{S\_Resolve}, noting that the assumed derivation must exist by
\pref{prop:total}.
\end{description}
\end{proof}

\fi

\end{document}